\documentclass{article}   

\usepackage{arxiv}

\usepackage[english]{babel}
\usepackage[T1]{fontenc}

\usepackage{color,soul}
\usepackage{url}
\usepackage{amsmath}
\usepackage{amsthm}
\usepackage{multirow}
\usepackage{amsfonts}
\usepackage{graphicx}
\usepackage{caption}
\usepackage{subfloat}
\usepackage{subcaption}
\usepackage{xspace}
\usepackage{scalerel}
\usepackage{hyperref}
\usepackage{background}

\backgroundsetup{
  scale=1,
  angle=0,
  opacity=1,
  color=black,
  position=current page.north,
  vshift=-1.5cm,
  contents={%
    \begin{minipage}{0.9\paperwidth}
      \raggedright
      \color{red} % Change the color as needed
      \fontsize{9}{9}\selectfont % Adjust font size as needed
      {This is the author's version of the work. It is posted here for your personal use. Not for redistribution. The definitive version was published in {ACM Transactions on Knowledge Discovery from Data}, \href{https://doi.org/10.1145/3647642}{https://doi.org/10.1145/3647642}.}
    \end{minipage}
  }
}
\let\lctau\tau % save the lowercase of '\tau'

\usepackage[colorinlistoftodos]{todonotes}
\usepackage[ruled,linesnumbered]{algorithm2e}

\newcommand{\bb}[1]{ \textbf{#1}} %\boldsymbol
\newcommand{\R}[1]{\mathbb{R}^{#1}}
\newcommand{\T}{\scalerel*{\lctau}{X}}%\mathcal{T}
\newcommand{\F}{f}%\mathcal{F}
\newcommand{\nsimp}{\emph{nSimplex}\xspace}
\newcommand{\zen}{\emph{Zen}\xspace}
\newcommand{\lwb}{\emph{Lwb}\xspace}
\newcommand{\upb}{\emph{Upb}\xspace}

\newtheorem{theorem}{Theorem}[section]

\newtheorem{lemma}[theorem]{Lemma}
		
\title{\nsimp \zen: A Novel Dimensionality Reduction for Euclidean and Hilbert Spaces}

\author{Richard Connor\\ %\orcid{0000-0003-4734-8103}
University of St Andrews\\
North Haugh, St Andrews,\\
Fife, Scotland, UK\\
\texttt{rchc@st-andrews.ac.uk} \And
Lucia Vadicamo\\ %\orcid{0000-0001-7182-7038}
ISTI-CNR\\
Via G. Moruzzi 1, \\
Pisa, Italy\\
\texttt{lucia.vadicamo@isti.cnr.it} 
}

\begin{document}
\maketitle

\begin{abstract}
Dimensionality reduction techniques map values from a high dimensional space to one with a lower dimension. The result is  a space which requires less physical memory and has a faster distance calculation. These techniques are widely used where required properties of  the reduced-dimension space give an acceptable accuracy with respect to the original space.

Many such transforms have been described. They have been classified in two main groups: \emph{linear} and \emph{topological}. Linear methods such as  Principal  Component Analysis (PCA) and Random Projection (RP)  define matrix-based transforms into a lower dimension of Euclidean space. Topological methods such as Multidimensional Scaling (MDS) attempt to preserve higher-level aspects such as the nearest-neighbour relation, and some may be applied to non-Euclidean  spaces.

Here, we introduce \nsimp \zen, a novel topological method of reducing dimensionality. Like MDS, it relies only upon  pairwise distances measured in the original space. The use of distances, rather than coordinates, allows the technique to be applied to both Euclidean and other Hilbert spaces, including those governed by Cosine, Jensen-Shannon and Quadratic Form distances.

We show that in almost all cases, due to geometric properties of high-dimensional spaces, our new technique gives better properties than others, especially with reduction to very low dimensions.
\end{abstract}


\section{Introduction}

The requirement to work with large, high-dimensional metric spaces is a long-standing and increasingly important requirement across many domains of computation. Typically, each element of such a space represents some real-world artefact, and the distance between elements gives a proxy (dis-)similarity function over the real-world domain.

As technology progresses, both the dimension of spaces and the size of collections tend to increase. These factors increasingly imply that apparently simple calculations, for example to find a few of the most similar elements within a large set, may become intractable. While sub-linear search times may in some cases be achieved using metric indexing techniques \cite{zezula2006similarity, Chavez2001similarity}, these also become ineffective in high dimensions and it has been shown that even approximate search complexity often degrades to linear as dimensionality increases \cite{rubenstein_hardness}.

In the domain of image search for example, a collection of a million images would no longer be considered large, and using modern techniques the dimension of the space used to represent them is likely to be several thousand. In this case an exhaustive search to find the  image in the collection which is most similar to another  will require upward  of $10^{9}$ numeric comparisons, with $10$GBytes of data passing through main memory. If we consider that the YFCC benchmark image set \cite{thomee2016yfcc100m} contains $10^{10}$ images, and the size of a representation derived from GoogleNet \cite{Szegedy2015_CVPR} has around $200,000$ dimensions, clearly any such computation is well beyond the scope of most computational contexts. 

The purpose of dimensionality reduction is to reduce the dimension of the space, maintaining as far as possible the relative distances. This not only implies that less space is required in memory for the dataset, but also gives a faster distance computation.

There are a number of well-known approaches to dimensionality reduction. The Johnson-Lindenstrauss lemma shows a surprisingly small  lower bound on the degree of distortion which is necessarily introduced when reducing a high-dimensional Euclidean space to a lower dimension. This bound may be achievable by using a Random Projection (RP) into the lower dimension, and if the domain is perfectly uniformly distributed then this may be the best technique that is achievable. 

However much data is implicitly non-uniform, for example data deriving from a Convolutional Neural Network (CNN) typically lies on some complex manifold within the representational space \cite{7974879}. In these cases non-random transforms, which  take advantage of  non-uniformities  within the data, can achieve better results. 

Dimensionality reduction mechanisms are often classified in two main groups. Linear transforms, such as Principal Component Analysis (PCA) \cite{pearson1901:PCA,hotelling1933analysis}, apply only to Euclidean spaces and perform a matrix-based transform derived from properties of the whole space, with the primary goal of  preserving  pairwise distances as far as possible. Topological transforms, such as Multidimensional Scaling (MDS) \cite{cox2008multidimensional}, use the individual pairwise distances of the original space in an attempt to preserve higher-level relationships, such as the pairwise ordering of distances or the nearest-neighbour relation. Topological methods may be applicable to non-vector spaces, or even non-metric spaces with an appropriate dissimilarity function.  

%{\color{blue} I think something needs to be said about applying DR to large spaces, where the cost of the pure analytic function is unscalable and a set of representatives is used to create a generic transform which is then applied to other data within the set. Also, should rule out trivial transforms to very low dimensions as used to visualise clusters etc... }

In this article we introduce a new  mechanism, which has interestingly different properties, and which can be applied to Euclidean and, more generally, to any space that is isometrically embeddable in Hilbert space. It is primarily a distance-preservation mechanism, and uses properties of high-dimensional Euclidean geometry which have not  been previously applied in this domain to preserve distances in the reduced-dimension space.

% For very large data, transformation of the whole set in a single operation is typically infeasible. Instead, the transformation of the data requires to be given as a function from individual data items to their projected form. To illustrate the point, principal component analysis can analyse an $n x m$ matrix and produce a coefficient matrix which is $m x m$, but where $n$ is very large this computation is intractable. Instead, the coefficient matrix must be produced by a subset of the data which is $k x m$ where $m < k < n$. This is in the wrong place and is too much detail for here!

\subsection{Outline of some Dimensionality Reduction mechanisms}
%\hl{first time this term appears... need to decide on a name!}
%
%\hl{Lucia: maybe we could rename the section as it is an outine of DR approches (not focused only on nSimplex)}

We refer to the novel Dimensionality Reduction (DR) transform introduced here as  \nsimp \zen. In this section, by way of motivation, we  briefly contrast it with the three other best-known DR transform techniques. More detail of the other transforms is given in Section \ref{sec_intro_subsec_dim_red}, and a full definition of \nsimp \zen in Section \ref{sec_nsimp_projection_main}. 

In all cases we consider transforming a given space 
% of $n$ elements
% with dimensionality\footnote{We will define this term in Section \ref{sec_related}} $m$ 
of $n$ elements,  with a dimensionality%
\footnote{the term \emph{dimensionality} can be usefully applied to Hilbert spaces, see note in Section %\ref{sec_DR_subsec_dimensionality}
\ref{sec_intro_subsec_dim_red}
} 
of $m$, into a Euclidean space with $n$ elements and $k$ dimensions, where $k < m$.

\begin{description}

\item [Random Projection (RP)] is applicable only to an $m$-dimensional Euclidean space, where $m > k$. A randomised, orthonormal matrix of $m \times k$ dimensions is generated, and the $n \times m$ matrix representing the data is transformed by matrix multiplication to an $n \times k$ matrix representing the reduced dimensional space.

\item [Principal Component Analysis (PCA)] is also applicable only to an $m$-dimensional Euclidean space. A matrix of $m \times m$ dimensions representing the principal components is generated from the original space, and  the $n \times m$ matrix representing the data is transformed by matrix multiplication with the most significant $k$ columns of this matrix to yield the reduced dimensional space.

\item [Multidimensional Scaling (MDS)] is applicable to any metric or semi-metric space, and considers only the distances among the $n$ elements of the input space. An $n \times n$ (upper triangular) matrix of pairwise distances is taken as input, and the output is an $n \times k$ matrix, representing a Euclidean space which preserves pairwise distances as far as possible. MDS does not scale well, but an adaptation can be applied to large Euclidean spaces. A variant of MDS, Landmark MDS, can be applied to general metric spaces. %large non-Euclidean spaces.

\item [\nsimp \zen] is applicable to any metric space which is isometrically embeddable in Hilbert space%
\footnote{This includes any Euclidean space, and also metric spaces governed by appropriate variants of Cosine, Jensen-Shannon, Quadratic Form, and Triangular distances  (see \cite{Connor2016:HilbertExclusion} and Appendix \ref{appendix_hilbert_metrics} for details.)}.
It takes as input a reference set $\mathcal{R}$ comprising $k$ elements of the input space. A simplex in $(k-1)$ Euclidean dimensions is constructed using the pairwise distances measured within this set. With reference to this structure, each further element of the input space is then mapped to a $k$-dimensional Cartesian coordinate according to its distances from each element of $\mathcal{R}$.

\end{description}

\nsimp in its simplest form is a mapping from $(\mathcal{U},d)$ to $(\mathbb{R}^k,\ell_2)$, where $k$ is the size of the reference set $\mathcal{R}$. This can be used in its own right as a DR technique. There are however two further functions which can be applied to the $\mathbb{R}^k$ space resulting from the mapping: \zen and \emph{Upb}, which are defined in Section \ref{sec_nsimp_projection_main}. There are thus three different spaces which may be formed as a result of the projection. \nsimp \zen refers to the mapping from $(\mathcal{U},d)$ to  $(\mathbb{R}^k,\zen)$ and is the main object of our attention.

In contrast with the other techniques, calculation of the \nsimp transform does not require any  operations over matrices. Instead, it relies upon higher-level geometric properties of the %$m$-dimensional space, 
original metric space, which are reflected within the generated $k$-dimensional space.

\subsection{Relation to previous work}

The ideas underlying  \nsimp \zen have emerged after several years of research in the intersection between similarity search and distance geometry, in particular with respect to work done in the early 20th Century by, among others, Blumenthal, Hilbert, Menger, and Wilson, which we summarise in Section \ref{sec_intro_subsec_general}.

In \cite{Connor2016:HilbertExclusion} we showed how the four-point property possessed by some metric spaces%
\footnote{see Section \ref{subsec_embeddings} for a description of this and the $n$-point property} could be used to improve many metric search techniques, and in \cite{Connor2016:IS_Supermetric,Connor2016:SISAP_Supermetric} it was shown that this observation could be applied in practice to a number of state-of-the-art metric indexing techniques. We extended this work in \cite{Connor2017:nsimplex} after making the observation that the major four-point spaces we had identified (Jensen-Shannon, Quadratic Form, Cosine, Triangular) also possessed the $n$-point property, and that the distance lower-bound we had identified for four-point spaces in the 2D projection applied more generally in higher dimensions. During this work we also identified the algorithm, referred to as \nsimp \textit{Projection}, for constructing the simplex as given here in Appendix \ref{appendix_simplex_construction}, and the proof given in Appendix \ref{appendix:ProofCorrectness} that the \lwb function given in Section \ref{sec_nsimp_projection_main} is a lower bound of the distance in the domain of the \nsimp transform%
\footnote{In fact these were not included in the published version due to space limitations, but were given in an adjunct \emph{arXiv} publication \cite{connor2017:arxiv_nSimplex}}.

We exploited the \nsimp Projection also in the context of similarity search as an intermediate step to transform the original data into permutations and binary strings to be employed in efficient approximate metric indexes \cite{vadicamo2023induced,vadicamo2019splx,vadicamo2019metric}.

In \cite{connor2019modelling} we noted the application of the \zen function introduced in Section \ref{sec_nsimp_projection_main} is a better estimator of true distance than the \lwb function described in \cite{Connor2017:nsimplex} in the domain of string similarity functions. Moreover, in \cite{Vadicamo2021_IS_Re-ranking} the \zen function has been exploited on other domains (image features and word embeddings) to effectively refine candidate results obtained using a permutation-based k-NN search without accessing the original data.
We subsequently observed the very non-uniform pattern of angles measured within the simplexes thus formed from high-dimensional spaces. This angular distribution was separately examined and published in \cite{connor2020sampled}.

The main contribution of this article is the full exposition and analysis of the \zen function of the \nsimp construction as a general dimensionality reduction technique for Euclidean and other Hilbert spaces, all of which content is entirely novel.

\subsection{Paper Outline}
The rest of the article is structured as follows. Section \ref{sec_related} gives an overview of the necessary mathematical background required to understand the mechanism and motivation of the \nsimp \zen transform. Section \ref{sec_intro_subsec_dim_red} gives background on dimensionality reduction, and outlines the four mechanisms with which we  compare \nsimp \zen\!.  Section  \ref{sec_nsimp_projection_main} gives a detailed account of the \nsimp transform, and its three related functions: \zen, \lwb and \upb\!.  Section \ref{sec_experiments} gives experimental results comparing the quality of the reductions given by \nsimp \zen with the transforms introduced in Section \ref{sec_intro_subsec_dim_red}, and Section \ref{sec_performance} compares the run-time performance of the various mechanisms. Finally, Section \ref{sec_discussion} gives a discussion of the \nsimp \zen mechanism and its results, and concludes with some possible future work. Table \ref{tab_notations} summarises notations used throughout this work.

%  \begin{table}[tbp]
% \small
% \caption{Table of abbreviations used}
% \begin{center}
% \begin{tabular}{|l|l|l|}
% \hline
% Abbreviation		&Meaning&Section Introduced\\
% \hline
% DR&Dimensionality Reduction&\ref{sec_intro_subsec_dim_red}\\
% PCA&Principal  Component Analysis&\ref{sec_related_sub_pca}\\
% MDS&Multi-Dimensional Scaling&\ref{sec_related_sub_mds}\\
% RP&Random Projection&\ref{sec_related_sub_random_projection}\\
% LMDS&Landmark Multi-Dimensional Scaling&\ref{sec_related_sub_lmds}\\
% \lwb&The lower bound function of the \nsimp transform&\ref{sec_sub_properties}, Equation \ref{eqn_lwb}\\
% \zen&The zenith function of the \nsimp transform&\ref{sec_sub_properties}, Equation \ref{eqn_zen}\\
% \upb&The upper bound function of the \nsimp transform&\ref{sec_sub_properties}, Equation \ref{eqn_upb}\\
% \hline
% \end{tabular}
% \end{center}
% \label{tab_abbreviations}
% \end{table}%
\begin{table}[tbp]
\small
\caption{Table of notations used}
\begin{center}
\begin{tabular}{p{0.15\linewidth}p{0.8\linewidth}}
\hline
\textbf{Notation}		&\textbf{Meaning}\\
\hline
$n$&The cardinality of both domain and range of a DR transform\\
$m$&The  dimension of a Euclidean space which is the domain of a DR transform\\
$k$&The  dimension of a Euclidean space which is the co-domain of a DR transform\\
$(\mathcal{U},d)$&A  metric space with domain $\mathcal{U}$ and distance function $d$\\
$(\mathcal{S},d)$&A (typically large) finite subspace of $\mathcal{U}$, $\mathcal{S} \subset \mathcal{U}$\\
$\mathcal{R}$&A (typically small) set of reference objects drawn from $\mathcal{S}$ \\
$u, u_i$&Individual objects drawn from $\mathcal{U}$\\
$s, s_i$&Individual objects drawn from $\mathcal{S}$\\
$r, r_i$&Individual reference objects drawn from $\mathcal{R}$\\
%$(\mathcal{H},d)$&A  metric space with the Hilbert properties\\
$\T$& A DR transform mapping some $(\mathcal{U},d)$ to $(\mathcal{U}',\zeta)$\\
$\sigma$	& An \nsimp transform based on $k$ reference objects, $\sigma : \mathcal{U} \rightarrow \mathbb{R}^k$\\
%$\ell_2, \ell^m_2$ & The Euclidean distance function, sometimes explicitly over $\mathbb{R}^m$\\
$\ell_2$ & The Euclidean distance function\\
$(\mathbb{R}^m,\ell_2)$ & A Euclidean space of $m$ dimensions\\
$\delta_{ij}$& The distance $d(i,j)$ where $i,j \in (\mathcal{S},d)$\\
$\zeta_{ij}$& The  corresponding distance  $\zeta(i,j)$ for corresponding  $i,j \in (\mathcal{S}',\zeta)$, where $\mathcal{S}' = \T(\mathcal{S})$\\

\hline
\end{tabular}
\end{center}
\label{tab_notations}
\end{table}%

\section{Background: metric spaces and their properties}
\label{sec_related}

The novel \nsimp mechanism is described fully in Section \ref{sec_nsimp_projection_main}. Before it can be understood in detail a significant amount of mathematical background is required, and provided in this section. Section \ref{sec_intro_subsec_general} gives some mathematical preliminaries, and Section \ref{sec_intro_subsec_related_angles} gives some more specific background in high-dimensional geometry. Section \ref{sec_intro_subsec_dim_red} gives a general background to dimensionality reduction.

\subsection{Metric spaces, embeddings and simplexes}
\label{sec_intro_subsec_general}
Our work relies on the ability to construct a \emph{simplex} in a $k$-dimensional Euclidean space, whose edge lengths correspond to the distances among any $(k+1)$ objects selected from any metric space which is $(k+1)$-isometrically embeddable in a Hilbert space. These underlying concepts are briefly explained in this subsection.

\subsubsection{Metric and Semimetric Spaces}

Let $(\mathcal{U},d)$ be a pair comprising a domain of objects $\mathcal{U}$ and a numeric dissimilarity function $d : \mathcal{U} \times \mathcal{U} \rightarrow \mathbb{R}$. In general the more similar the objects $x,y \in \mathcal{U}$, the smaller the value of $d(x,y)$. For a space to be semimetric, it requires $d$ to be positive or zero, with $d(x,y) = 0$ if and only if $x = y$, and symmetric, i.e. $d(x,y) = d(y,x)$. A (proper) metric space, governed by a (proper) distance function, also possesses the triangle inequality property, i.e. $d(x,z) \le d(x,y) + d(y,z)$. For the rest of this article, we use the terms \emph{distance} and \emph{metric space} to refer to functions and spaces with these properties.

\subsubsection{Metric Spaces and Isometric Embeddings}
\label{subsec_embeddings}
For metric spaces $(\mathcal{U},d)$ and $(\mathcal{U}',\zeta)$ we say that $\mathcal{U}$ 
is \textit{isometrically embeddable} in $\mathcal{U}'$ if there exists a function $\F: \mathcal{U} \rightarrow \mathcal{U}'$ such that $\zeta(\F(x),\F(y)) = d(x,y)$ for all $x, y \in \mathcal{U}$.  A \textit{finite isometric embedding} is defined when such a function exists for a finite subset of $\mathcal{U}$. A finite isometric embedding may be generalised to any fixed size of subset; for example, we state that $(\mathcal{U},d)$ is finitely $n$-embeddable in $(\mathcal{U}',\zeta)$ where such a function exists for any subset of $n$ values selected from $\mathcal{U}$.

These concepts give rise to an alternative definition of a metric space. Normally, a metric space is defined as a semimetric space which has the triangle inequality property. Alternatively, a metric space may be defined as a semimetric space which is finitely 3-embeddable in 2D Euclidean space, this being an equivalent property.

Finite isometric embeddings
are summarised by Blumenthal \cite{blumenthal1933note}.
%, giving a general  summary of the underlying mathematics relevant to this paper. 
He defines the  \emph{four-point property} to refer to any space that is finitely  4-embeddable in $3$-dimensional Euclidean space.
Wilson \cite{wilson1932relation} shows various properties of such spaces, and Blumenthal points out that results given by Wilson, when combined with work by Menger \cite{Menger1928}, generalise to show that some
spaces with the four-point property also have the $n$-point property: that is, for any $n$, they are finitely $n$-embeddable in $(n-1)$-dimensional Euclidean space.
%This is  a more general result than we used in \cite{Connor2016:HilbertExclusion}, where we referred to a more modern formulation for high dimensional Euclidean space.
In a  later work, Blumenthal \cite{blumenthal1953} shows that any space which is isometrically embeddable in a Hilbert space has the $n$-point property. This is a generalisation of the better-known result that any $n$ points from a Euclidean space of any dimension may be isometrically embedded in $(n-1)$-dimensional Euclidean space, and is the main abstract result we rely upon here.

\subsubsection{Hilbert Spaces}
A Hilbert space is a real or complex inner product space that is also a complete metric space with respect to the distance function induced by the inner product.
Hilbert spaces possess inner product and distance functions with properties analogous to,  but not necessarily the same as, the dot product and distance functions of a Euclidean space.

A Hilbert space is a generalisation of a Euclidean space, allowing the study of vector spaces which do not necessarily have finite coordinate systems. Hilbert spaces have always been important in the study of abstract geometry. %All Euclidean spaces are of course isometrically embeddable in Hilbert space;
All Euclidean vector spaces are Hilbert spaces; in more recent years however, many further useful spaces have been identified as being isometrically embeddable in Hilbert space. These include spaces governed by the Jensen-Shannon, Quadratic Form, Triangular, and Cosine\footnote{in one particular form, see Appendix \ref{appendix_hilbert_metrics} for details of this and other metrics.} distances. These spaces do not have  Euclidean coordinates, and so cannot be manipulated via matrix arithmetic, but can be used as the domain of the transform we propose here as they inherit the $n$-point property of a Hilbert space.

%

% The two most important properties of Hilbert spaces for our purposes are:
% \begin{itemize}
% \item any $n$ values from a Hilbert space may be isometrically embedded in an $(n-1)$-dimensional Euclidean space, and
% \item all Hilbert spaces are continuous
% \end{itemize}

% The inner product function also assigns a sense of orthogonality to any Hilbert space.

%  It is an immediate tautology that any Hilbert space is also  metric space, as a metric space has an alternative definition of one which supports the isometric embedding of any 3 values into a 2-dimensional Euclidean space

\subsubsection{Construction of a simplex}
A simplex is the generalisation of a triangle or a tetrahedron in arbitrary dimensions of Euclidean space. 
In one dimension, a simplex is a line segment. In two dimensions it is a triangle, while in three dimensions it is a tetrahedron.
In general, a point $\bb v_1$ forms a 0-simplex, and the $n$-simplex of vertices $\bb v_1,\dots,\bb v_{n+1}$ is given by the union of the simplex formed from $\bb v_1,\dots, \bb v_n$ with the line segments joining $\bb v_{n+1}$ to all vertices of that simplex. 

The property that a Hilbert space $(\mathcal{U},d)$ is finitely $(n+1)$-embeddable in $n$-dimensional Euclidean space directly implies that, for any $(n+1)$ objects in $\mathcal{U}$, it is possible to construct a simplex with $(n+1)$ vertices in $\R{n}$, where each vertex corresponds to an object in $\mathcal{U}$, and the edges joining all pairs of adjacent vertices correspond with the distances between the corresponding objects in $\mathcal{U}$. %In general  $n$ Euclidean dimensions are required to contain this simplex.

In Appendix \ref{appendix_simplex_construction} we show an algorithm for determining Cartesian coordinates for the vertices of a  simplex, given only the distances between all pairs of points. The algorithm is  inductive, at each stage allowing the apex of an $n$-dimensional simplex to be determined given the coordinates of an $(n-1)$-dimensional simplex, and the distances from the new apex to each vertex in the existing simplex. 
%This is important because, given a fixed base simplex over which many new apices are to be constructed, the time required to compute each one is linear with the number of dimensions.
% We are not sure if this is generally known; we were unable to find any academic literature describing a method, and required to invent it ourselves! A way of constructing coordinates for a regular simplex is given in Wikipedia.

%Moreover, the $n+1$ vertices of each $n$-simplex lie on an $(n-1)$-dimensional surface of an $n$-dimensional space. 

The outcome of this algorithm represents a simplex in $n$-dimensional space as a lower triangular $n+1$ by $n$ matrix representing the Cartesian coordinates of each vertex.
For example, %the following matrix represents  
the rows of the following matrix represent the coordinates $v_{i,j}$ of four vertices $\bb v_1,\dots, \bb v_4$ of a tetrahedron in 3D space:
\begin{equation}
\begin{bmatrix}
0		&	0		&	0		\\
v_{2,1}	&	0		&	0		\\
v_{3,1}	&	v_{3,2}	&	0	\\
v_{4,1}	&	v_{4,2}	&	v_{4,3}
\end{bmatrix}
\end{equation}

This matrix is derived from four objects $o_1,\dots, o_4$ in the Hilbert space, and the distances $d(o_i,o_j)$ are the same as the distances %$||i-j||$ where $i$ and $j$ are vectors given by the corresponding rows in the matrix.
%$\|\bb{v}_i-\bb{v}_j\|$ 
$\ell_2(\bb{v}_i,\bb{v}_j)$ 
where $\bb{v}_i$ and $\bb{v}_j$ 
are vectors given by the $i$-th and $j$-th rows of the matrix, respectively.

For all such sets of objects, the invariant that $v_{i,j} = 0$ whenever $j \ge i$ can be  maintained without loss of generality. For any simplex constructed, this can be achieved by rotation and translation within the Euclidean space while maintaining the distances among all the vertices. Furthermore, if we restrict $v_{i,j} \ge 0$ whenever $j = i-1$ then in each row this component represents the \emph{altitude} of the %$i^{th}$ 
point $\bb v_i$ with respect to a base simplex formed by $\{\bb v_1, \dots, \bb v_{i-1}\}$, which is  represented by the matrix derived by selecting elements above and to the left of the entry $v_{i,j}$.

Finally, we note that as long as the entry $v_{i,j}$ %$(i-1)$th component of each row $i$ 
is non-zero, i.e. represents a non-zero altitude above the base simplex defined by rows $1$ to $i-1$, then the set of vectors defined by the rows forms a basis for the $n$-dimensional space in which it is constructed. In this way, the process of forming the simplex gives an interesting comparison to the Gram-Schmidt method for forming a basis in a Euclidean space, but the simplex formation method does not require access to a coordinate space defining the original metric space, and can thus be applied to any metric space which is isometrically embeddable in a Hilbert space.

% As a corollary, the volume of the simplex can be calculated by the product of this diagonal divided by the factorial of the dimension.

%%FIGURE?? 
%If a metric space $(U,d)$ has the $n$-point property, we can construct an $n$-simplex in $\ell_2^n$ based only on the distances measured among $n+1$ points of $U$.
%This corresponds to selecting an isometric embedding $\phi$ of the $n+1$ points $p_1,\dots,p_{n+1} \in U$ into $\ell_2^n$ and considering the $n$-simplex spanned by $\phi(p_1),\dots,\phi(p_{n+1})$.

%As a base case, the one-dimensional simplex constructed for two objects $p_1,p_2$ with $d(p_1,p_2) = \delta$ is represented as
%$\begin{bmatrix}
%0		\\
%\delta
%\end{bmatrix};$
%the two-dimensional simplex constructed for three objects $p_1,p_2, p_3$ with $d(p_i,p_j) = \delta_{i,j}$ is
%\[
%\begin{bmatrix}
%0				&	0				\\
%v_{2,1}	&	0			\\
%v_{3,1}	&	v_{3,2}		
%\end{bmatrix}
%\]
%where
%\begin{align*}
%v_{2,1}&=\delta_{1,2}\\
%v_{3,1}&=\frac{\delta_{1,3}^2-\delta_{2,3}^2}{2\delta_{1,2}}+\frac{\delta_{1,2}}{2}\\
%v_{3,2}&=\sqrt{\delta_{2,3}^2-(v_{3,1}-\delta_{1,2})^2}.
%\end{align*}

\subsection{Angles in High-Dimensional Metric Spaces}
\label{sec_intro_subsec_related_angles}

In Section \ref{sec_zen_function} we rely upon a property on the distribution of angles in high-dimensional Euclidean space that is described in this section. While the property itself is relatively straightforward, its derivation from high-dimensional Euclidean geometry is less so, and we therefore give a short justification.

In the context of a uniformly distributed Euclidean space of $n$ dimensions, we are interested in the distribution of angles formed by a hyperplane ${H}$ and a object $\bb c$ on a hypersphere\footnote{A hyperspheres in $\R{n}$ is a $n-1$-sphere. Note the possibly confusing conventions in the naming of  $n$-spheres and $n$-balls: in general, the \emph{surface} of an $n$-ball is denoted as an $(n-1)$-sphere. For example, the volume contained by a sphere in 3D space is a $3$-ball, whereas its surface is a $2$-sphere.} centred on a point $\bb b\in {H}$. Without loss of generality this distribution can  be measured as the angle between three points $\bb a$, $\bb b$, and $\bb c$, where objects $\bb a$ and $\bb b$ are fixed in $H$, and $\bb c$ is sampled within a fixed radius $r$ from $\bb b$. The considered situation can be easily depicted in 2D (i.e., considering the plane through these three points) in Figure \ref{fig_high_dim_hypersphere_1}.
Given this arrangement, we wish to understand the distribution of the angle $\theta \in [0,\pi]$, that is the angle formed by the three points.

In a high-dimensional vector space it is generally known that two randomly selected vectors are very likely to be close to orthogonal \cite[Chapter 2]{chapter_blum_hopcroft_kannan_2020,blum_hopcroft_kannan_2020}, and it is therefore no surprise that the value of $\theta$ is likely to be close to $\pi/ 2$. In fact this is the only result required for the understanding of Section \ref{sec_zen_function}, however further explanation is reasonably required.

% \begin{figure}[h]
% \begin{center}

% \begin{tikzpicture}[scale=0.3]
% \draw[dotted] (5,0) arc (0:180:5);
% \draw [dotted](0,0) -- (-5,0);
% % a
% \draw (10,0) circle (2pt);
% \draw (10,-0.6)  node {$\bb a$};
% % b
% \draw (0,0) circle (2pt);
% \draw (0,-0.6) node{$ \bb b$};
% %r
% \draw [->, dotted](0,0) -- (-4,3);
% \draw (-2,2.5) node[anchor=east] {$r$};

% \draw (1.8,1) node {$\theta$};
% \draw  [->] (3,0) arc  (0:53.1:3);

% %abc
% \draw[thick]  (3,4) -- (0,0)-- (10,0);

% \draw  (3,4) circle (2pt) node[anchor=south] {$\bb c$};
% \end{tikzpicture}
% \end{center}

% \caption{
% In $n$ dimensions, for fixed $\bb a$ and $\bb b$ within a given plane, $\bb c$ is sampled from within the same plane at a fixed radius $r$ from $\bb b$. %The locus of $\bb c$ is  an $(n-2)$-sphere. Symmetry implies that the mean value of $\theta$ is zero. As $n$ increases, so too does the probability of $\theta$ being close to zero.
% As $n$ increases, so too does the probability of $\theta$ being close to $\pi/ 2$.
% }
% \label{fig_high_dim_hypersphere_1}
% \end{figure}

\begin{figure}[tbp]
\centering
%fig a
 \begin{subfigure}[b]{0.28\textwidth}
 \begin{tikzpicture}[scale=0.28]
\draw[dotted] (5,0) arc (0:180:5);
\draw [dotted](0,0) -- (-5,0);
% a
\draw (7,0) circle (2pt);
\draw (7,-0.6)  node {$\bb a$};
% b
\draw (0,0) circle (2pt);
\draw (0,-0.6) node{$ \bb b$};
%r
\draw [->, dotted](0,0) -- (-4,3);
\draw (-2,2.5) node[anchor=east] {$r$};

\draw (1.8,1) node {$\theta$};
\draw  [->] (3,0) arc  (0:53.1:3);

%abc
\draw[thick]  (3,4) -- (0,0)-- (7,0);

\draw  (3,4) circle (2pt) node[anchor=south] {$\bb c$};
\end{tikzpicture}
\caption{}\label{fig_high_dim_hypersphere_1}
\end{subfigure}%
%fig b
 \begin{subfigure}[b]{0.28\textwidth}
    \begin{tikzpicture}[scale=0.28]
    \shade (0,0) circle (5);
    \draw[dotted] (0,0) circle (5);
    
    \filldraw[fill=green!20,draw=green!50!black] (3,-4) -- (3,4) arc (53.13:-53.13:5)--(3,-4);
    % a
    \draw (7,0) circle (2pt);
    \draw (7,-0.6)  node {$\bb a$};
    % b
    \draw (0,0) circle (2pt);
    \draw (0,-0.6) node{$ \bb b$};
    %r
    %\draw [->, dotted](0,0) -- (-4,3);
    %\draw (-2,2.5) node[anchor=east] {$r$};
    
    \draw (1.5,1) node {$\theta$};
    \draw  [->] (2.5,0) arc  (0:53.1:2.5);
    \draw (1.5,-1) node {$\theta$};
    \draw  [->] (2.5,0) arc  (0:-53.1:2.5);
    
    %abc
    \draw[thick]  (3,4) -- (0,0)-- (7,0);
    \draw[thick]  (3,-4) -- (0,0);
    
    \draw  (3,4) circle (2pt) node[anchor=south] {$\bb c$};

    \draw (-5,-6) circle (1pt) node[anchor=south] {$-r$};
    \draw (5,-6) circle (1pt) node[anchor=south] {$r$};
    \draw[<->] (-5,-6) -- (5,-6);
    \filldraw[fill=red!20,draw=red!50!black] (3,-6) circle (2pt);
    \filldraw[fill=red!20,draw=red!50!black] (3,-5.5) node{\color{red}{$t$}};
    \end{tikzpicture}
    \caption{}\label{fig_area}
\end{subfigure}%
%fig c
 \begin{subfigure}[b]{0.44\textwidth}
\includegraphics[width=\textwidth]{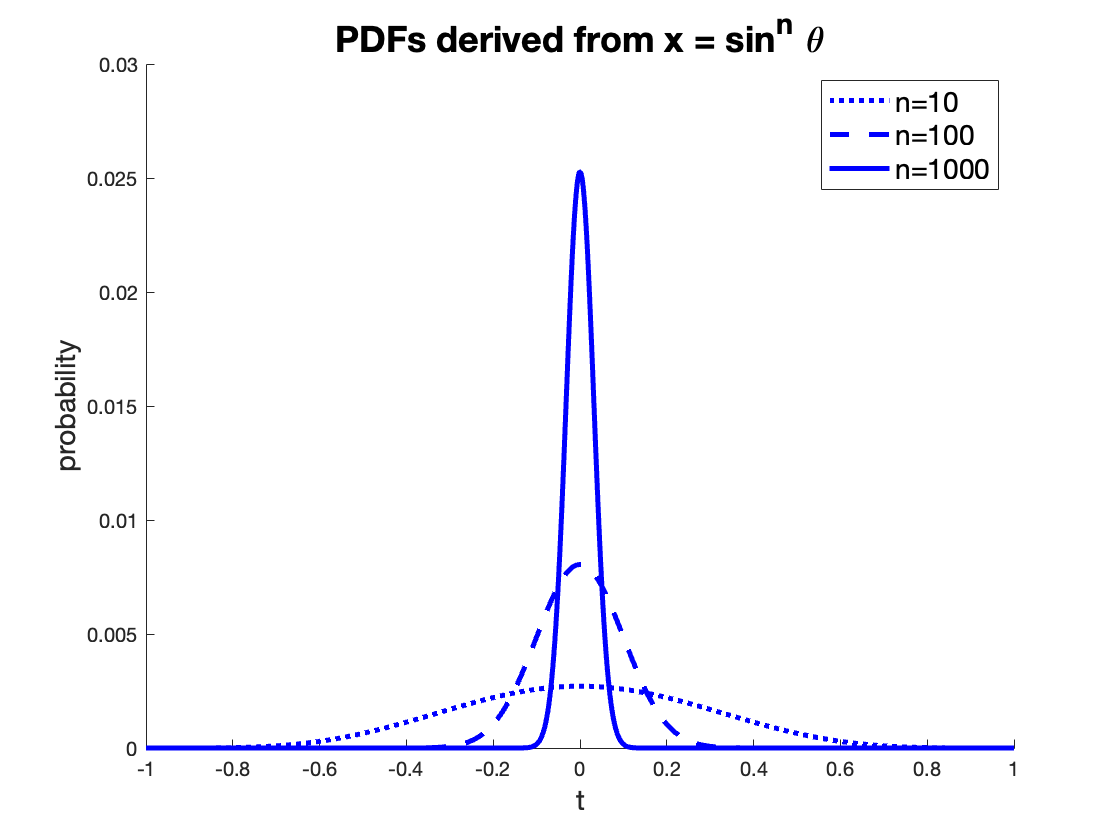}  
\caption{}\label{fig_high_dim_hypersphere_2}
\end{subfigure}%
%\caption{sd}
%\end{subfigure}
%

\caption{
In $n$ dimensions, for fixed $\bb a$ and $\bb b$ within a given plane, $\bb c$ is sampled from within the same plane at a fixed radius $r$ from $\bb b$.
%For fixed $\bb a$ and $\bb b$, $\bb c$ is sampled from the surface of a $n$-ball of radius $r$, centred on $\bb b$, intersected by the plane containing $\bb a,\bb b,\bb c$.
As the dimensionality of the space increases, the probability of $\theta$ being close to  $\pi/ 2$ increases rapidly: the right-hand plot shows probability density functions for various dimensions as $t=r\cos \theta$ varies between $-r$ and $r$.
}

\end{figure}
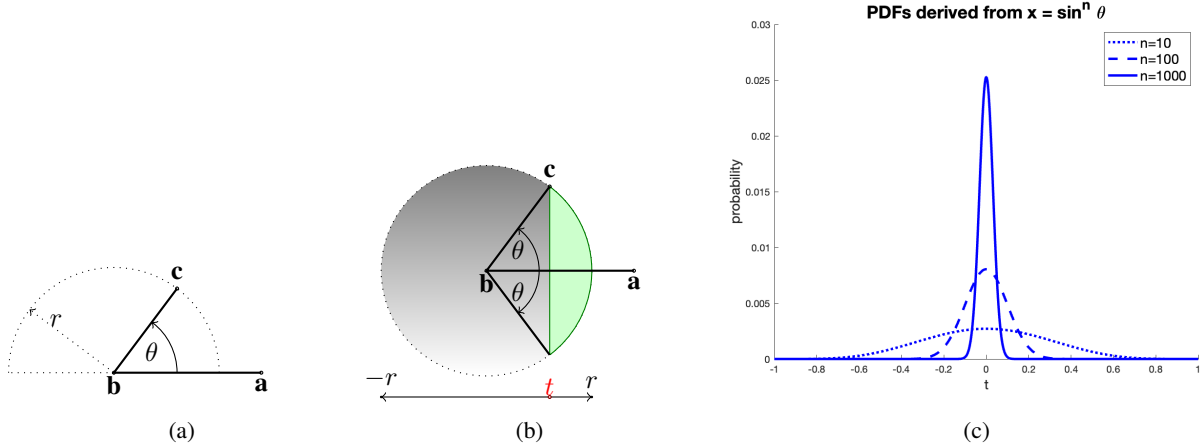

In \cite{connor2020sampled} we quantified the distribution of this angle and we observed that as the dimension of the space $n$ increases so too does the probability of $\theta$ being close to $\pi/ 2$. In particular we observed that
\begin{itemize}
    \item the total volume $V_n$ of a $n$-ball with radius $r$ is $V_n=r^n U_{n-1} \int_0^\pi \sin^n \psi d\psi$, where $U_{n-1}$ is the volume of a unit (n-1)-ball;
    \item  the volume $V_n(\theta)$ of the portion of the $n$-ball delimited by the hyperplane through $\bb c$ and orthogonal to $H$ (i.e. the portion delimited by $\theta$  and denoted by the green-shaded area in Figure \ref{fig_area}) is proportional to $r^n \int_0^\theta \sin^n \psi d\psi$.  
\end{itemize} 

Actually, here we are interested in the ``surface'' of the portion of the $n$-ball delimited by $\theta$, rather than its volume. However such regions are strongly related geometrically; in fact in general the surface of an $n$-ball, i.e. an $(n-1)$-sphere, has a volume in $(n-1)$ dimensions. The angular distribution of volume in an $n$-sphere is identical to that in an $(n-1)$-ball \cite{voelker2017}, therefore the PDF for the distribution of points on an $(n-1)$-sphere is given by a normalisation of the same formula.

While the formula to computes the volume is difficult to quantify, requiring the integration of high powers of the sine function, it has also been observed that for high values of $n$ the function is numerically almost indistinguishable from the normal distribution function given by setting $\mu = \pi / 2$ and $\sigma = 1 / \sqrt{(n)}$ \cite{cai2013distributions,stack_overflow}, for which integral values are highly accessible. Figure \ref{fig_high_dim_hypersphere_2} quantifies this volume as PDFs for various dimensions $n$. It can be observed that as $n$ increases the distribution of $t$ concentrates around the mean value $0$, so the distribution of the angle $\theta$ concentrates around $\pi/2$.

%the surface defined by an $n$-sphere, rather than the volume of an $n$-ball.  However such regions are strongly related geometrically; in fact there is no real notion of ``surface'' beyond a sphere defined in 3 dimensions, and in general the surface of an $n$-ball, i.e. an $(n-1)$-sphere, has a volume in $(n-1)$ dimensions. The angular distribution of volume in an $n$-sphere is identical to that in an $(n-1)$-ball \cite{voelker2017}, therefore the PDF for the distribution of points on an $(n-1)$-sphere is given by a normalisation of the same formula.

% In \cite{connor2020sampled} we quantified the distribution of volume within an $n$-ball, in a similar situation. The outcome of this is that a probability distribution function (PDF) for the volume of an $n$-ball according to the angle $\theta$ is directly proportional to the function
% \[h(\theta) = \sin^{n} \theta, \theta \in [0,\pi]\]
% %

% Figure \ref{fig_high_dim_hypersphere_2} shows a quantification of this volume for an $n$-ball. The volume denoted by the green-shaded area in the left-hand diagram is given by the definite integral 
% \[\int_{\pi / 2 - \theta}^{\pi / 2 + \theta} k \sin^{n} \psi  \, d\psi\]
% %
% for some constant $k$, where $n$ is the dimension of the Euclidean space. The right-hand side figure quantifies this formula as PDFs for various dimensions.

This theoretical model can be verified by experiment in unbounded Euclidean spaces. As usual, there are various factors in real-world spaces (in particular boundedness and non-uniformity) which affect the observed distribution of angles, but the theoretical effect is still highly visible; a deeper analysis  is given in \cite{connor2020sampled}. In Section \ref{sec_experiments} we analyse Euclidean spaces drawn from examples in 100 to 4,096 dimensions, in which  the effect is very evident.

\section{Related work: dimensionality reduction}
\label{sec_intro_subsec_dim_red}
\begin{table}[tb]
\footnotesize
\caption{Dimensionality reduction algorithms used in this research.} \label{table:dr}
\begin{center}
\begin{tabular}{p{0.07\linewidth}p{0.3\linewidth}p{0.26\linewidth}p{0.26\linewidth}}
\hline
 \textbf{Name}  &  \textbf{Goals} & \textbf{Input} & \textbf{Output}\\
\hline
PCA  & Project data onto a lower-dimensional subspace while maximising the retained variance & Euclidean space: original data (or its covariance matrix) & Projection matrix (principal components) and projected data  \\
RP  & Project data onto a lower-dimensional subspace with small distortion in pairwise distances  &  Euclidean space: original data & Random projection matrix and projected data\\ %uses random projection matrices to project the data into lower dimensional spaces
MDS  & Represent data in a lower-dimensional space, preserving the rank ordering of pairwise distances & Any semi-metric space: distance matrix (all pairwise distances) & Transformed data (coordinates) \\
LMDS  & As MDS, but using a subset of landmarks to reduce computational cost & Any semi-metric space: pairwise distances between data objects and landmark points, and pairwise distances among all landmark points & MDS on landmarks and transformed data (coordinates) \\
nSimplex Zen & Represent data in a lower-dimensional space with small distortion in pairwise distances & Any Hilbert-embeddable space: pairwise distances between data objects and reference objects, and pairwise distances among all reference objects & Matrix representing a base simplex and the transformed data (coordinates) \\
\hline
\end{tabular}
\end{center}
\label{tab_abbreviations}
\end{table}%

In simple terms, \textit{dimensionality reduction} refers to the transformation of a set of values in dimension $m$ to a set of values in dimension $k$, where $k < m$. Associated with such a transformation is a controlled loss of information within the reduced-dimension space. 

The notion of {dimensionality} in general metric spaces is itself complex, and indeed the establishment of a generally agreed definition of \textit{intrinsic dimensionality} is an ongoing research issue. In Euclidean spaces dimensionality can be understood as the minimum value of $m$ for which an isometric embedding in an $(\mathbb{R}^{m},\ell_2)$  space can be defined. For example, a set of points on a plane can be defined in a 3-dimensional space, but the intrinsic dimensionality of this set is 2. Mechanisms such as PCA and MDS can be used to detect the value of intrinsic dimensionality in a Euclidean space.
In general, metric spaces cannot be isometrically embedded in Euclidean space and so this definition cannot be used. The concept of dimensionality is important in non-vector spaces, and can be defined  upon properties such as the distribution of distances among objects of the space, as shown for example in \cite{Chavez2001similarity}. In this article we  refer to ``dimensionality reduction'' transforms which map from %non-Euclidean 
metric spaces to Euclidean spaces of a lower dimensionality, in these cases we implicitly rely upon such definitions of dimensionality.

Dimensionality reduction techniques are usually classified into two groups: linear and topological. Linear mechanisms are concerned with preserving the accuracy of simple distances, while topological mechanisms are concerned with the preservation of  higher-level properties, for example  relative rather than absolute distances. %Here we introduce PCA and MDS as perhaps the best-known and most commonly used general mechanisms in either class.
In all cases, the main purpose is to reduce the cost of distance comparisons, by reducing both the size of the data and the cost of the distance measurement.
In most contexts, we are interested in applying such transforms to a finite metric space $(\mathcal{S},d)$ which is a (typically very large) subset of an infinite space $(\mathcal{U},d)$. For example in metric search, the task is to find, from the subset $\mathcal{S}$, values which are similar to a query value $q \in \mathcal{U}$, where typically $q \notin \mathcal{S}$.
This causes us to reconsider the general notion of dimensionality reduction. Rather than a transform %$\T$ 
which maps some finite space $(\mathcal{S},d)$ to another space $(\mathcal{S}',\zeta)$,  where  $(\mathcal{S},d)$ is a finite subset of some infinite space $(\mathcal{U},d)$, we require a more general transform $\T: (\mathcal{U},d) \to(\mathcal{U}',\zeta)$ 
which can map individual elements of $\mathcal{U}$ to $\mathcal{U}'$. This is important for two reasons. First, any analytic technique that analyses the entire finite domain $\mathcal{S}$ in order to map to a new one will be intractable for very large finite domains. Second, such a technique would not allow the mapping of a query value $q$ where $q \in \mathcal{U}$ but  $q \notin \mathcal{S}$.
We make this distinction as classical definitions of techniques such as %PCA and 
MDS assume analysis of the entire finite domain. %\footnote{For the PCA a subset of points is sometimes used to estimate the transformation.}. 

In this section, as well as giving outline descriptions of these techniques, we also show how they can be applied in this more general context. There are many dimensionality reduction techniques, many of which have been developed for specific contexts. The Principal Component Analysis (PCA) \cite{pearson1901:PCA,hotelling1933analysis} stands out as the most widely used dimensionality reduction algorithm, finding applications in data compression, data processing, feature extraction, and data visualisation, among others. Random Projection (RP) \cite{dasgupta2013experiments}, Multidimensional Scaling (MDS) \cite{cox2008multidimensional,kruskal1964multidimensional,LMDS}, and their variants aim at representing the data in a lower-dimensional space in such a way that the distances between points in that space approximate the pairwise dissimilarities in the original space. Isomap \cite{tenenbaum2000global} focuses on preserving geodesic distances and local structures, often used in data analysis and machine learning. Additionally, other techniques such as t-Distributed Stochastic Neighbour Embedding (t-SNE) \cite{van2008visualizing} and Uniform Manifold Approximation and Projection (UMAP) \cite{mcinnes2018umap}, are commonly employed for exploratory data analysis and visualisation.   Given our focus on distance-preserving mechanisms and not on data visualisation, here we introduce three general mechanisms most relevant to our context, namely PCA, RP, MDS,  which we believe are the  mechanisms in most common use. Table \ref{table:dr} provides an overview of these mechanisms, alongside our proposed nSimplex Zen.

% \subsection{Dimensionality}
% \label{sec_DR_subsec_dimensionality}
% The notion of dimensionality in general metric spaces is itself complex, and indeed the establishment of a generally agreed definition of intrinsic dimensionality is an ongoing research issue. In Euclidean spaces dimensionality can be understood as the minimum value of $m$ for which an isometric embedding in an $(\mathbb{R}^m,\ell_2)$  space can be defined. For example, a set of points on a plane can be defined in a 3-dimensional space, but the intrinsic dimensionality of this set is 2. Mechanisms such as PCA and MDS can be used to detect the value of intrinsic dimensionality in a Euclidean space.

% In general, metric spaces cannot be isometrically embedded in Euclidean space and so this definition cannot be used. The concept of dimensionality is important in non-Euclidean spaces, and can be defined  upon properties such as the distribution of distances among objects of the space, as shown for example in \cite{Chavez2001similarity}. In this article we  refer to ``dimensionality reduction'' transforms which map from non-Euclidean spaces to Euclidean spaces of a lower dimensionality, in these cases we implicitly rely upon such definitions of dimensionality.

\subsection{Random Projection}
\label{sec_related_sub_random_projection}
According to the Johnson-Lindenstrauss Flattening Lemma (see e.g. \cite[page 358]{matousek2013book}), a random projection can be used to transform a finite set of $n$ Euclidean vectors into a $k$-dimensional Euclidean space ($k<n$) with a ``small'' distortion. Specifically the Lemma asserts that for any $n$ points  of the space $(\mathbb{R}^m,\ell_2)$, and for every $0<\epsilon < 1$, there is a mapping into $(\mathbb{R}^k,\ell_2)$ that preserves all the pairwise distances within a factor of $1 + \epsilon$, where $k = O(\epsilon^{-2}\log n)$. Note that this lemma depends on the size, and not the dimensionality, of the domain.

%The distortion depends on $n$, and we have have almost an isometric embedding if $k=O(\log n)$. 
The low dimensional projection anticipated by the Johnson-Lindenstrauss lemma is particularly simple to implement. Specifically, for a Euclidean space represented by an $n \times m$ matrix, a suitable  transform into $k$ dimensions can be achieved through a randomly generated $m \times k$ orthonormal matrix.
Practically, there are even better ways of achieving the projection. Achlioptas \cite{achlioptas_rp} shows that equally good results can usually be achieved with a much cheaper transform, by creating a $m \times k$ pseudo-orthogonal matrix with, for example, the randomised strategy:
\begin{equation}
R_{i,j} = \sqrt{3} \times
\begin{cases}
+1 \quad &\text{with probability 1/6}\\
0 \quad &\text{with probability 2/3}\\
-1 \quad &\text{with probability 1/6}
\end{cases}
\end{equation}
This greatly improves the efficiency of the projection, as it introduces many zeros into the projection matrix, and allows integer arithmetic to be used instead of floating point. Further, the strategy has been shown to give better outcomes in some circumstances than a truly orthonormal matrix. We use this strategy in our comparative experiments in Section \ref{sec_experiments}.

It is self-evident that, for Euclidean data which is  uniformly distributed,  a random projection is no worse than any other linear technique. However, much real-world data is not uniformly distributed, often in ways that are difficult to predict or analyse, in which case other techniques typically perform better.

\subsection{Principal  Component Analysis}
\label{sec_related_sub_pca}

PCA \cite{pearson1901:PCA,hotelling1933analysis} is probably the best known and most widely used unsupervised dimensionality reduction technique, and has been used also for feature extraction and data visualisation.  
The main idea is to find a linear transformation of $m$-dimensional vectors to $k$-dimensional vectors ($k < m$) that best preserves the variance of the input data. Specifically, PCA determines the principal  components of the data, which are those  directions within the vector space showing maximum variance. The first such direction is found, and represented by a unit vector; then, the second direction is found within the $(m-1)$-dimensional subspace orthogonal to this unit vector, and so on until a set of $m$ orthonormal vectors is established. These vectors are represented in an $m$-dimensional square matrix whose columns correspond to the unit vectors established by this process (i.e., the so-called principal components).

If the intrinsic dimensionality of the data is less than $m$, then the last steps of the process will discover a variance of 0 in all directions and the unit vectors derived become arbitrary.

The principal  components can be computed by solving a maximisation problem. However, it has been shown that the principal  components are the eigenvectors of the covariance matrix of the centred input data. Thus typically they are computed by using spectral analysis via Singular Value Decomposition of the data rather then solving the optimisation problem, which is more expensive.

The eigenvalues $\lambda_1, \dots, \lambda_m$ give the variances of their respective principal components. Moreover, the ratio
\begin{equation}
\label{eqn_pca_variance}
    \frac{\sum_{i=1}^k \lambda_i}{\sum_{j=1}^m \lambda_j}
\end{equation}
represents the proportion of the total variance in the original data set accounted for by the first $k$ principal components.

The dimensionality reduction transform itself is achieved by multiplication of the matrix representing the domain by the first $k$ columns of the principal component matrix. If the value given by Equation \ref{eqn_pca_variance} is large, then the loss of accuracy in distances measured in the projected space will be correspondingly small.

While PCA is defined as the orthogonal projection of the data onto a lower dimensional linear subspace, such that the variance of the projected data is maximised, there also exists an equivalent definition of PCA that gives rise to the same algorithm. In the latter, the PCA is defined as the linear projection that minimises the average projection cost, defined as the mean squared distance between the data points and their projections \cite{Pearson1901}. This property  implies that PCA is the best strategy for dimensionality reduction in a Euclidean space, where the goal is to minimise the introduced inaccuracy of arbitrary distance measurements.

For a very large data set, the cost of calculating the principal components using the entire set is likely to be intractable. However this cost may be avoided by using  a representative sample of the data to generate the principal components. As the projection to construct the reduced-dimension set comprises multiplication of the $n \times m$ data matrix by the $m \times k$ principal component matrix, principal components derived from a representative subset can be used to transform the remainder of the data. 
%When a similarity query is required against an item which is not present within the original data set, then this query can also be transformed using the component matrix before the query is performed in the reduced-dimension space.

\subsection{Multidimensional Scaling}
\label{sec_related_sub_mds}
MDS \cite{cox2008multidimensional} is a technique which analyses the pairwise distances within a finite semimetric space $(\mathcal{S},d)$ and, given a target dimension $k$, generates a $k$-dimensional Euclidean space which preserves topological features of these distances as far as possible. There are two main variants of MDS, so-called ``classic'' (metric) and non-metric.  Here we consider the ``non-metric'' version as this may be applied to spaces not governed by the Euclidean distance and can thus be  compared with the \nsimp technique.

MDS iteratively constructs  Euclidean vectors, using a gradient descent technique, in order to minimise a stress formula. In the non-metric variant, this is typically  Kruskal's (see Section \ref{sec_global_stress}) \emph{stress1} definition:
\begin{equation}
S_K = \sqrt{\frac{\sum_{i<j}(\zeta_{ij} - d^*_{ij})^2}{\sum_{i<j} \zeta_{ij}^2}}
\end{equation}
In this formula, $i$ and $j$ are indices over the data objects in $\mathcal{S}$, $\zeta_{ij}=\zeta(\T(s_i),\T(s_j)$ is the Euclidean distance measured in the reduced space, and $d^*_{ij}=d^*(s_i,s_j)$ is a function generated by an isotonic regression over the true distances $d(s_i,s_j)$ as a function of the reduced distances $\zeta_{ij}$. Stress is therefore affected not  by the absolute difference between distances in the two spaces, but instead according to the relative ordering of distances between them: if this is preserved, then the measured stress will be lower.

MDS is an expensive ($\mathcal{O}(n^4)$) algorithm to compute, significantly limiting the size of data to which it can be applied.  It has the further disadvantage that as the analysis is over a given finite set of distances among objects, it cannot therefore  produce a transform which may be applied to other non-manifest elements of the same domain. This would imply that a representative sample  cannot be used to construct a transform which can subsequently be applied to a very large domain, and also that a query from the same universal domain cannot be subsequently transformed into the generated $k$-dimensional space.

In fact it may be possible to generate such a transform when MDS is applied to an $m$-dimensional Euclidean space, using Procrustes analysis and a pseudo-inverse matrix operator, as follows. First, a representative set of $l$ objects is selected from the $n$ objects of the domain, and the $l \times l$ distance matrix is generated. MDS takes this as input and an $l \times k$ matrix is produced to represent the $k$-dimensional space representing those objects. Procrustes analysis can then be used to produce a transform from this space back to a best fit within the original $m$-dimensional space. Although in general this transform is represented by a non-square matrix, and therefore is not guaranteed to have an inverse, a pseudo-inverse technique can be used to successfully construct the inverse transform. This can therefore be subsequently applied to other samples from the same original space. We use this technique in Section \ref{sec_experiments} to compare MDS as a dimensionality reduction technique for large Euclidean spaces.

\subsection{Landmark Multidimensional Scaling}
\label{sec_related_sub_lmds}
Landmark MDS (LMDS) \cite{LMDS} is a technique which allows MDS to be used for the generation of a general transform over metric spaces, using a triangulation technique. A representative set of landmark values $\mathcal{L}$ is selected from the domain $\mathcal{U}$, and classical MDS is applied in order to transform the (typically) non-Euclidean space $\mathcal{L}$ to a $k$-dimensional Euclidean space, minimising the stress as above. 

As already noted, classical MDS does not generate a transform function which can be applied to data not included in the manifest space whose distances are used to construct this transform. Instead, LDMS allows the addition of further elements of the domain  to the transformed space using only the distances calculated to each element of $\mathcal{L}$. A triangulation approach is then used to place each subsequent element into the reduced-dimension space with minimal stress on this set of distances.

In this manner, LMDS extends classical MDS in such a way that it can be extended for use over very large data sets and non-manifest queries in non-Euclidean spaces. We use this technique in Section  \ref{sec_sub_exp_jsd} to compare LMDS with our \nsimp \zen transform over spaces not governed by the Euclidean distance.

\section{The \nsimp Projection}
\label{sec_nsimp_projection_main}
The \nsimp transform can be applied to Hilbert spaces in general, and relies on the Hilbert property that any $k+1$ values can be isometrically embedded in an $k$-dimensional Euclidean space.
In outline, the transform from a Hilbert space $\mathcal{U}$ to $\R{k}$ %a $k$-dimensional Euclidean space 
is defined as follows:
\begin{enumerate}
\item $k$ values $r_1, r_2, \dots r_k$ are first selected from $\mathcal{U}$ to form a reference set $\mathcal{R}$. (Typically, $\mathcal{R}$ will be selected from  a large finite subset $\mathcal{S}$ of $\mathcal{U}$.)
\item All pairwise distances among the values  in $\mathcal{R}$ are calculated, and used to construct a \textit{base simplex} $\Sigma$ in a Euclidean space of $k-1$ dimensions, where each vertex $\bb v_i$ in $\Sigma$ corresponds to one value $r_i$ in $\mathcal{R}$ with  $\ell_2(\bb v_i, \bb v_j)=d(r_i,r_j)$ for all $i,j=1, \dots, k$.
\item For any further value $u \in \mathcal{U}$, the distances between $u$ and all values in $\mathcal{R}$ are calculated.
\item These distances are used to construct a point $\bb v_u$ in $k$-dimensional Euclidean space, where $\bb v_u$ is the apex of a simplex formed by its addition to the base simplex $\Sigma$, such that $\ell_2(\bb v_u, \bb v_i)=d(u,r_i)$ for all $i=1, \dots, k$
\end{enumerate}

These apex points form the target of the transform. The process therefore gives a mapping $\sigma$ from the general Hilbert space $\mathcal{U}$ to a $k$-dimensional Euclidean space, where $\sigma(u)=\bb v_u,\, \forall \, u \in \mathcal{U}$.
MATLAB and Java code for computing the nSimplex projection is available on GitHub at \url{https://bitbucket.org/richardconnor/metric-space-framework}\footnote{
Matlab implementation: \url{https://bitbucket.org/richardconnor/metric-space-framework/src/master/MATLAB/NSimplexProjection/}; 
Java implementation: \url{https://bitbucket.org/richardconnor/metric-space-framework/src/master/src/n_point_surrogate/}
}.
For Python users, the corresponding code implementation can be accessed at \url{https://github.com/vadicamo/nSimplex}.

As a simple example, Figure \ref{fig_simple_projection}  shows a projection from a 3D Euclidean space $\mathcal{S}$ to a 2D space. In this case the reference set $\mathcal{R}$ comprises two values $r_1$ and $r_2$  selected randomly from $\mathcal{S}$, and the base simplex formed is a line segment. This is arbitrarily embedded in the 2D target space with vertex coordinates $\bb v_1=[0,0]$ and $\bb v_2=[d(r_1,r_2),0]$. Every other value $u$ from $\mathcal{S}$ is then placed into the 2D projection according to its distances from these two reference values, therefore forming for each point an apex of a triangle whose base is the line (1D simplex) formed from $\bb v_1$ and $\bb v_2$.

\begin{figure}[tbp]
\centering
\begin{subfigure}{0.45\textwidth}
\centering
{\includegraphics[trim=0cm 0cm 0cm 1cm, clip,width=0.7\textwidth]{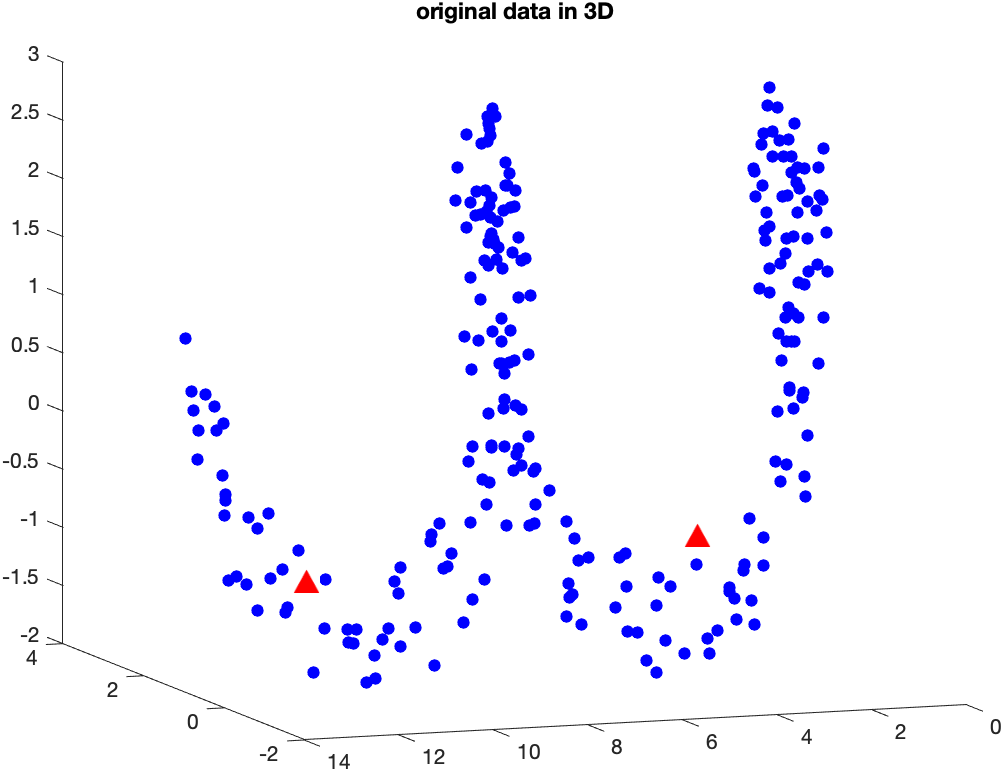}}
\caption{Original data in 3D}
\end{subfigure}
\begin{subfigure}{0.45\textwidth}
\centering
\includegraphics[trim=0cm 0cm 0cm 1cm, clip,width=0.7\textwidth]{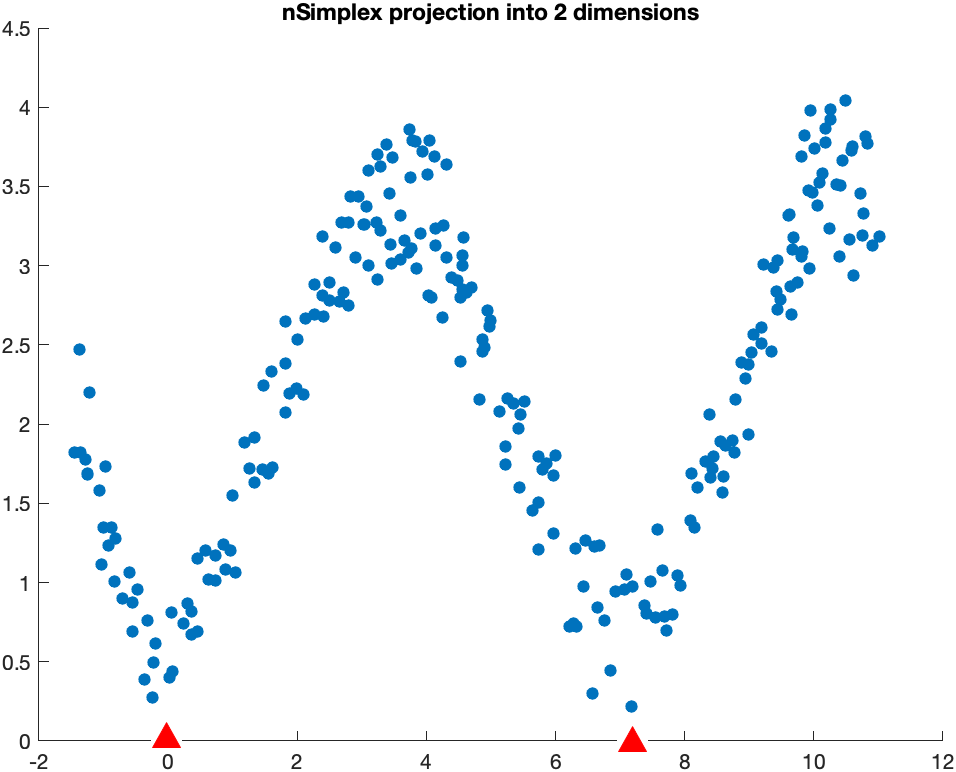}
\caption{\nsimp projection into 2 dimentions}
\end{subfigure}
\caption{Example projection from 3D to 2D using \nsimp. The left figure shows some generated points roughly in a 3D spiral pattern. Two of these points %(coloured in red) 
(depicted with red triangles) 
have been randomly selected to form the reference set $\mathcal{R}$. The right figure shows the 2D projection, formed over a 1D simplex derived from the distance between these points, whose vertices are shown in red. Each other point from the 3D set has been plotted at the apex of the triangle formed from its distances to these two points.}
\label{fig_simple_projection}
\end{figure}

In fact, not quite any set of values can be used for $\mathcal{R}$. The distances within the set must be able to form a set of linearly independent points in the projected space; in most spaces this is rarely an issue for a random selection. In fact the choice made for $\mathcal{R}$ affects various aspects of the projection, and will be discussed in detail later. It may also be noted that, in Step (4), there are two possible apex points that might be formed, one on either side of the hyperplane containing the base simplex. In our example, we have aligned the base (1-dimensional) simplex with the X axis, and selected apices with a positive Y coordinate. This choice can in fact be generalised over any number of dimensions, as the base simplex can always be formed with only zero values in the $k^\textit{th}$ dimension, i.e., it is constructed so that it lies in the hyperplane $\{[x_1, \dots, x_k]\in \R{k} | x_k=0\}$. We include in Appendix \ref{appendix_simplex_construction} an algorithm to construct a simplex with these properties in arbitrary dimensions. 

In this example, the 2D projection could be formed for any metric space, as the triangle inequality property means that it is always possible to construct the apex points of the triangles, i.e. determine a point $\bb v_u$ such that $\ell_2(\bb v_u, \bb v_1)=d(u,r_1)$ and $\ell_2(\bb v_u, \bb v_2)=d(u,r_2)$\footnote{The apex is in the intersection of a hypersphere centred in $\bb v_1$ with radius $d(u,r_1)$ and a hypersphere centred in $\bb v_2$ with radius $d(u,r_2)$. The intersection exists because $d(r_1,r_2)\leq d(u,r_1)+d(u,r_2)$.}. However as we will show the properties of the derived space are stronger if the domain of the transform has the Hilbert properties.

In due course we will define three functions over this $k$-dimensional coordinate space. First, however, we will introduce its important properties.

\subsection{Properties}
\label{sec_sub_properties}

For a Hilbert space $\mathcal{U}$ governed by a distance function $d$, we refer to $\sigma_\mathcal{R} : \mathcal{U} \rightarrow \mathbb{R}^k$ as an \nsimp transform defined by some appropriate set $\mathcal{R}$ of $k$ reference points selected from $\mathcal{U}$. For the sake of simplicity we henceforth use the notation $\sigma$ in place of $\sigma_\mathcal{R}$.
%We refer to the Euclidean distance over a $k$-dimensional space as $\ell_2^k$. 

The most important properties of the \nsimp transform are the following:

\begin{itemize}
\item $\sigma$ is a contraction mapping, i.e.
\[
\forall u_i,u_j \in \mathcal{U}, \quad \ell_2( \sigma(u_i),\sigma(u_j)) \quad \le \quad d(u_i,u_j)
\]
\item Over the same coordinate space, there exists a function \upb  which is an expansion mapping, i.e.
\[
\forall u_i,u_j \in \mathcal{U}, \quad d(u_i,u_j)  \quad \le \quad \textit{Upb}( \sigma(u_i),\sigma(u_j)) 
\]
\item Finally, there exists a further function \zen which gives a value between these two, i.e.
\[
\forall u_i,u_j \in \mathcal{U}, \quad \ell_2( \sigma(u_i),\sigma(u_j)) \quad \le \quad \textit{Zen}( \sigma(u_i),\sigma(u_j))  \quad \le \quad \textit{Upb}( \sigma(u_i),\sigma(u_j)) 
\]
\end{itemize}

It can be seen from these inequalities that the \zen function is a better estimator of the true distance than either $\ell_2$ or \upb. In fact, as we will show, the \zen function acts as an excellent estimator of true distance particularly when the original space is high dimensional, allowing good estimates to be made even when these are projected onto relatively low dimensions. To give better consistency of naming, we will henceforth refer to the $\ell_2$ function as \lwb when it is used in this context.

We first give definitions of the three functions, and will give a geometric explanation in the following section.
Let $\mathbb{R}^k$ be a space in the co-domain of some \nsimp transform $\sigma$.
Let the Euclidean coordinates of $\bb x, \bb y \in \mathbb{R}^k$ be given by $[x_1, x_2, \dots , x_k]$ and $[y_1, y_2, \dots , y_k]$ respectively.
Then
\begin{align}
 \textit{base\_dist}(\bb x, \bb y) &\quad=\quad \sum_{i=1}^{k-1} (x_i - y_i)^2 \\
\notag\\
\textit{Lwb}(\bb x, \bb y) &\quad=\quad \sqrt{ \textit{base\_dist}(\bb x, \bb y) + (x_k - y_k)^2} \label{eqn_lwb} \\
\notag\\
\textit{Upb}(\bb x, \bb y) &\quad=\quad \sqrt{ \textit{base\_dist}(\bb x, \bb y) + (x_k + y_k)^2} \label{eqn_upb} \\
\notag\\
\textit{Zen}(\bb x, \bb y) &\quad=\quad \sqrt{ \textit{base\_dist}(\bb x, \bb y) + x_k^2 + y_k^2} \label{eqn_zen}
\end{align}

Of the three functions, only \lwb is a proper metric. The others are not even semimetric, as for example they do not have the identity property: i.e. $Zen(\bb x, \bb x) \neq 0$  if the last vector component is non-zero. They do, however, all possess the triangle inequality property, and so are suitable for use with metric search techniques. In fact the lack of the identity property from the \zen function is actually a requirement for it to produce very good estimates when used in low dimensions.

Furthermore, it can be seen that the three functions can, if required, be evaluated efficiently as a triple, by observing that

\[lwb^2(\bb x, \bb y) + 2x_ky_k = zen^2(\bb x, \bb y) = upb^2(\bb x, \bb y) - 2x_ky_k \]

\subsection{Geometry of the Simplex}

The easiest introduction to the intuition of the \lwb and \upb functions is to consider first a projection into two dimensions. Although this is not  the primary  intended use, it is useful to illustrate principles that apply also more generally in higher dimensions with the simpler case.

\begin{figure}[tbp]
\begin{subfigure}{0.35\textwidth}
\includegraphics[width=\textwidth]{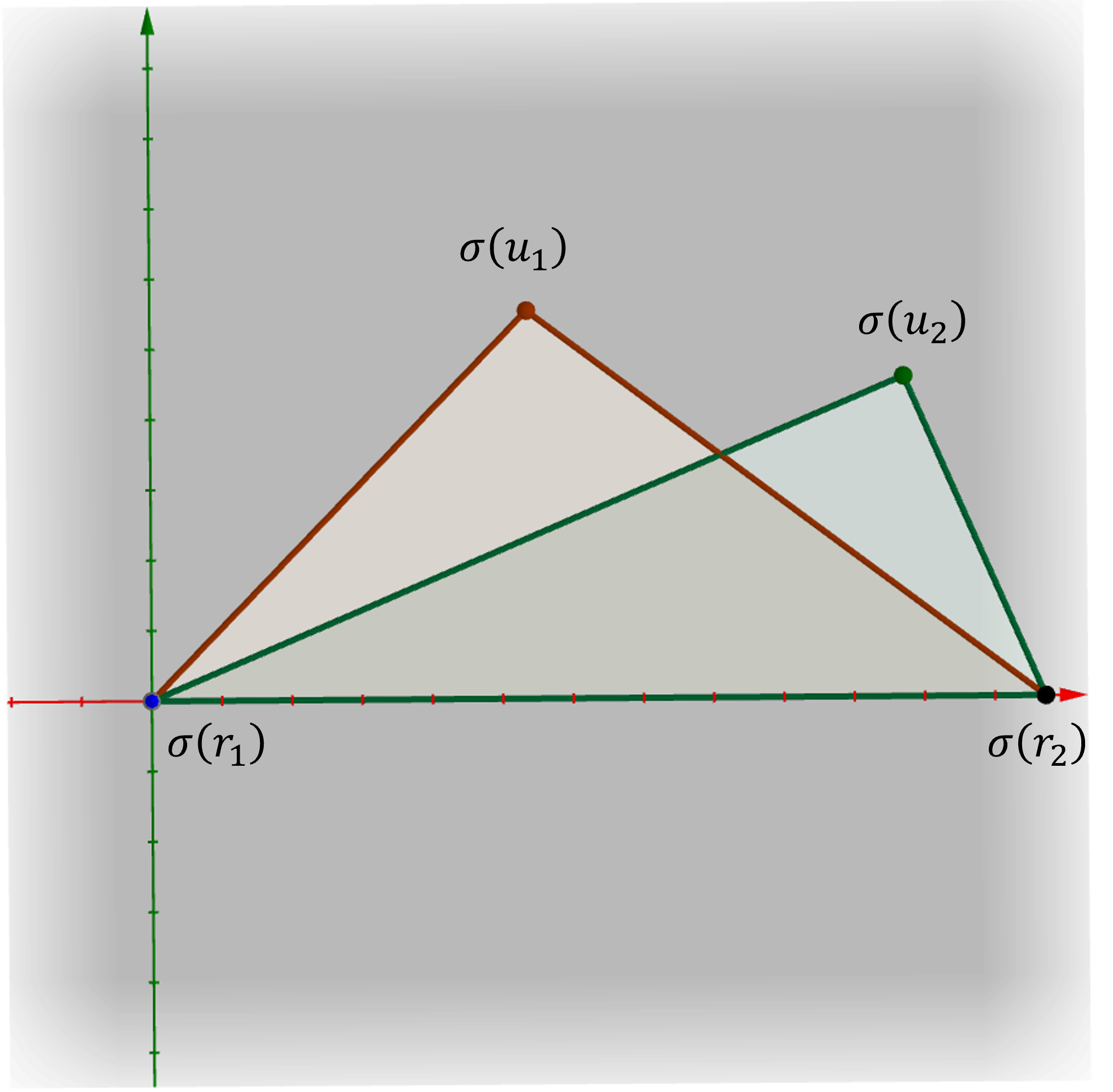} 
\bigskip\caption{ A 2D \nsimp projection of a space $(\mathcal{U},d)$. The projection $\sigma$ is constructed according to the reference objects $r_1,r_2 \in (\mathcal{U},d)$. The distances $d(u,r_1)$ and $d(u,r_2)$ give a unique position in the 2D plane for any element $u \in \mathcal{U}$.}
\label{fig_triangles_a}
\end{subfigure}
\hfill
\begin{subfigure}{0.55\textwidth}
\includegraphics[width=\textwidth]{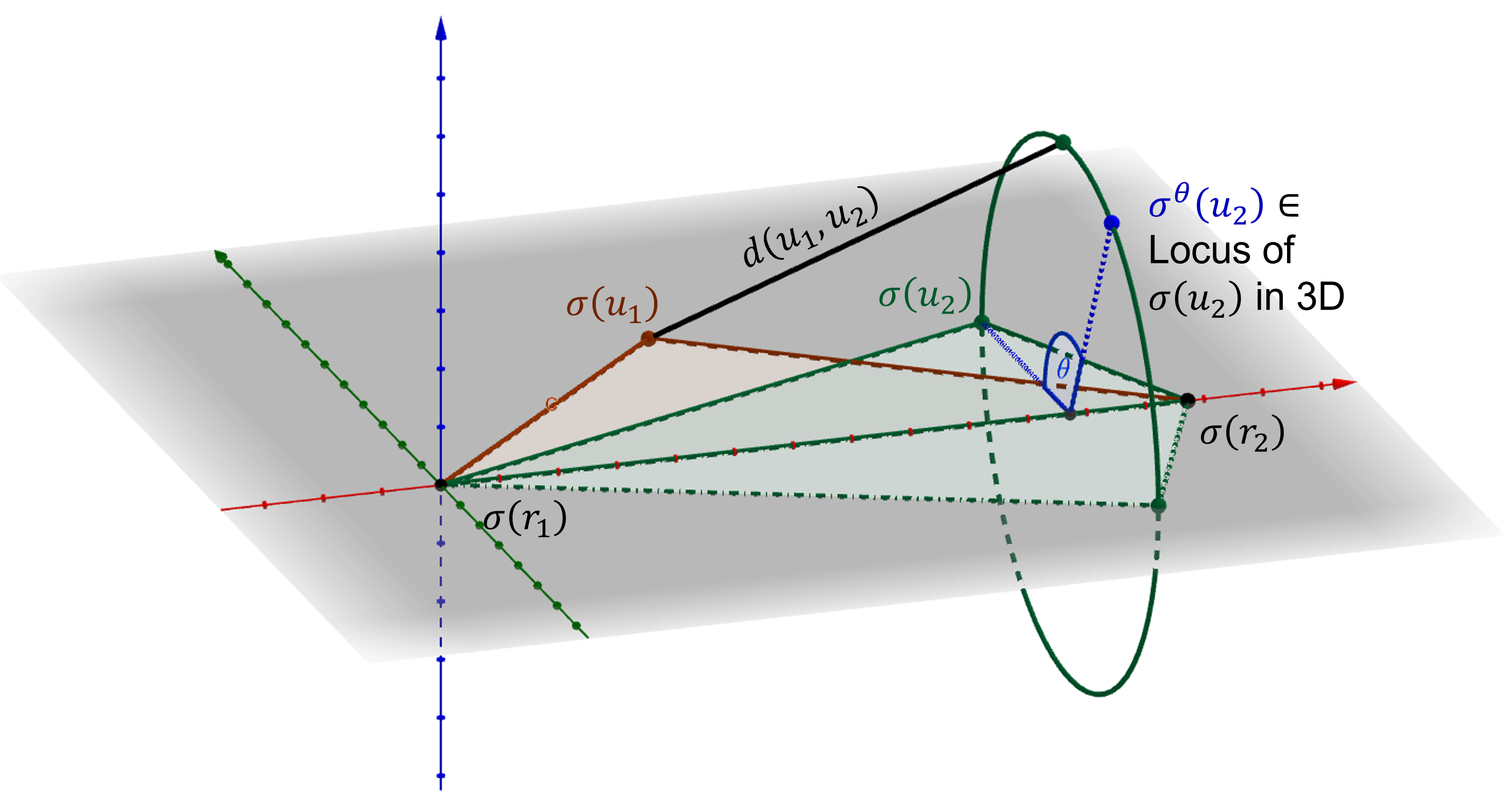}
\caption{From the Hilbert properties, any four objects from $(\mathcal{U},d)$ can be isometrically embedded in a 3D space. It can be seen that the distance $\ell_2(\sigma(u_1),\sigma(u_2))$ is a lower bound of the true distance $d(u_1,u_2)$.}
\label{fig_triangles_b}
\end{subfigure}
\caption{Two-dimensional projection of two values based on two reference objects (\ref{fig_triangles_a}), and the two possible planar tetrahedra formed by all four objects (\ref{fig_triangles_b}).}
\label{fig_triangles}
\end{figure}

Figure \ref{fig_triangles_a} shows two objects $u_1$ and $u_2$ from a (potentially high-dimensional) Hilbert space $(\mathcal{U},d)$ projected into two dimensions, using two reference objects $r_1$ and $r_2$. The reference objects are used to form the one-dimensional simplex comprising the line segment $[(0,0),(d(r_1,r_2),0)]$ and the objects $u_1$ and $u_2$ are projected into the 2D space, each as a separate apex of the base simplex formed by this line, according to their respective distances from $r_1$ and $r_2$. The points in the 2D space are thus the projections $\sigma(u_1)$ and $\sigma(u_2)$ of the \nsimp transform, where $\sigma$ is a mapping $\sigma : \mathcal{U} \rightarrow \mathbb{R}^2$ defined by the two reference points $r_1$ and $r_2$.

Due to the Hilbert properties of the projection domain, any $4$ values can be isometrically embedded in an $3$-dimensional Euclidean space therefore there must exist a tetrahedron in 3-dimensional Euclidean space whose vertices correspond to the four objects $\{r_1,r_2,u_1,u_2\}$, and whose edge lengths correspond to the distances between each corresponding pair. Two of the faces of this tetrahedron are congruent with the two triangles illustrated in Figure \ref{fig_triangles_a}. Considering only the 2D projection, 5 of the 6 inter-vertex distances have been calculated and are directly available from the projection (i.e., $d(r_1, r_2)$ and $d(r_i, u_j)$, $i,j=1,2$). The distance $d(u_1,u_2)$ is not available. Without loss of generality the vertices of the tetrahedron are $\{\sigma (r_1),\sigma(r_2),\sigma(u_1),\bb v_{u_2}\}$, where the vertex $\bb v_{u_2} \in \R{3}$ can be calculated only by explicitly computing $d(u_1,u_2)$.  However, it is possible to put upper and lower bounds on this distance from the tetrahedral geometry which is guaranteed to exist in a 3D projection.

Figure \ref{fig_triangles_b} shows a third dimension added to this diagram, which can accommodate the fourth unknown vertex $\bb v_{u_2}$ of the tetrahedron. Note that this is a  \emph{hypothetical} space, in that it is not explicitly constructed by the 2D \nsimp projection, but only used to reason about properties of the 2D projection. We introduce the term $\sigma^\theta(u)$ to refer to the mapping of an object $u \in \mathcal{U}$ into this $(\mathbb{R}^3,\ell_2)$ space so that $\ell_2(\sigma^\theta(u), \sigma(r_i)) =d(u, r_i)$, $i=1,2$, while still considering the 2D projection $\sigma$. The angle $\theta$ is the angle formed by the point $\sigma^\theta(u)$ and the hyperplane in $\R{3}$ containing the other three vertex of the tetrahedron.

The two adjacent faces of the tetrahedron share the line segment $[\sigma(r_1),\sigma(r_2)]$ as their common edge. Given that the tetrahedron must exist, it can be fully defined by the five available edge lengths in combination with the true angle $\theta^*$ between these faces, which must be somewhere in the interval $[0,\pi]$ radians. Without loss of generality, we fix the point $\sigma(u_1)$ in the $XY$ plane. The locus of the point $\sigma^\theta(u_2)$ in the higher dimension is thus restricted to the circle defined by the rotation of  apex point $\sigma(u_2)$  around the $X$ axis, and its exact location within the 3D space could be determined with knowledge of the distance $d(u_1,u_2)$, or equivalently from the knowledge of the exact angle $\theta^*$.

It is clear that the lower and upper bounds of the distance $d(u_1,u_2)$ in the original space occur with the planar tetrahedra formed when the angle $\theta$ is $0$ and $\pi$ radians respectively. These planar tetrahedra are contained within the 2D space of the original projection $\sigma$, and the 3-dimensional model does not need to be explicitly formed in order to establish their geometry. The tetrahedron defined by the angle $\theta=0$ has vertices exactly as already projected. Due to the manner in which the projection is constructed, with the final coordinate of the projection representing the altitude of the apex point above the hyperplane containing the base simplex (see Appendix \ref{appendix_simplex_construction} for full details), the tetrahedron defined by the angle $\theta=\pi$ can be created simply by taking the negative value of the $Y$ coordinate of point $\sigma(u_2)$.
These observations lead directly to the derivation of the \lwb and \upb functions as defined in Equations \ref{eqn_lwb} and \ref{eqn_upb} respectively.
It should be noted that these bounds apply only to projections made from general Hilbert spaces.

As noted above, any metric space can be projected into two dimensions, as the ability to perform this mapping is guaranteed by the triangle inequality property. However, the lower and upper bound properties do not hold for the 2D projection unless a stronger condition, the ability to isometrically embed any four objects into 3D Euclidean space%
\footnote{This is the so-called \emph{four-point property}; it is slightly more general than Hilbert properties, and some useful  non-Hilbert metric spaces possess this. We have previously defined such spaces as \emph{supermetric} \cite{Connor2016:IS_Supermetric}, and shown how general metric search techniques can be improved through its use.}, also holds in the domain of the projection.

While the intuitive argument given is valid only for the two-dimensional projection, it carries through a projection into any number of dimensions, as the Hilbert properties give the ability to isometrically embed any $k$ objects into $(k-1)$ Euclidean dimensions. It is possible, for example when $k=3$, to rotate the apex of a tetrahedron through a fourth dimension, around the plane containing its triangular base, whilst preserving the edge lengths, but this is not so clear in terms of intuition.
The general result as stated above, that the \lwb and \upb functions given in Section \ref{sec_sub_properties} are lower and upper bounds respectively of the true distance, is independent of the dimension of the projection when applied to any Hilbert space.
We enclose a  proof of correctness of this result in Appendix \ref{appendix:ProofCorrectness}.

\subsection{The \zen function}
\label{sec_zen_function}

Figure \ref{fig_rotating_triangles_a} is an illustration of the same 2D projection as in Figure \ref{fig_triangles}, but  shows the triangle $\Delta \sigma(r_1)\sigma(r_2)\sigma^\theta(u_2)$ with a different orientation in the hypothetical 3D space, while Figure \ref{fig_rotating_triangles_b}  shows the case where this triangle is set at the angle $\pi/2$ with respect to %triangle $\Delta \sigma(r_1)\sigma(r_2)\sigma(u_1)$. 
the hyperplane $H$ containing $\sigma(r_1)$, $\sigma(r_2)$, $\sigma(u_1)$.
The \zen (zenith) function is named after this last orientation, and gives the $\ell_2$ distance between the points $\sigma(u_1)$ and $\sigma^\theta(u_2)$ when $\sigma^\theta(u_2)$ is at the zenith of this circle, i.e. the point with the highest altitude above  the hyperplane $H$. %the plane containing $\sigma(u_1)$. 
This distance can be simply calculated using only the information in the projection, as given in Equation \ref{eqn_zen}. In this section, we explain why this function provides the best estimator for the true distance $d(u_1,u_2)$ in an original high-dimensional space.

\begin{figure}[tbp]
\begin{subfigure}{0.49\textwidth}
\includegraphics[width=\textwidth]{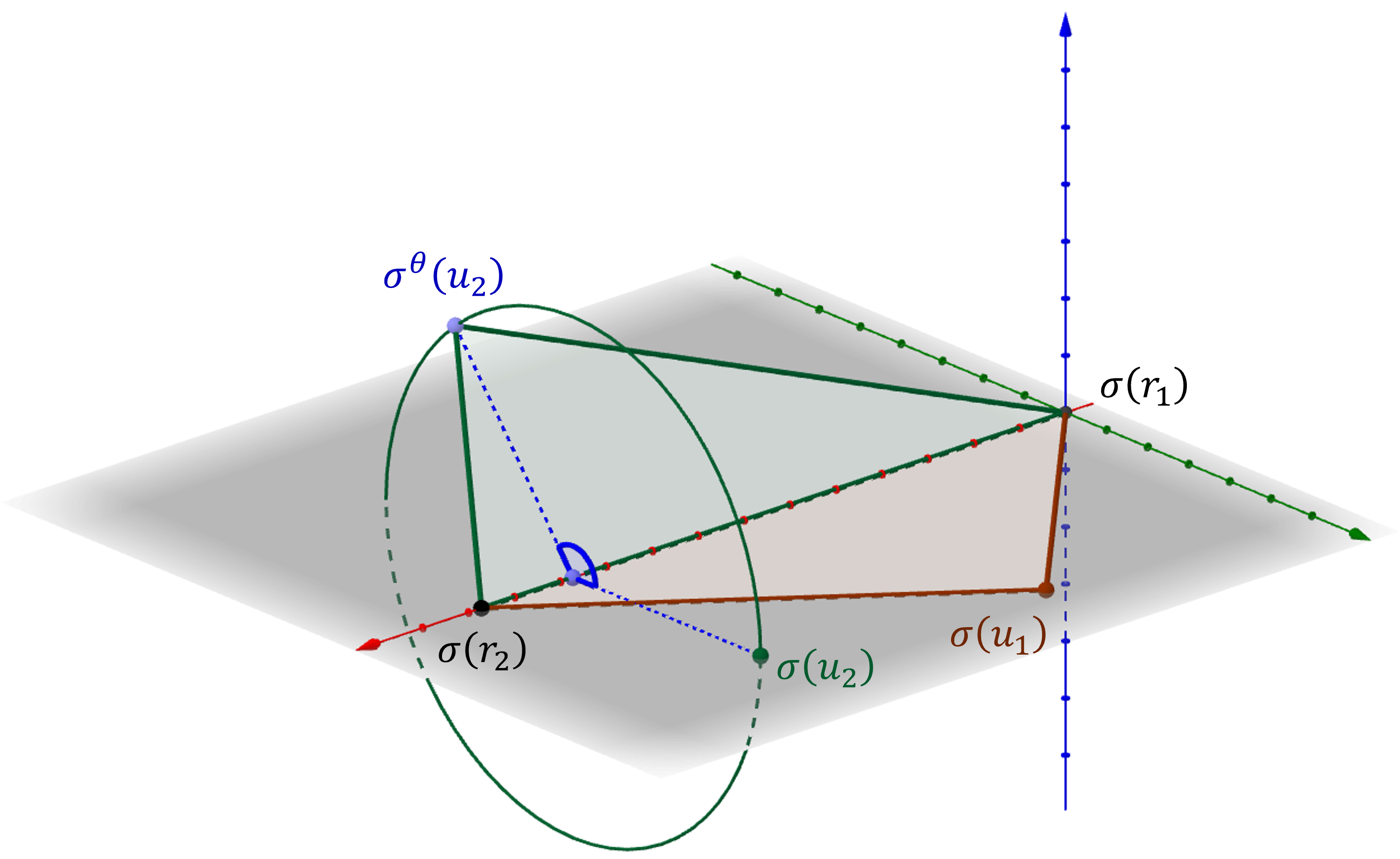}
\caption{Four objects $r_1,r_2,u_1$ and $u_2$ are projected into two dimensions. $\sigma^\theta(u_2)$ is a hypothetical position in 3D space.}
\label{fig_rotating_triangles_a}
\end{subfigure}
\hfill
\begin{subfigure}{0.49\textwidth}
\includegraphics[width=\textwidth]{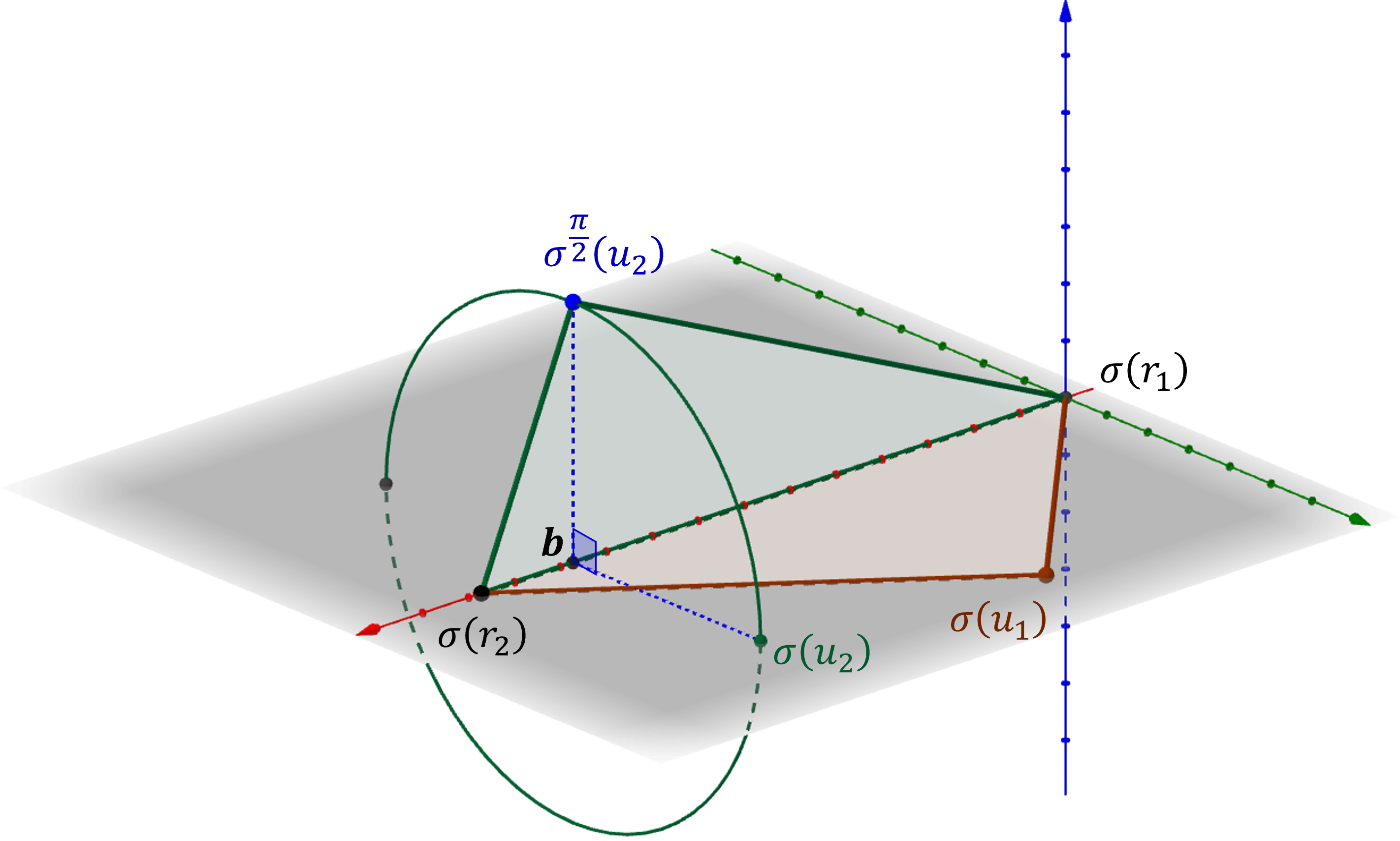}
\caption{Adding $\bb b$, the centre of the locus of $\sigma^\theta(u_2)$, and setting the angle  $\theta=\angle \sigma(u_2)\bb b\sigma^\theta(u_2)$ to $\pi/2$.}
\label{fig_rotating_triangles_b}
\end{subfigure}
\caption{The \zen function is defined when the angle between the two triangles is set at $\pi/2$ in the \emph{hypothetical} further dimension. There is no requirement to calculate a projection in this dimension:  $\textit{Zen}(\sigma(u_1),\sigma(u_2)) = \ell_2(\sigma(u_1),\sigma^\theta(u_2))$.}
\label{fig_rotating_triangles}
\end{figure}
 
Considering Figure \ref{fig_rotating_triangles} again, and all of the possible tetrahedra that could be formed from the five known and the one unknown distances, it might be supposed that there is no constraint on the particular position in which the point $\sigma^\theta(u_2)$ is most likely to lie on the circle depicted. This assumption however typifies the danger of basing intuition on low-dimensional spaces. In fact there is no absolute constraint, but there is a very significant probabilistic constraint, assuming the domain is evenly distributed, and this gets tighter as the \emph{dimensionality} of the {domain} increases.

% The right hand side of Figure \ref{fig_rotating_triangles} shows the 3D projection more from the perspective of this circle. 

\subsubsection{Considering higher dimensions}

%In the hypothetical 3D space deriving from the 2D projection, the locus of this point defines a simple circle.
%More generally, 
If the projection is onto a $k$-dimensional space, %where $k > 2$, 
then the hypothetical space being considered is in $k+1$ dimensions. 
{
In that case, $k$ reference objects and any two data points $u_1$ and $u_2$ are projected in $\R{k}$ using the \nsimp projection $\sigma$.
The vertex $\sigma^\theta(u_2)\in \R{k+1}$ is obtained by rotating $\sigma(u_2)$ around the $k-2$ dimensional space containing the base simplex formed by $\{ \sigma (r_1), \dots, \sigma(r_k)\}$; the angle $\theta$ is the angle formed by $\sigma^\theta(u_2)$ and the hyperplane $H=\{[x_1, \dots, x_{k+1}]\in \R{k+1} \,  |\, x_{k+1}=0\}$ containing both $\sigma(u_1)$ and the simplex base. In other words, 
$\sigma^\theta(u_2)$ lies in the intersection of  $k$ hyperspheres  $B_i=\{\bb v\in\R{k+1}| \ell_2(\bb v,\sigma(r_i))=d(u_2,r_i)\}$  for $i=1, \dots, k$, that forms a circle on a plane orthogonal to the hyperplane containing the base simplex of vertices $\sigma(r_1), \dots, \sigma(r_k)$ and the projected point $\sigma(u_1)$.  The exact angle $\theta^*$ that would give $\ell_2(\sigma(u_1), \sigma^{\theta^*}(u_2))=d(u_1, u_2)$ it is not known without explicitly calculating  $d(u_1, u_2)$.  However, as shown in Section \ref{sec_intro_subsec_related_angles}, in high dimensional space the most likely value for this angle is $\pi / 2$, and furthermore, as the dimension of the space increases, the variance rapidly decreases. This variance is a factor of the dimensionality of the domain, rather than the range, of the projection.
}
With this angle set to $\pi / 2$, the distance $\ell_2(\sigma(u_1),\sigma^\theta(u_2))$ in the hypothetical further dimension can be simply calculated given the projection values 
$\sigma(u_1)$ and $\sigma(u_2)$ in the $k$-dimensional space of the projection. This finally gives the explanation of the \zen formula (Equation \ref{eqn_zen}) which gives this distance in the context of the projection from $m$ to $k$ dimensions.

In Appendix \ref{appendix:ProofCorrectness} we give a formal derivation of this intuitive argument in arbitrary Hilbert spaces. In particular, we show that  if $\sigma: \mathcal{U}\to \mathbb{R}^k$ is the \nsimp transform defined by a set of $k$ reference points then for any $u_1, u_2 \in \mathcal{U}$ given the transformed points $\bb x=\sigma(u_1)$ and $\bb y=\sigma(u_2)$ it holds
	\begin{equation}
  d(u_1,u_2)= \sqrt{ \sum_{i=1}^{k-1}(x_i-y_i)^2+x_k^2+y_k^2-2x_ky_{k} \cos \theta}
  \label{eqn_zen_with_theta}
  \end{equation}
where $\theta$ corresponds to the angle $\angle \sigma(u_2) \bb b \sigma^\theta(u_2) $ in Figure  \ref{fig_rotating_triangles_b}. It is clear from this form that as the probability of $\theta$ being close to $\pi / 2$ increases (as happens when the dimensionality of domain $\mathcal{U}$ increases)  the \zen function applied to the $\sigma$ projection gives an increasingly accurate estimate of the true distance $d(u_1,u_2)$.

There is one caveat here however. If $d(u_1,u_2)$ is very small, this can affect the probability of  $\theta$ being close to $\pi / 2$. Considering Figure \ref{fig_rotating_triangles} it can be seen that if $d(u_1,u_2)$ is close to the lower-bound, shown in the diagram by the distance $d(\sigma(u_1),\sigma(u_2))$, then $\theta$ will be significantly less than $\pi / 2$. Equation \ref{eqn_zen_with_theta} shows that, for small values of $\theta$, the \zen function will not be such a good estimator.

The implications of  the general result however are quite extraordinary: it is possible to compress a space of perhaps thousands of dimensions into a very low-dimensional space of only a few dimensions, where the majority of pairwise distances are well-preserved. We demonstrate that this is in fact the case in Section \ref{sec_experiments}.

% The above argument does not require to take into account that the projected value $A$ is likely to lie outside the hyperplane containing the $(k-1)$-sphere. This does not alter the geometric argument that the zenith angle gives the best estimator of the distance $d(A,B)$, but in fact causes a decrease in the associated variance of the true angle. If the 2D projection is considered, then the view from $A$ onto the circle in the 3D hypothetical space is an ellipse,

\section{Experimental Analysis}
\label{sec_experiments}
Experimental analysis is presented in four main sections, each of which tests dimension reduction over  a different class of metric space. Section \ref{sec_sub_exp_gen_euc} tests the different  transforms against uniformly generated Euclidean spaces, and Section \ref{sec_sub_exp_real_euc} uses two high-dimensional Euclidean spaces deriving from real-world applications. Section \ref{sec_sub_exp_cos}  tests two spaces governed by the Cosine metric, and Section \ref{sec_sub_exp_jsd} tests two spaces governed by the Jensen-Shannon metric.

For the first three of these sections, the mechanisms tested are: \nsimp \zen, PCA, MDS and RP. 
 For Jensen-Shannon distance, where there is  no coordinate space, the  mechanisms tested are  \nsimp \zen and LMDS.

First, in Sections \ref{subsec_quality_measurement} and \ref{subsec_test_data} respectively the quality measures and  data sets used are introduced.

All of the code used to generate our experimental results is available from \url{https://github.com/richardconnor/dr-matlab-code}.

\subsection{Quality Measurement}
\label{subsec_quality_measurement}
Dimensionality reduction is a very generic concept, defining any mechanism whose purpose is to transform a set of values into a lower-dimensional space whilst maintaining, as far as possible, the most important aspects of the geometry of the original space. This rather general definition leaves much room for the interpretation of quality, depending on the context of use. A comprehensive survey of quality measurement techniques is given in \cite{GRACIA20141}; based on this, we have picked the following measures as the most representative for the general context.

%In the following, we  adopt the notation used in \cite{GRACIA20141}:  for 
For a  space $(\mathcal{S},d)$ which has been reduced to a lower-dimensional space $(\mathcal{S}',\zeta)$ using a DR transform $\T$, we adopt the following notation and measures
 \begin{align*}
\delta_{ij} \quad &= \quad d(s_i,s_j)\quad s_i,s_j \in (\mathcal{S},d)  \\
\zeta_{ij} \quad &= \quad \zeta({s_i'},{s_j'})\quad s_i',s_j' \in (\mathcal{S}',\zeta) \quad \text{where } {s_i'}=\T(s_i) \, \forall i
\end{align*}
 %where ${s_i'}=\T(s_i)$ for all $i$.

\begin{description}
\item[Shepard Plots]  A scatter plot of sampled distances $\delta_{ij}$ from the domain, plotted against distances $\zeta_{ij}$, % from the range of the function, 
which gives a simple visual impression of quality. Plots are typically overlaid with the monotonic function implied by the Kruskal Stress measurement.
\item[Kruskal Stress] The Kruskal \emph{stress1} criterion, which gives a measure of the monotonicity of the transform. This is a topological measure; stress will be zero if the DR transform is purely monotonic, independent of the actual values of $\delta_{ij}$ and $\zeta_{ij}$.
\item[Sammon Stress]  Deriving from Sammon mapping, the Sammon stress formula is affected by the absolute differences between $\delta_{ij}$ and $\zeta_{ij}$, as well as their topological relationship.
\item[Quadratic Loss] A purely distance-based measure, which particularly punishes the existence of outliers in $\delta_{ij} - \zeta_{ij}$.
\item[Spearman Rho] A topological measure of order preservation of distances within sampled pairs of objects from the domain, essentially a measure of the likelihood that  $\delta_{ij}< \delta_{ik}$ implies  $\zeta_{ij}< \zeta_{ik}$.
\item[kNN Query Recall] Here the results of $k$NN searches in the reduced space are tested for quality against the same search performed in the original space.
This aspect is not measured in \cite{GRACIA20141}, and we are not aware of any commonly accepted measure for testing it. Nonetheless it seems that nearest-neighbour search over the reduced space is an important use of these techniques. It is not captured by any of the quality metrics listed above, as behaviour over very small distances may differ from randomly sampled distances. We have therefore devised our own measure of recall, described in Appendix \ref{sec_intro_subsec_quality}, where discounted cumulative gain is measured over a relevance function based on rank.

\end{description}
Appendix \ref{sec_intro_subsec_quality} gives fuller background on  all of these measurements.

\subsection{Test data and methodology}
\label{subsec_test_data}
\begin{table}[tbp]
\caption{Data sets used in experiments and their outline properties. \emph{Cosine} distance refers to the $\ell_2$ metric applied over $L_2-$normalised data, and for Jensen-Shannon distance the data is $L_1$-normalised as required.}
\begin{center}
\small
\begin{tabular}{p{4cm}p{3cm}p{3cm}p{2.5cm}}
\hline
\textbf{Data Set}		&\textbf{Representational} \newline \textbf{Dimension}&\textbf{Metric}&\textbf{Dimension of } \newline \textbf{80\%  variance}\\
\hline
100-dimensional  generated	&100		&Euclidean&80\\
500-dimensional  generated	&500		&Euclidean&400\\
Twitter GloVe					&200		&Euclidean&120\\
MirFlickr fc6					&4096		&Euclidean&109\\
ANN SIFT						&128		&Cosine&28\\
MirFlickr fc6 RELU			&4096		&Cosine	&1111\\
100-dimensional generated	&100		&Jensen-Shannon&n/a\\
MirFlickr GIST					&480		&Jensen-Shannon&n/a\\
\hline
\end{tabular}
\end{center}
\label{tab_dataset_overview}
\end{table}%
In all cases we have used data sets that are  widely available, or can be recreated using widely available software, and have at least one million elements to allow reasonable recall experiments. The data sets used and their main features are given in Table \ref{tab_dataset_overview}. Details are given in Appendix \ref{appendix_data_sets}.

In all experiments, a randomly selected subset of objects from the domain is used as a \emph{witness}%
\footnote{or \emph{training} set; we prefer the term \emph{witness} in this context  to denote  a relatively small representative subset.}.
 set with which to create the transforms. The witness set should be sufficiently large for the initial analysis of any  manifold within which the actual data set is embedded, depending on the technique being considered, to allow the  general $(\mathcal{U} \rightarrow \mathbb{R}^k)$ transform to be created. The RP transform is created without reference to the domain, and  the \nsimp transform is created from a  set of $k$ objects randomly selected from the witness set.

For the majority of quality tests a  further (non-intersecting) subset of objects is used as the domain of the transform. For Shepard plots,  a subset of just 50 objects
% (giving $50 \choose 2$ plotted pairs of distances)
  is used to avoid  overcrowding the plot. A set of $10^4$ objects  is used to calculate the  Kruskal stress used to annotate the plot, and the other quality measures other than recall.  For recall experiments, a set of $10^6$ elements is used, against which the ground truth of 1,000 nearest neighbours is  calculated for 100 elements of this subset. All experiments have been multiply repeated with different random selections to ensure the results shown are representative and repeatable.
  
\begin{figure}[!tp]
\centering
\begin{subfigure}{0.34\textwidth}
\includegraphics[width=\textwidth]{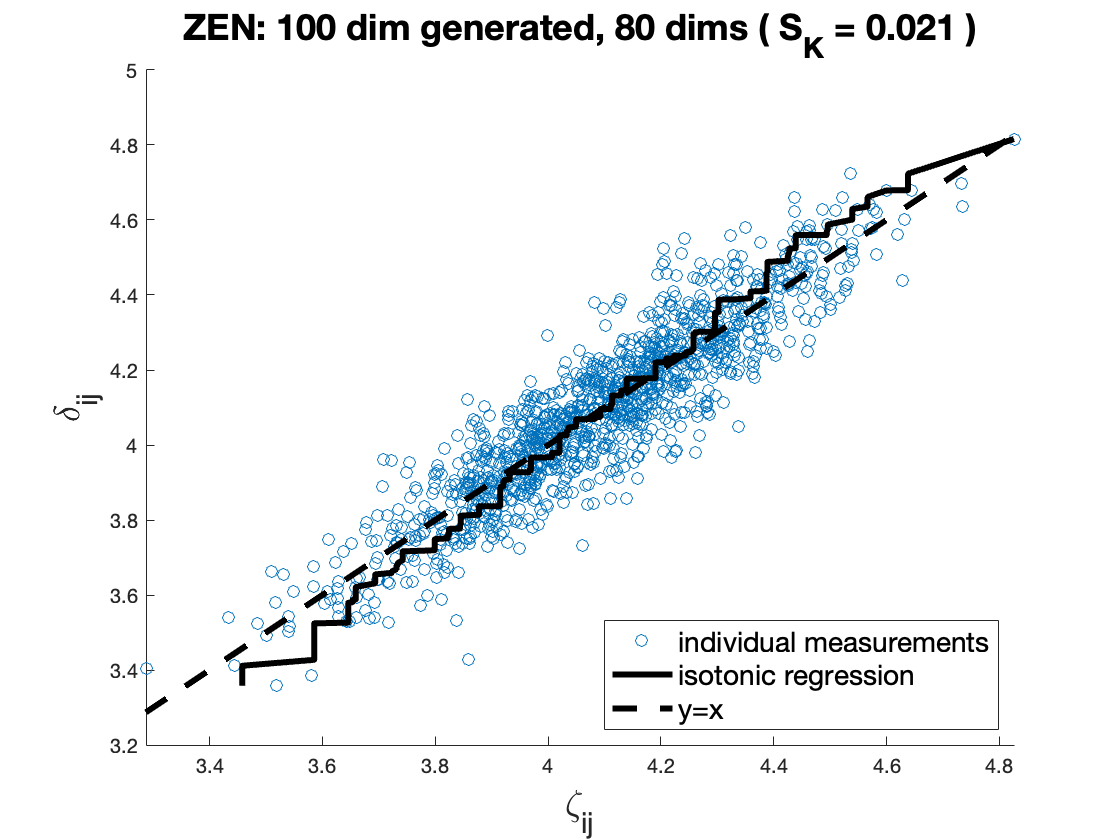}\hfill
\caption{\nsimp \zen}
\end{subfigure}
\begin{subfigure}{0.34\textwidth}
\includegraphics[width=\textwidth]{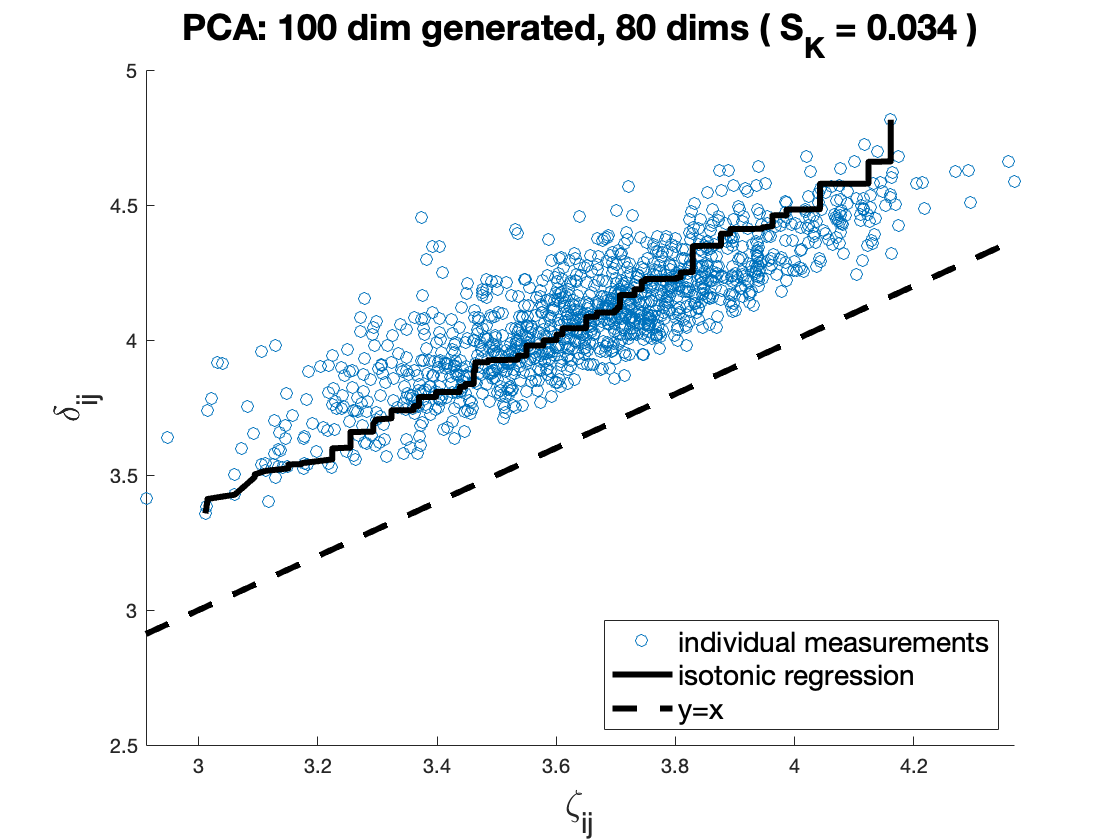}
\caption{PCA}
\end{subfigure}\\
%second row of diagrams
\begin{subfigure}{0.34\textwidth}
\includegraphics[width=\textwidth]{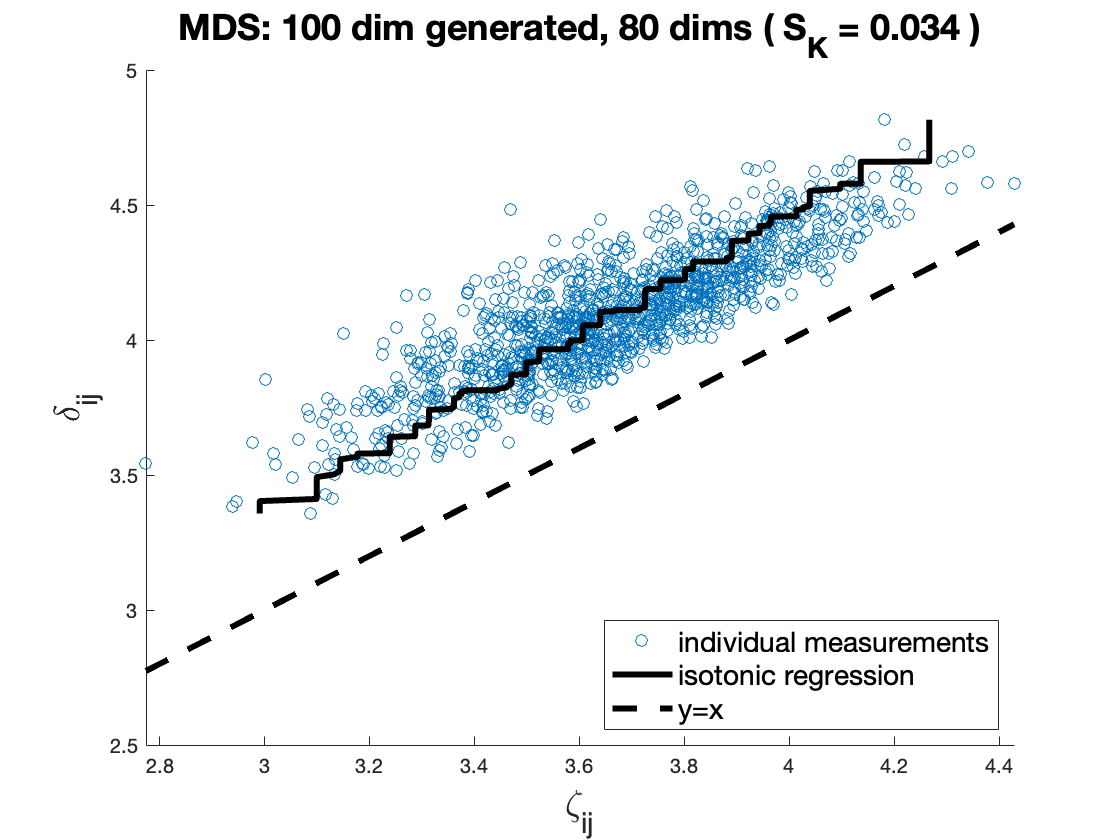}\hfill
\caption{MDS}
\end{subfigure} 
\begin{subfigure}{0.34\textwidth}
\includegraphics[width=\textwidth]{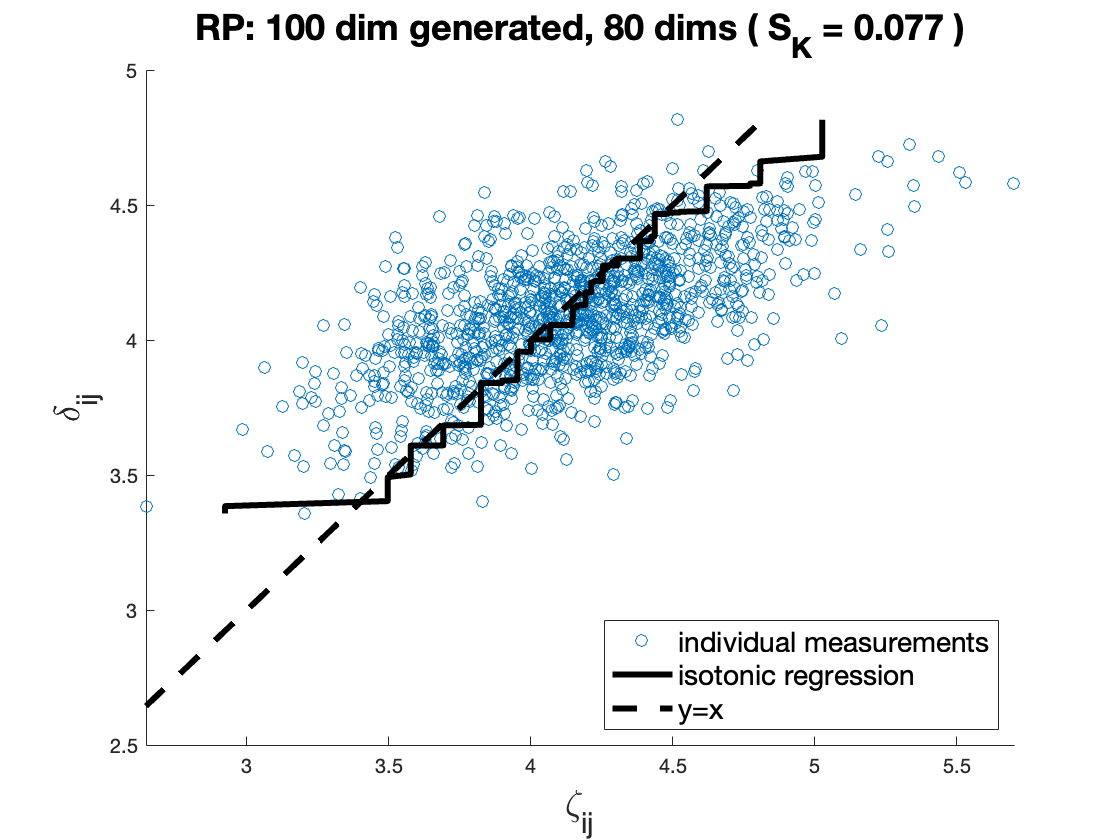}
\caption{RP}
\end{subfigure}
\caption{Shepard plots for the  reduction transforms, each having  reduced 100-dimensional generated data to 80 dimensions. For 50 randomly selected values, all pairwise distances are plotted in both original and transformed spaces. The Y-axis represents true distance, and the X-axis is the distance measured in the reduced space. The solid black line shows the fitted least-squares monotonic regression function from which the Kruskal stress ($S_K$) is measured. It can be seen that \nsimp \zen and RP point clouds are centred around the true distance function ($y =x$, the dashed line), whereas PCA is a contraction mapping. While MDS gives the appearance of a contraction mapping, in fact this is not a guarantee. 
%Kruskal stress is measured over the pairwise distances among a sample of 1,000 values, only a fraction of which is plotted in the charts to avoid overcrowding.
}
\label{fig_10dim_shepards}
\end{figure}
\subsection{Generated Euclidean spaces}
\label{sec_sub_exp_gen_euc}

%{\color{blue}maybe change this back to give 30 and 100 dim, with 500 in an appendix, showing the effect of higher dimensions making zen better; 100 and 500 don't really have a noticeable difference...}

Uniformly distributed Euclidean spaces%
\footnote{The experiments have been repeated for generated data with a Gaussian distribution, the results are not significantly different from those shown here.}
 were generated in 100 and 500 dimensions. Reduced-dimension versions were produced using RP, PCA, MDS, and \nsimp, %the four dimensionality reduction techniques introduced in Section \ref{sec_intro_subsec_dim_red}, 
 and tested using the quality measures outlined in Section \ref{subsec_quality_measurement}. %, described in detail in Appendix \ref{sec_intro_subsec_quality}.

In the case of generated data, the witness set contains no useful information about the data, as there is no lower-dimensional manifold contained within the representational space. Both PCA and MDS therefore  effectively apply a random projection to the experimental data. The PCA transform is guaranteed to be orthonormal, and while the MDS transform is not, it is always  close to this given a uniform distribution of the witness data. The RP technique used in these experiments is much further from orthonormal, especially with  lower dimensionalities.
%, and will not perform as well as either PCA or MDS over uniform data. It is however substantially faster, which is the main purpose of the technique described in Section \ref{sec_intro_subsec_dim_red}.

 It is generally perceived that there is  no value in applying non-random dimensionality reduction to  uniformly distributed data, but these  experiments  demonstrate that the $Zen$ function preserves distances better than the other methods, due to the geometric model  described in Section \ref{sec_zen_function}, even in the absence of a lower-dimensional manifold.

\subsubsection{100 dimensional generated space}
\begin{figure}[tbp]
\centering
\begin{subfigure}{0.32\textwidth}
\includegraphics[width=\textwidth]{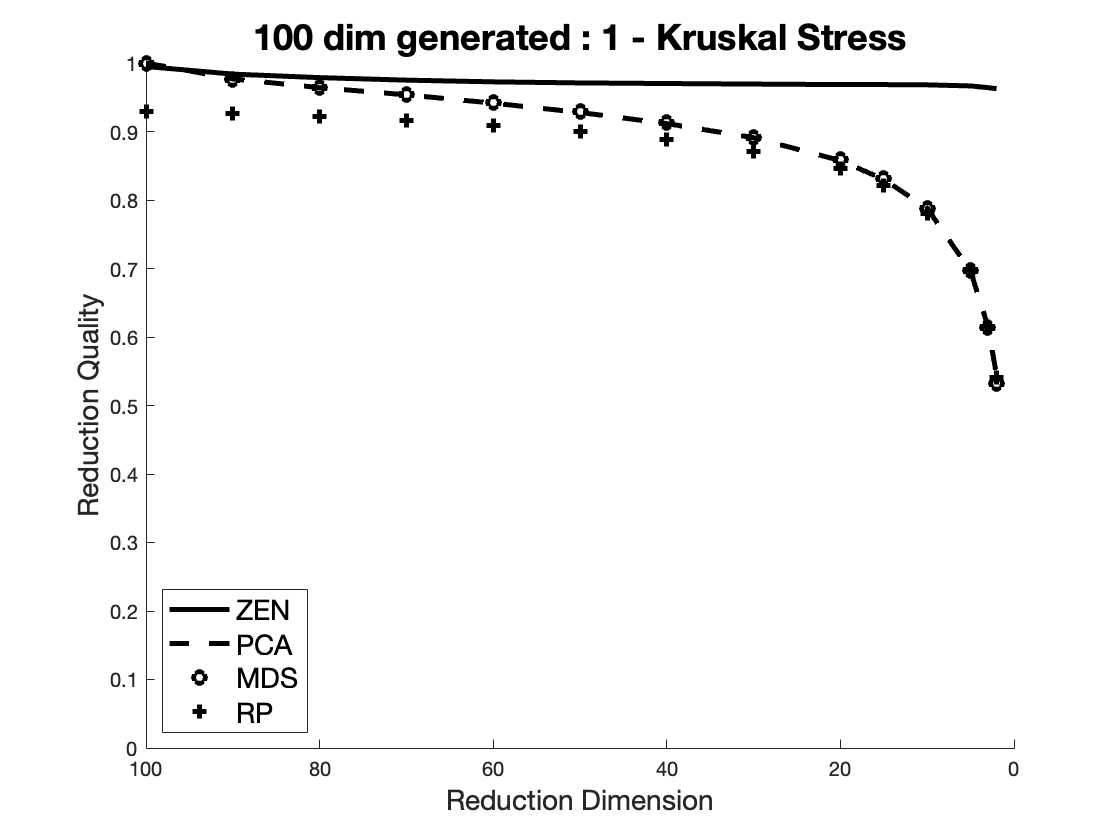}\hfill
\caption{Kruskal stress}
\end{subfigure}
\begin{subfigure}{0.32\textwidth}
\includegraphics[width=\textwidth]{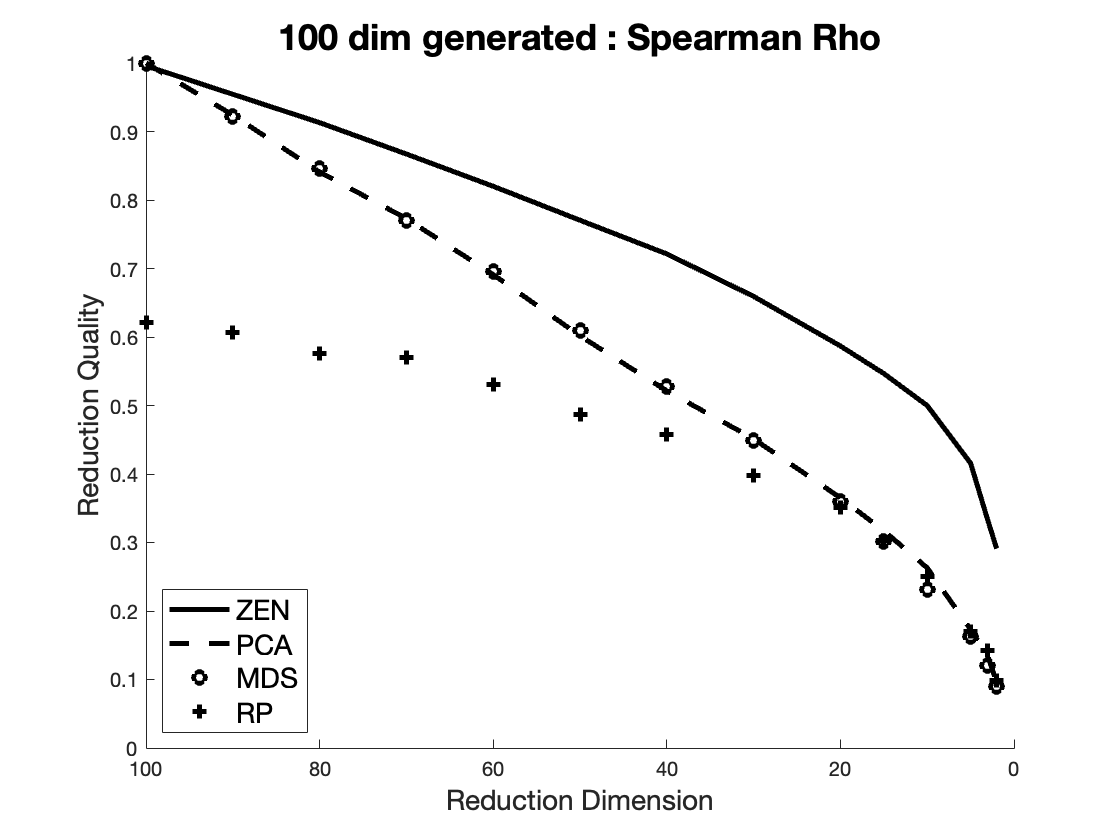}
\caption{Spearman Rho}
\end{subfigure}
\begin{subfigure}{0.32\textwidth}
\includegraphics[width=\textwidth]{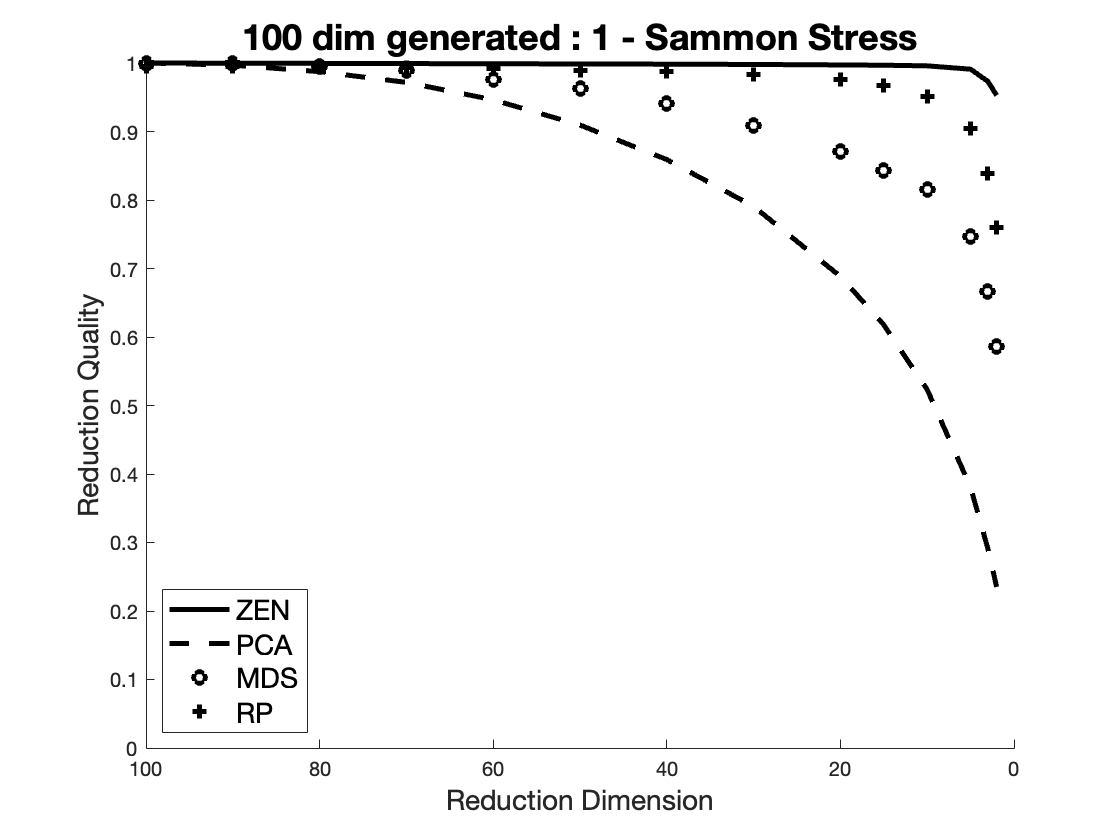}\hfill
\caption{1- Sammon stress}
\end{subfigure}
\begin{subfigure}{0.32\textwidth}
\includegraphics[width=\textwidth]{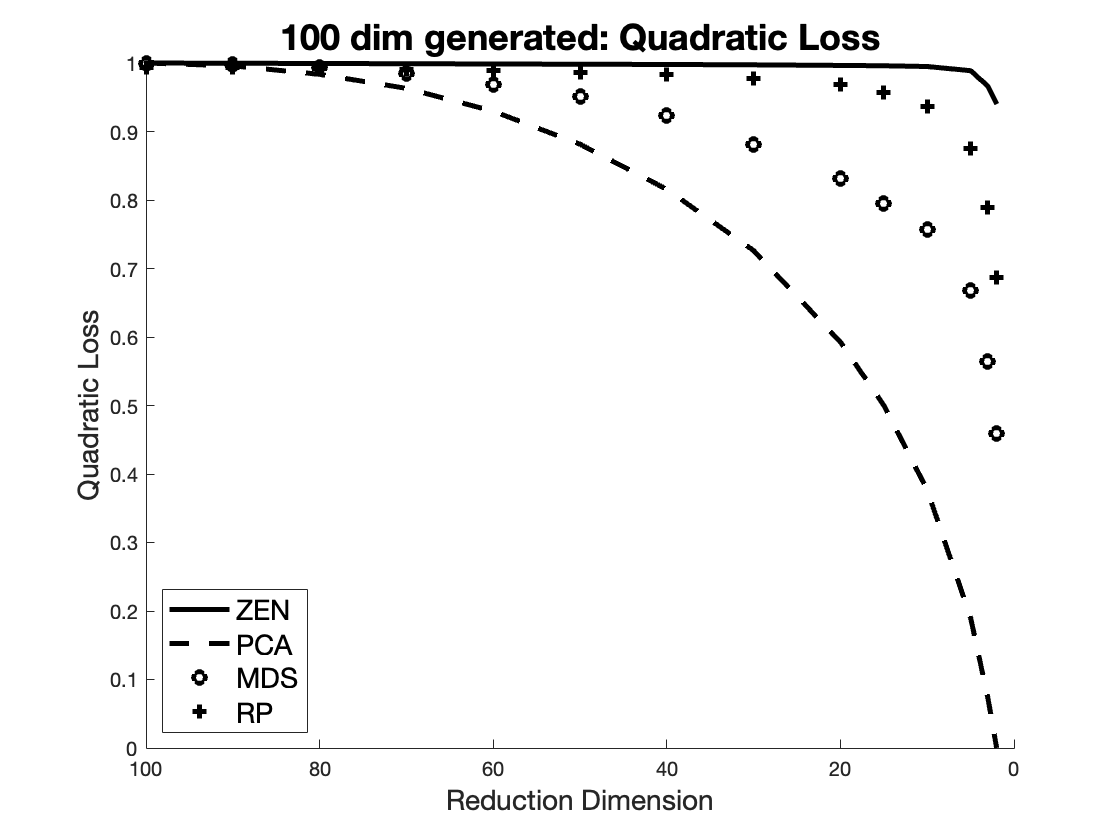}
\caption{Quadratic loss}
\end{subfigure}
\begin{subfigure}{0.32\textwidth}
\includegraphics[width=\textwidth]{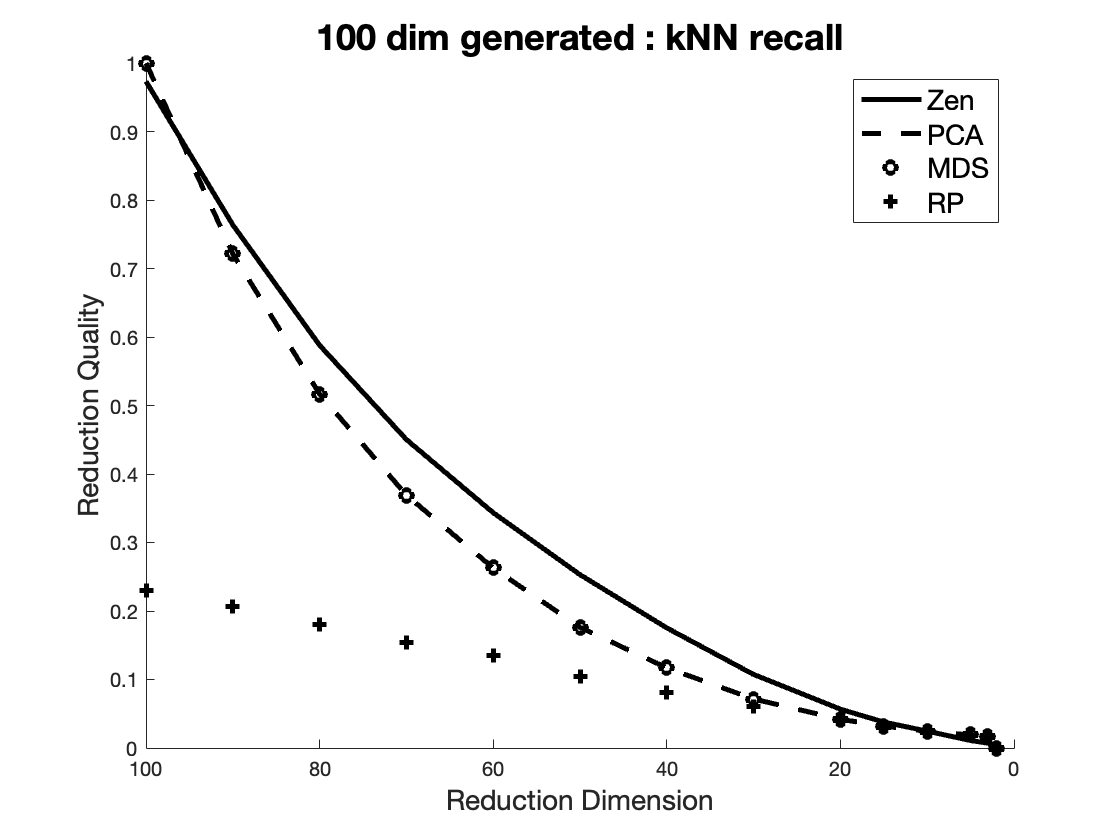}
\caption{kNN Recall}
\end{subfigure}
\caption{Five measures of quality as the reduction dimension decreases. The X-axis in each plot is the dimension of the reduction used, in this case starting from 100 on the left and ending at 2 on the right. The Y-axis shows the measure of quality, for each transform, at each dimension. All quality measures are normalised into $[0,1]$ to make the comparisons clearer, as recommended in \cite{GRACIA20141}. For all measures, a value of 1 implies a perfect representation of the original space, a value of 0 means the transform has no effective value.}
\label{fig_100dim_quality}
\end{figure}

Figure  \ref{fig_10dim_shepards} shows Shepard plots for the 100-dimensional data reduced to 80 dimensions using various reduction transforms. The reduction to 80 dimensions has been chosen as it represents the number of dimensions that explain 80\% of the variance using PCA analysis and Eq. \eqref{eqn_pca_variance}.
% The top row of the figure gives \zen, PCA and MDS respectively, the main comparisons of interest. The lower row gives RP along with the \nsimp $lwb$ and $upb$ functions.

As can be seen, \nsimp \zen is the best transform according to the Kruskal stress criterion, giving a significantly better outcome than either PCA or MDS. 
%While the spread of the point clouds appears quite similar, the more telling visual aspect is the relative gradients of the monotonic regression function fitted to the data: the greater gradient in the Zen transform underlies the better Kruskal score, as it gives the implication that there is less deviation from the regression function. Note that the offset from the line $y = x$ is not significant with respect to this measure.
%
In this example, PCA gives a slightly better outcome than MDS, but this is not significant.
%: in fact, as the data is not contained within a meaningful manifold, both methods give effectively random projections. 
It is however significant than \nsimp \zen gives a better outcome: this function relies upon properties of the domain geometry which are not available to a linear transform. RP, as expected, performs the worst of the four transforms.

Figure \ref{fig_100dim_quality} shows the outcomes of the various quality measures, as the reduction dimension is reduced from the dimension of the original data down to 2 dimensions. Each quality measure has been normalised into the range $[0,1]$, where 1 implies a perfect outcome and 0 implies that the transform has no effective value. The expectation is that, for each measure, the outcome will start high and monotonically reduce as the dimension of the reduction decreases.

It is clear that for this data set the \nsimp \zen transform consistently performs better  than any of the other techniques for all measures and for all reduction dimensions.
It is also particularly evident that the Kruskal quality of  \nsimp \zen  does not appear to significantly degrade as the reduction dimension is reduced to surprisingly low dimensions. Extraordinarily, the Kruskal stress of \nsimp \zen in the 2-dimensional reduction is less than that of the other techniques at 80 dimensions. This aspect will again be discussed further in Section \ref{sec_discussion}.

The poor Quadratic Loss and Sammon Stress outcomes  for PCA are due to the mechanism being a contraction transform, which therefore introduces a consistent  error across all measurements. These quality measures punish any absolute, rather than relative, error. 
%The absolute loss in accuracy could be  corrected to some degree by use of a scaling function deduced from a subset of the results. 
The \nsimp \zen transform gives much better results in these tests, as the underlying geometry  holds the transformed distances close to those measured in the original space as seen in Figure \ref{fig_10dim_shepards}.

\subsubsection{500-dimensional generated space}
{
\begin{figure}[tp]
\centering
\begin{subfigure}{0.32\textwidth}
\includegraphics[width=\textwidth]{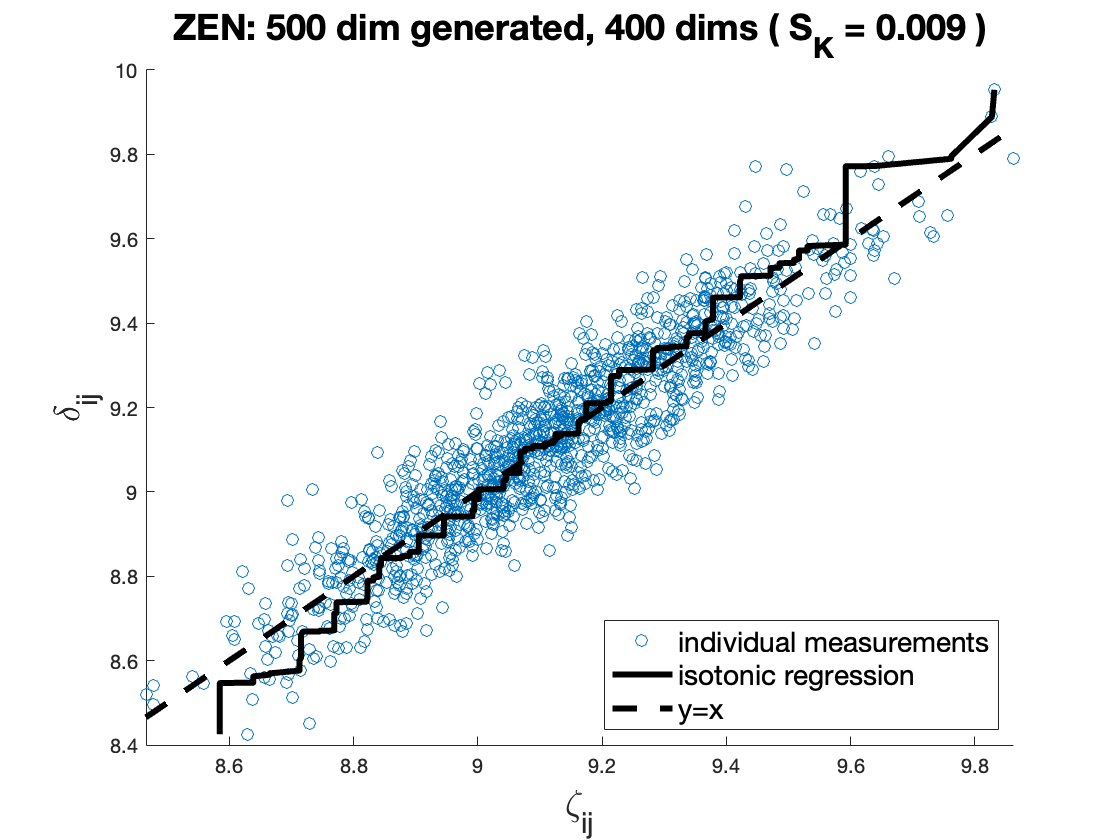}\hfill
\caption{Zen}
\end{subfigure}
\begin{subfigure}{0.32\textwidth}
\includegraphics[width=\textwidth]{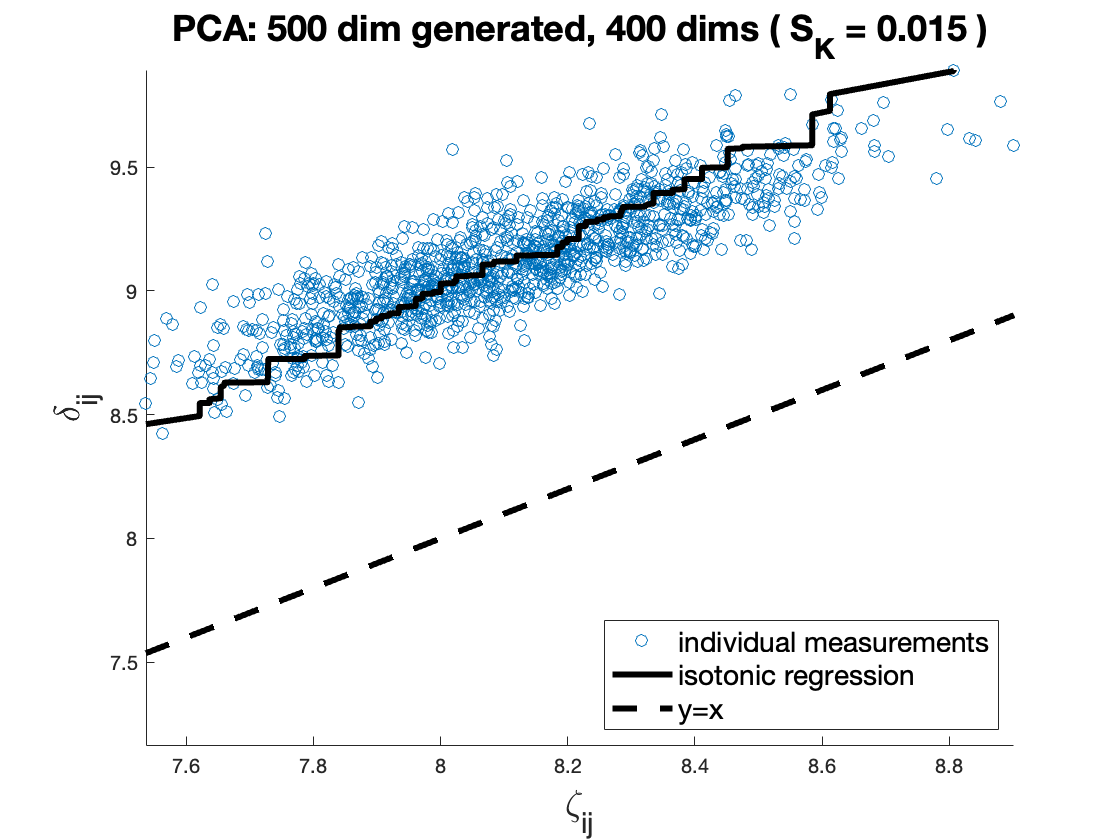}
\caption{PCA}
\end{subfigure}
%
%%second row of diagrams
%\begin{subfigure}{0.32\textwidth}
%\includegraphics[width=\textwidth]{figures/euc_gen_500_shep_mds}\hfill
%\caption{MDS}
%\end{subfigure}
%
\begin{subfigure}{0.32\textwidth}
\includegraphics[width=\textwidth]{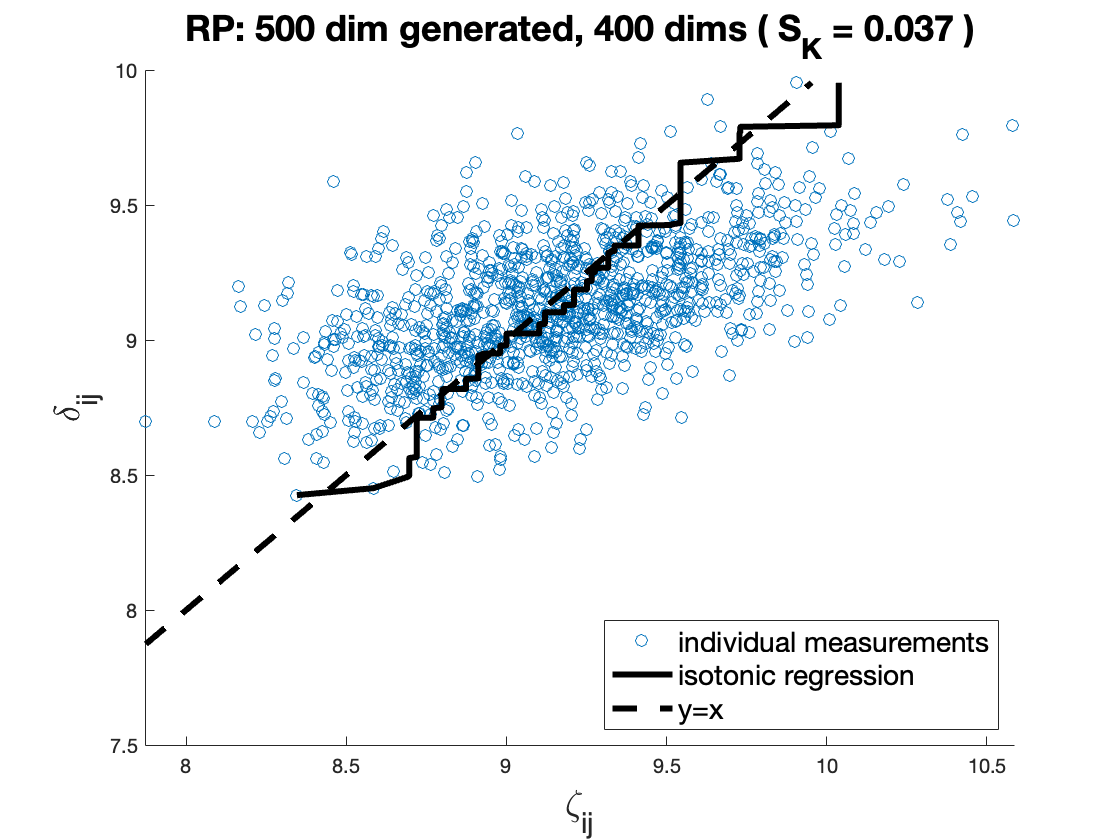}
\caption{RP}
\end{subfigure}
\caption{Shepard plots for the DR mechanisms applied to 500-dimensional space, reduced to 400 dimensions. MDS, as before, gives very similar results to PCA, and from now on we omit that figure from this analysis.}
\label{fig_500_dim_shepards}
\end{figure}
\begin{figure}[tp]
\centering
\begin{subfigure}{0.33\textwidth}
\includegraphics[width=\textwidth]{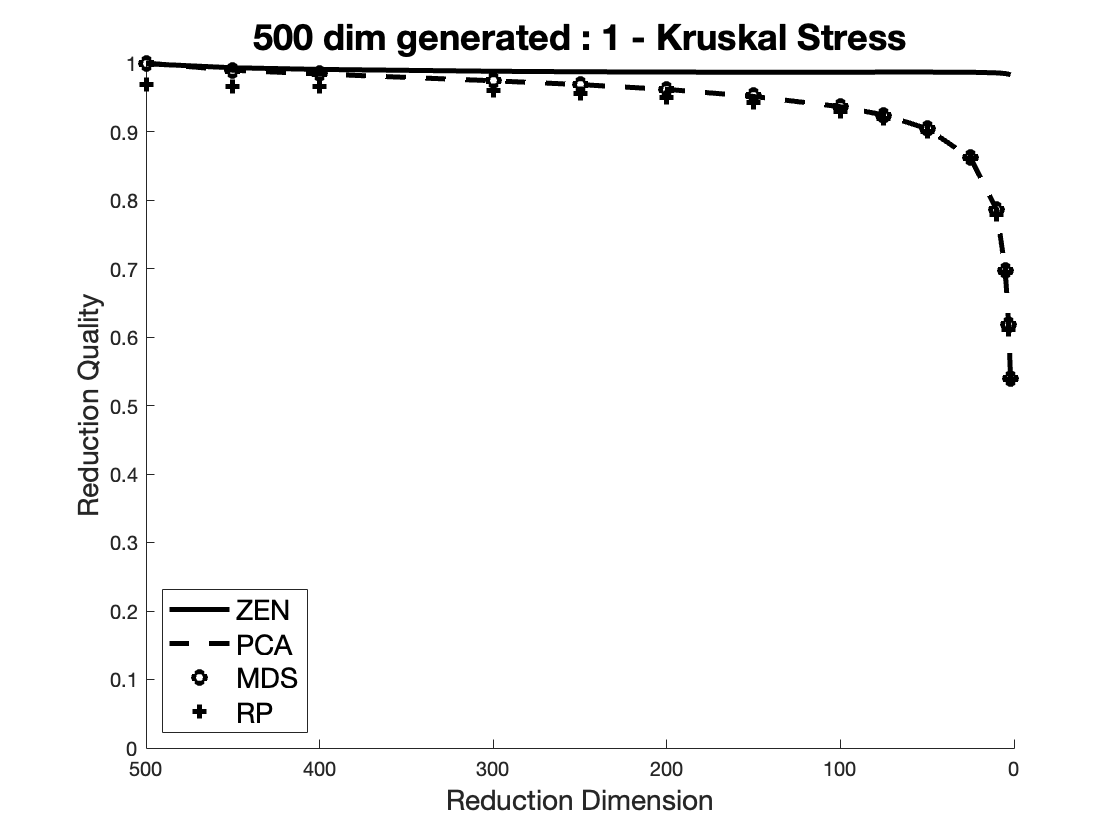}\hfill
\caption{Kruskal stress}
\end{subfigure}
\begin{subfigure}{0.32\textwidth}
\includegraphics[width=\textwidth]{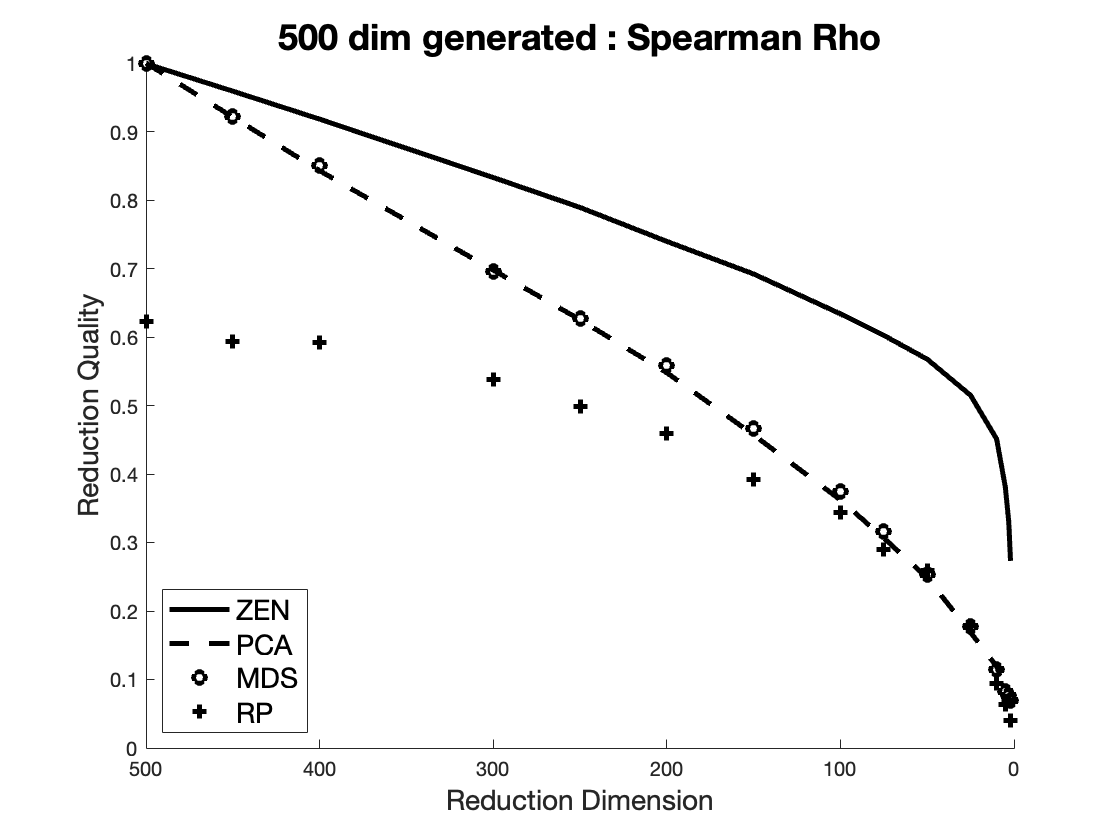}
\caption{Spearman Rho}
\end{subfigure}
\begin{subfigure}{0.32\textwidth}
\includegraphics[width=\textwidth]{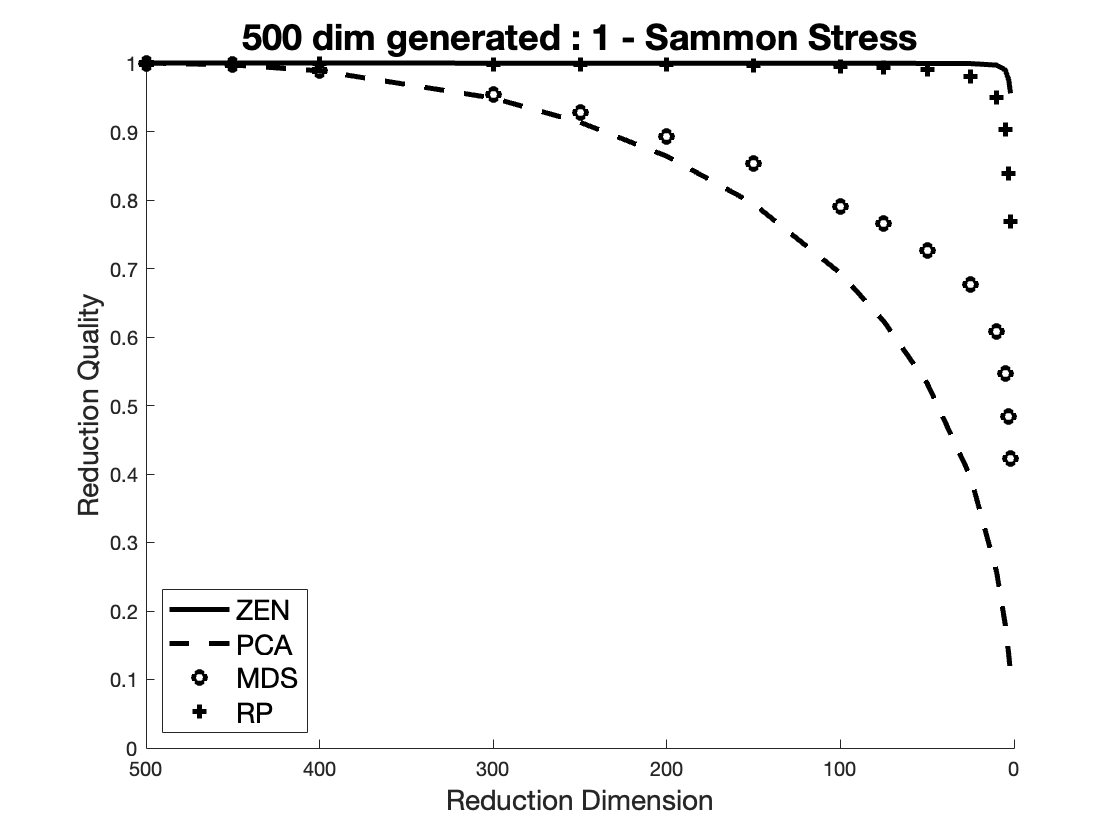}\hfill
\caption{Sammon stress}
\end{subfigure}
\begin{subfigure}{0.32\textwidth}
\includegraphics[width=\textwidth]{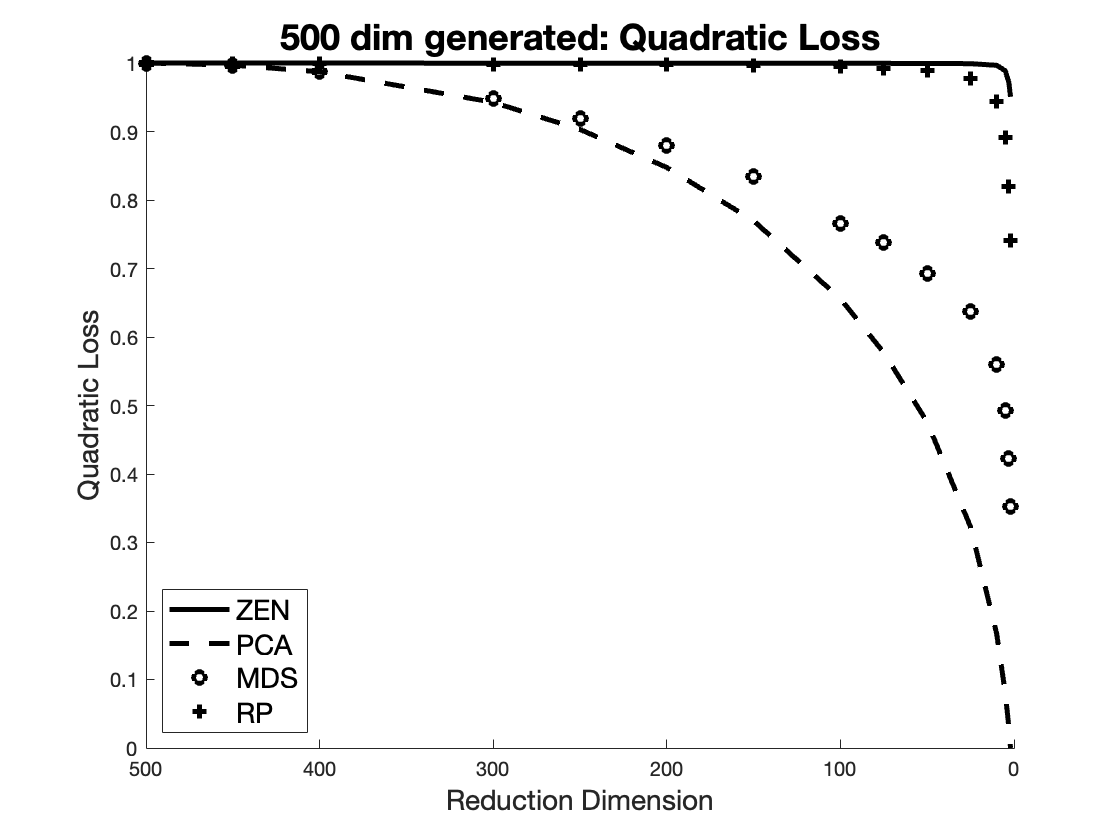}
\caption{Quadratic loss}
\end{subfigure}
\begin{subfigure}{0.32\textwidth}
\includegraphics[width=\textwidth]{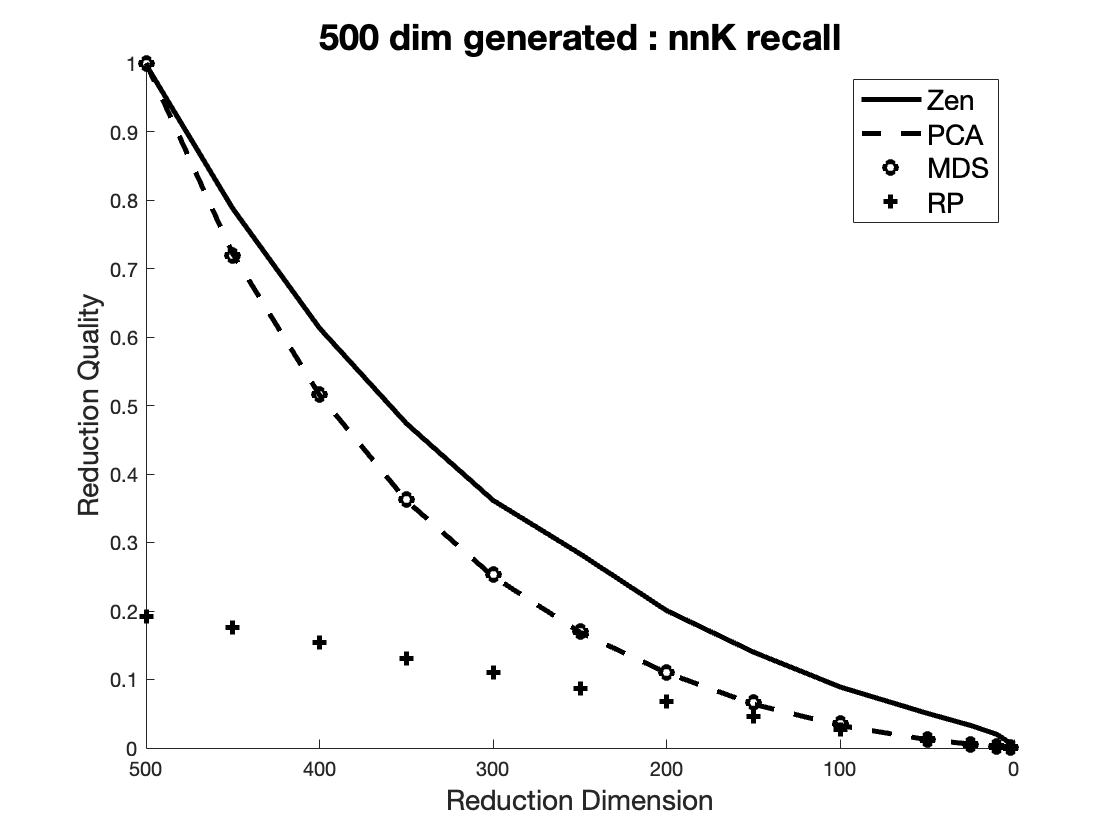}\hfill
\caption{kNN Recall}
\end{subfigure}
\caption{Quality metrics for 500 dimensional Euclidean spaces reduced to between 500 and 2 dimensions.}
\label{fig_500dim_quality}
\end{figure}
}

We repeat the above analysis for a  higher dimension space.  Figure \ref{fig_500_dim_shepards} gives Shepard plots for \nsimp \zen, PCA and RP reduced to 400 dimensions.

Even although the reduction is again to 80\% of the original dimensions, it can be seen that the higher dimensions give relatively better outcomes, as predicted by Johnson-Lindenstrauss. This is particularly evident in the quality charts shown in Figure \ref{fig_500dim_quality}, where it can be further seen that, as dimensions are reduced, RP starts to give equal performance to both PCA and MDS in all quality  measures, and is better for Sammon stress and Quadratic Loss. However, from our perspective, the key result is substantially better performance of the  \nsimp \zen transform for all quality measures across all dimensions.

\subsection{Euclidean spaces from other applications}
\label{sec_sub_exp_real_euc}
\begin{figure}[tbp]
\includegraphics[width=0.32\textwidth]{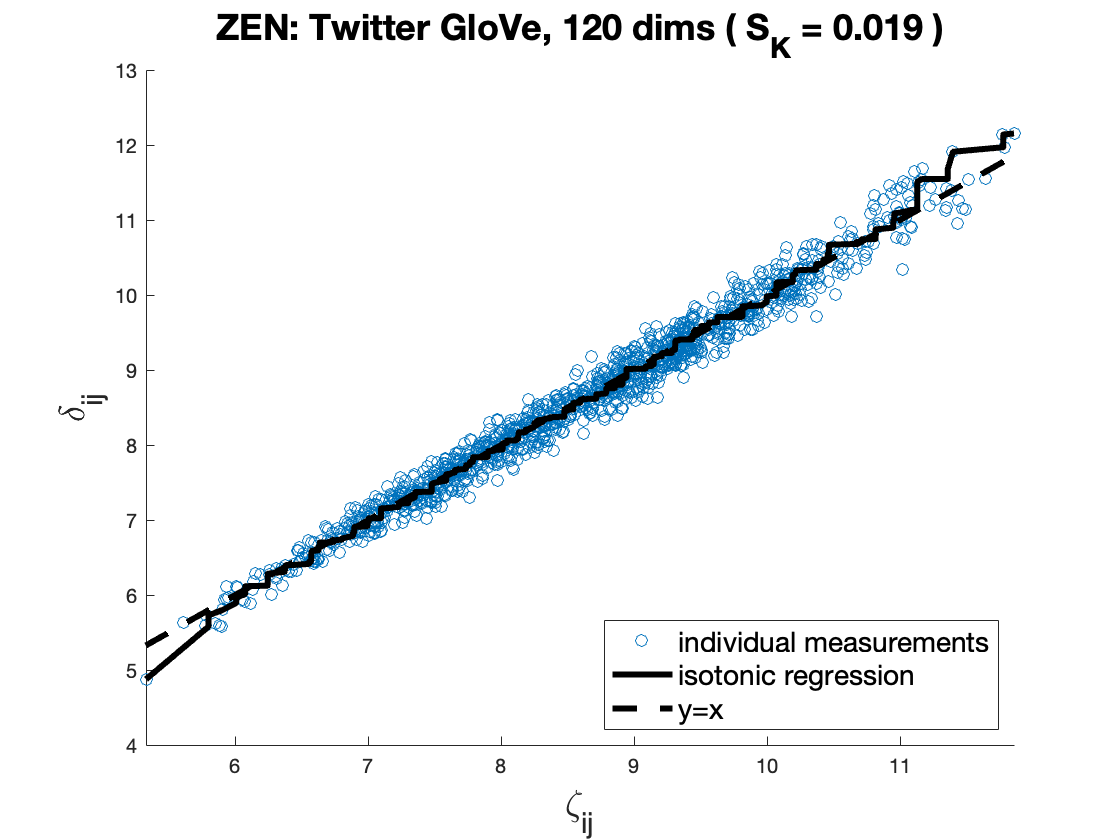} \hfill
\includegraphics[width=0.32\textwidth]{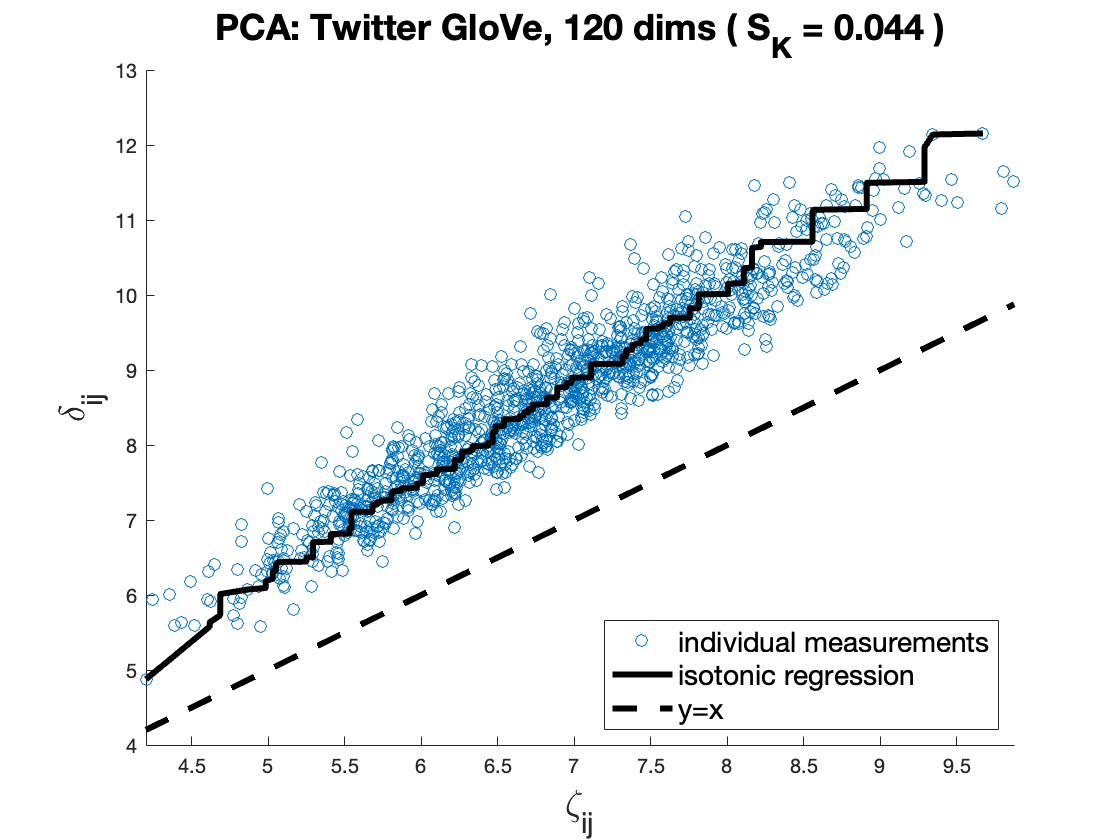} \hfill
\includegraphics[width=0.32\textwidth]{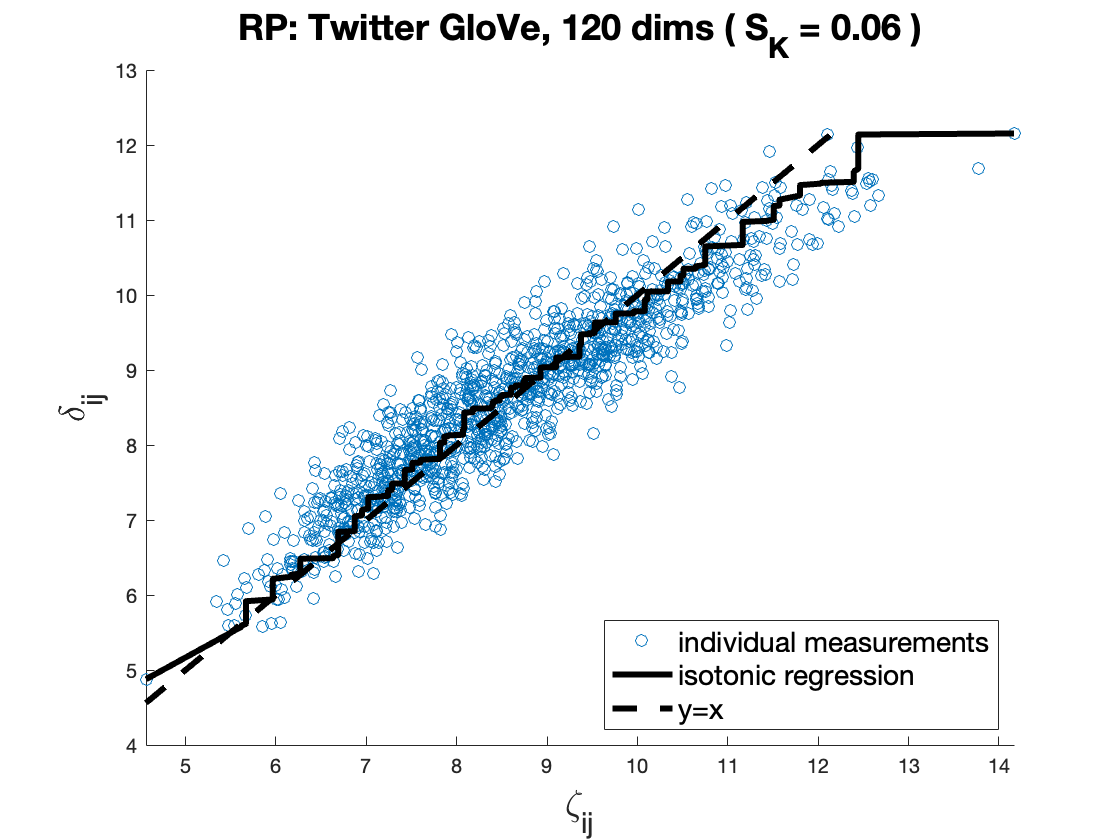} \hfill
\caption{\nsimp \zen, PCA and RP transforms mapping GloVe from 200 to 120 dimensions. MDS gives, as before, very similar results to PCA and is omitted from the Figure.}
\label{fig_glove_shepards}
\end{figure}

\begin{figure}[tbp]
\centering
\begin{subfigure}{0.33\textwidth}
\includegraphics[width=\textwidth]{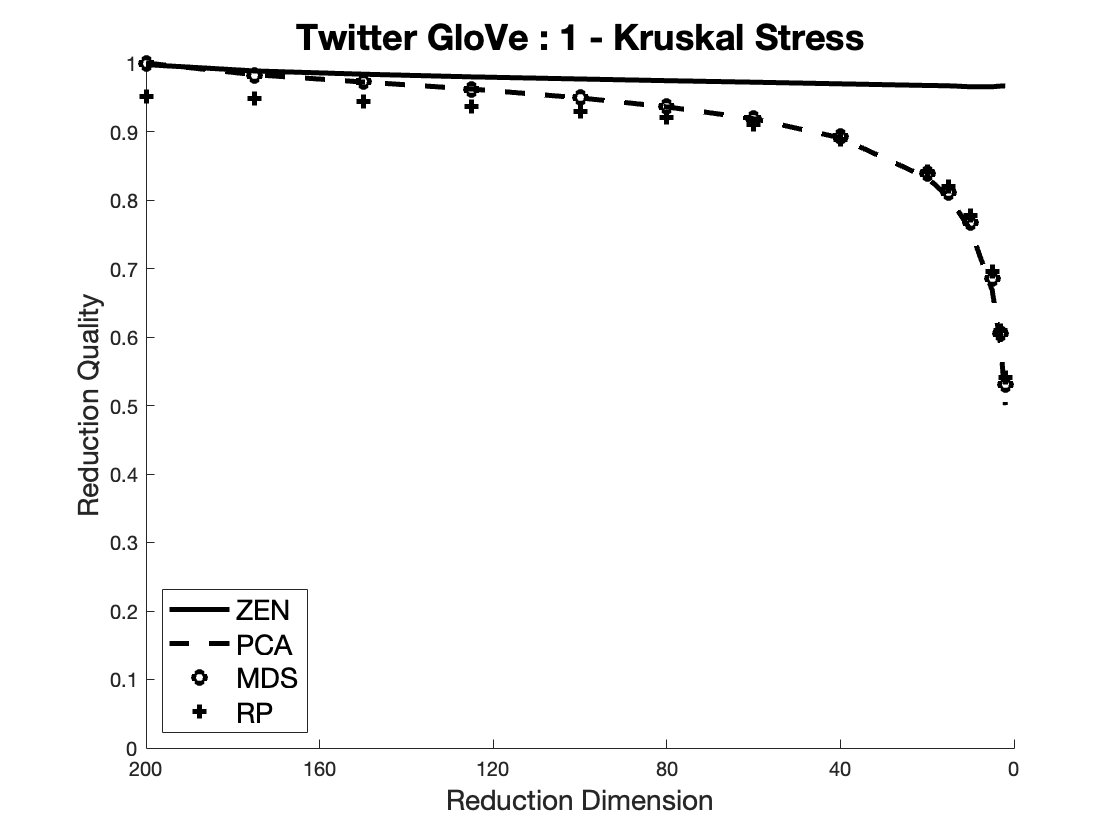}\hfill
\caption{Kruskal stress}
\end{subfigure}
\begin{subfigure}{0.32\textwidth}
\includegraphics[width=\textwidth]{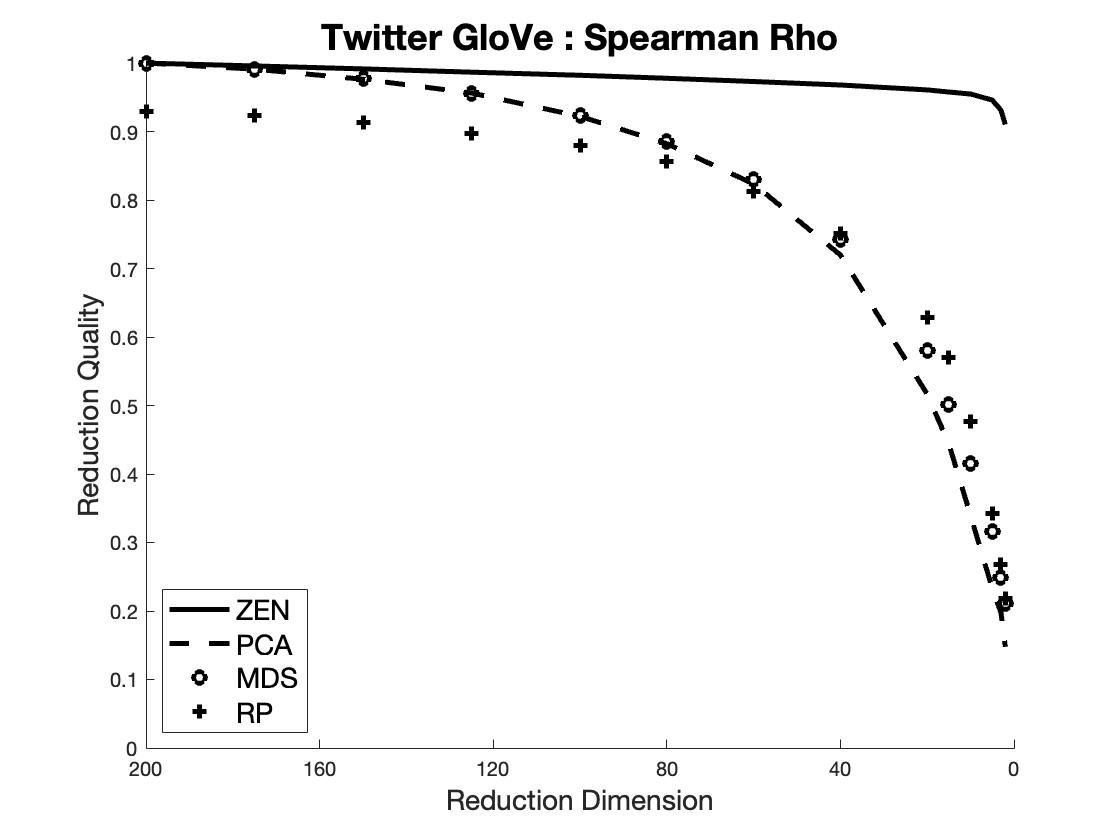}
\caption{Spearman Rho}
\end{subfigure}
\begin{subfigure}{0.32\textwidth}
\includegraphics[width=\textwidth]{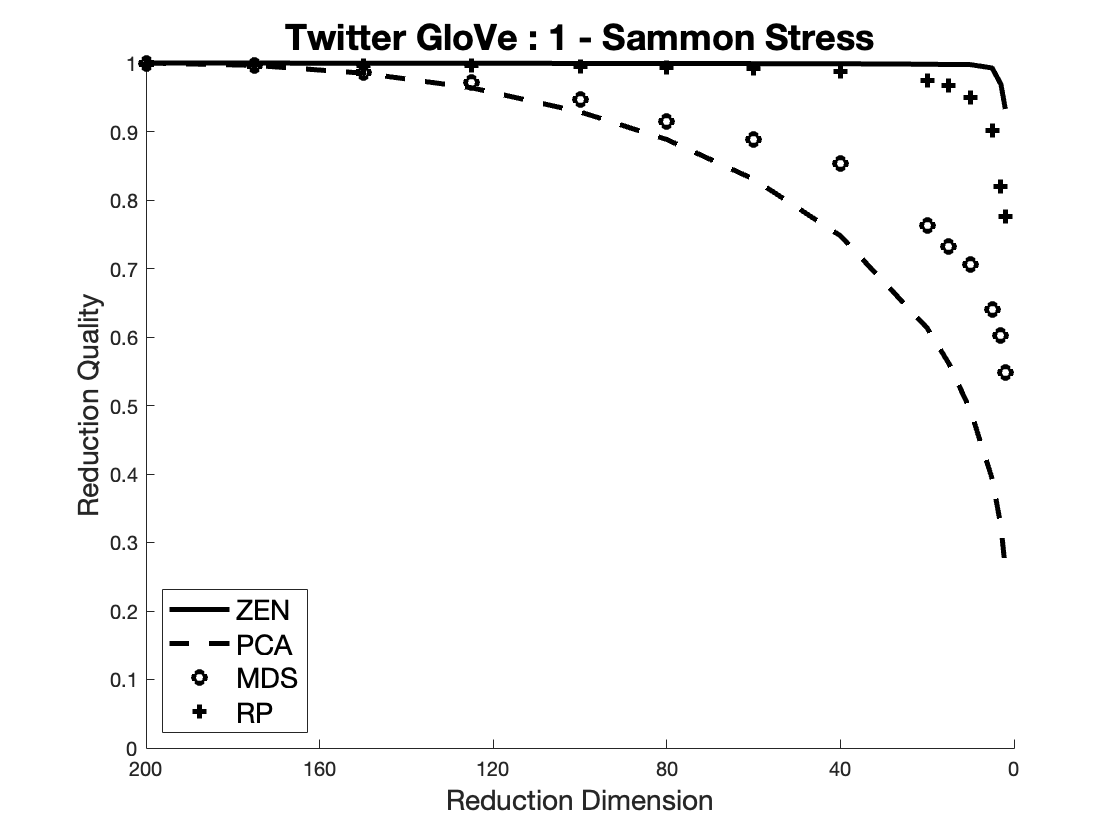}\hfill
\caption{Sammon stress}
\end{subfigure}
\begin{subfigure}{0.32\textwidth}
\includegraphics[width=\textwidth]{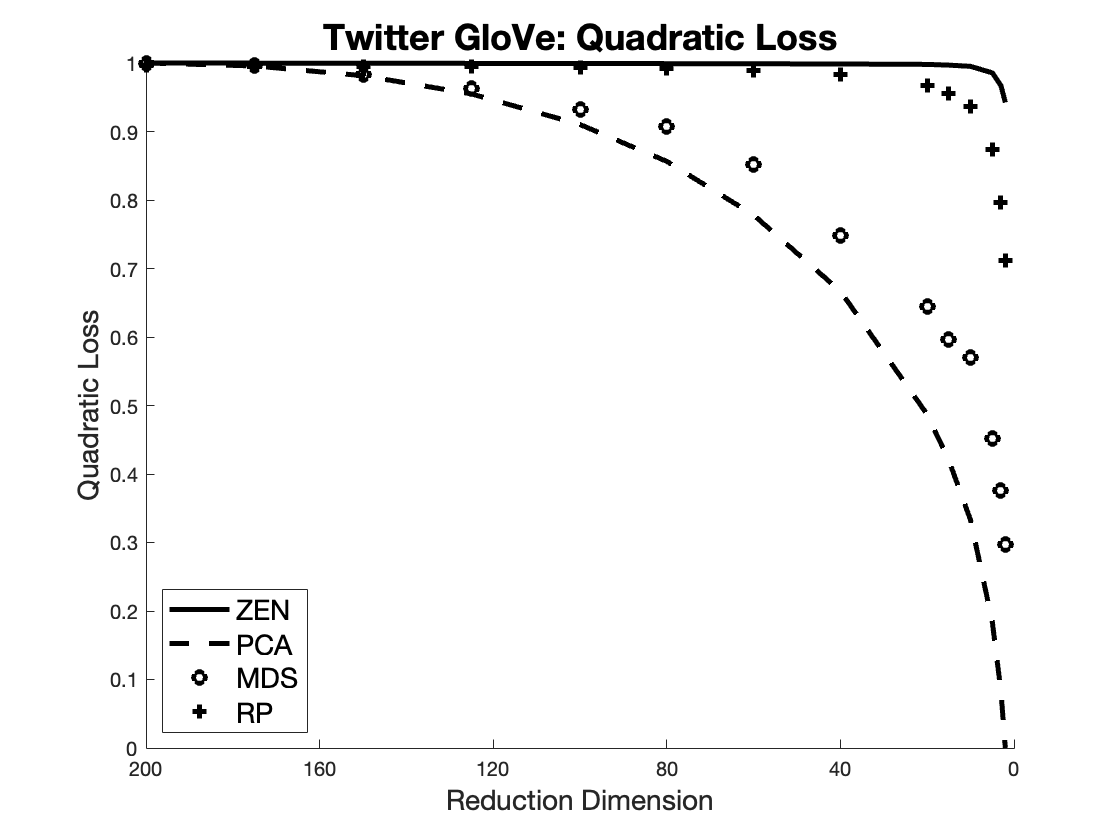}
\caption{Quadratic loss}
\end{subfigure}
\begin{subfigure}{0.32\textwidth}
\includegraphics[width=\textwidth]{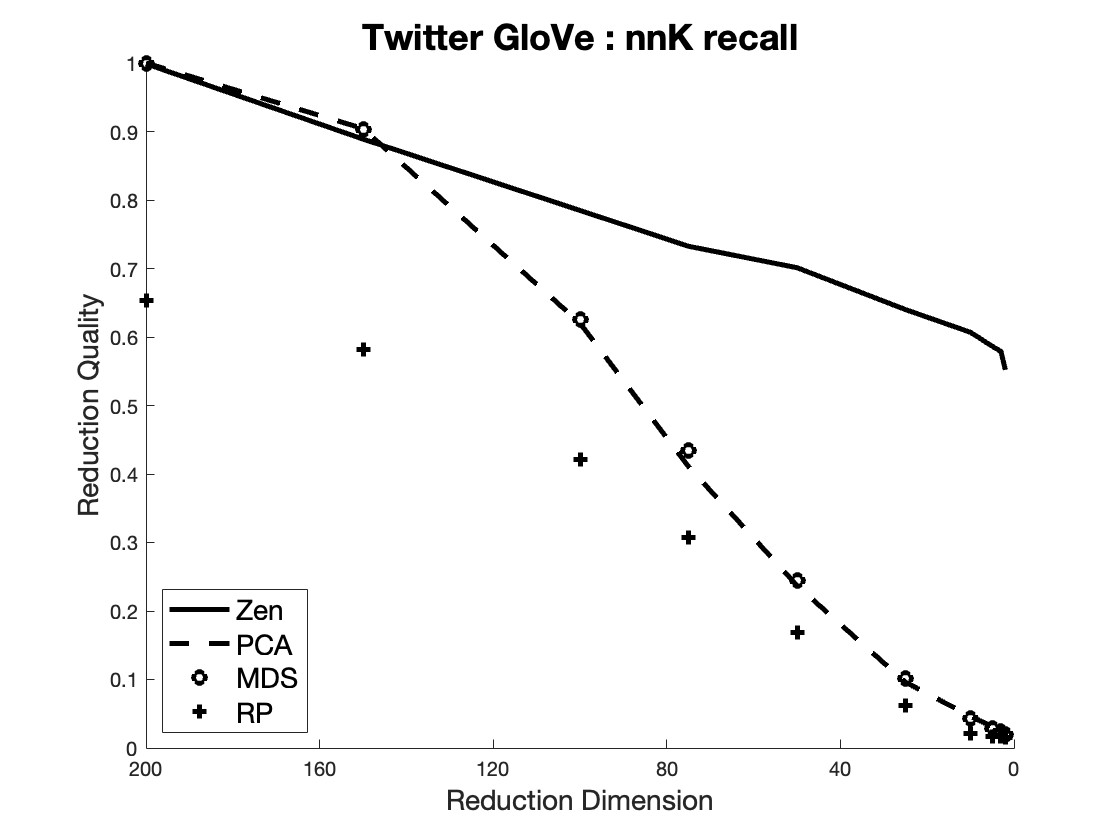}\hfill
\caption{kNN Recall}
\end{subfigure}
\caption{Quality metrics for  Twitter GloVe reduced  to between 200 and 2 dimensions.}
\label{fig_glove_quality}
\end{figure}
 In this section we examine the  reduction of some high-dimensional Euclidean data sets produced in the application of representational techniques to real-world data. These spaces are known to lie within complex manifolds of the Euclidean space in which they are embedded, and are therefore  better subjects for dimensionality reduction than the uniform spaces of the previous section.
  
\subsubsection{Twitter GloVe 200}
GloVe \cite{pennington2014glove} is an unsupervised learning algorithm for obtaining vector representations for words, with the  intent that the distance between vectors is semantically significant.
Twitter GloVe  is the outcome of this algorithm applied to $2 \times 10^9$ individual short texts from Twitter, from which $10^6$ individual tokens are assigned vector values. The achieved semantic similarity is quite striking, for example the closest vectors to the term \emph{frog} are, in order: \emph{frogs, toad, litoria and leptodactylidae}.

Linear substructures are also preserved, for example the relative difference between the word pair (\emph{man}, \emph{king}) is similar to that between the pair (\emph{woman}, \emph{queen}). The authors have considered both Euclidean and Cosine distances over the records and have found no significant advantage to either metric. Pre-trained word vectors %.Variants 
with 25, 50, 100 and 200 dimensions are available online\footnote{\url{https://nlp.stanford.edu/projects/glove_}}, here we have used the 200 dimension version.

The same experiments as above were performed on the Twitter GloVe data set. Before creating the Shepard plots, PCA was used to find the number of dimensions necessary to explain 80\% of the variance (according to Eq. \ref{eqn_pca_variance}). This value is 120, significantly less than the 160 dimensions that would be required for uniform data. This was selected as the reduction dimension to illustrate using Shepherd plots in  in Figure \ref{fig_glove_shepards}. As can be seen, in this context the \nsimp $Zen$ transform   performs by far the best of those tested, the Kruskal stress now being less than half of that obtained using PCA or MDS.

The plots of transform quality with reducing dimensions are shown in Figure \ref{fig_glove_quality}. Some of these results are quite startling: the \nsimp \zen transform is almost always best, for all quality measures, and in some cases maintains high quality values down to tiny reduction dimensions compared to the other techniques.

The reason for this relative increase in performance is, we believe, due to the nature of the manifold in which the data lie.  PCA, MDS and RP all  produce linear transforms of this manifold, whereas \nsimp \zen's transform is non-linear, allowing it to respect the geometry of the original manifold with respect to each object mapped into the lower dimension. The \nsimp \zen transform shown in Figure \ref{fig_glove_shepards} is produced with reference to only 120 reference objects, as opposed to the $1,000$ objects used for PCA , and in Figure \ref{fig_glove_quality} the transform at each dimension is produced using only that number of reference points to represent the manifold in which the domain lies. Even for example with a random selection of only 20 reference objects, mapping to 20 dimensions, it is clear that \nsimp \zen performs far better than the linear transforms which use many more reference objects, even when mapped to many more dimensions.

%For those readers, like ourselves, frankly sceptical that so much information  from 200  dimensions  can be contained in just 2 dimensions, Figure \ref{fig_glove_2_shepards} shows Shepard plots for the 2-dimensional reduction. As can be seen from the plots, the Kruskal stress for the $zen$ reduction to two dimensions is less than that for the PCA reduction to 120 dimensions.
%%
%%\begin{figure}[th]
%%\centering
%%\includegraphics[width=0.45\textwidth]{figures/glove_quality_chart}
%%\caption{Reduction quality  for Twitter GloVe data, as reduction dimension decreases. It can be seen that \nsimp $Zen$ is the best transform for any reduction to less than around 150 dimensions, and again performs surprisingly well at very low dimensions.}
%%\label{fig_glove_quality}
%%\end{figure}
%
%\begin{figure}[th]
%\includegraphics[width=0.32\textwidth]{figures/glove_glove_zen_2_shepard} \hfill
%\includegraphics[width=0.32\textwidth]{figures/glove_glove_pca_2_shepard} \hfill
%\includegraphics[width=0.32\textwidth]{figures/glove_glove_rp_2_shepard}
%\caption{\nsimp \zen, PCA and RP  transforms mapping GloVe to 2 dimensions. Extraordinarily, it can be seen that the Kruskal stress of zen reduced to 2 dimensions is better than that of the PCA reduction to 120 dimensions.  Interestingly, the absolute position of the point cloud of the zen projection varies significantly with the random choice of reference values, but the quality measures do not.}
%\label{fig_glove_2_shepards}
%\end{figure}

\subsubsection{MirFlickr 1M / Alexnet}
\label{sec_subsec_MF_Alex}
\begin{figure}[tbp]
\includegraphics[width=0.32\textwidth]{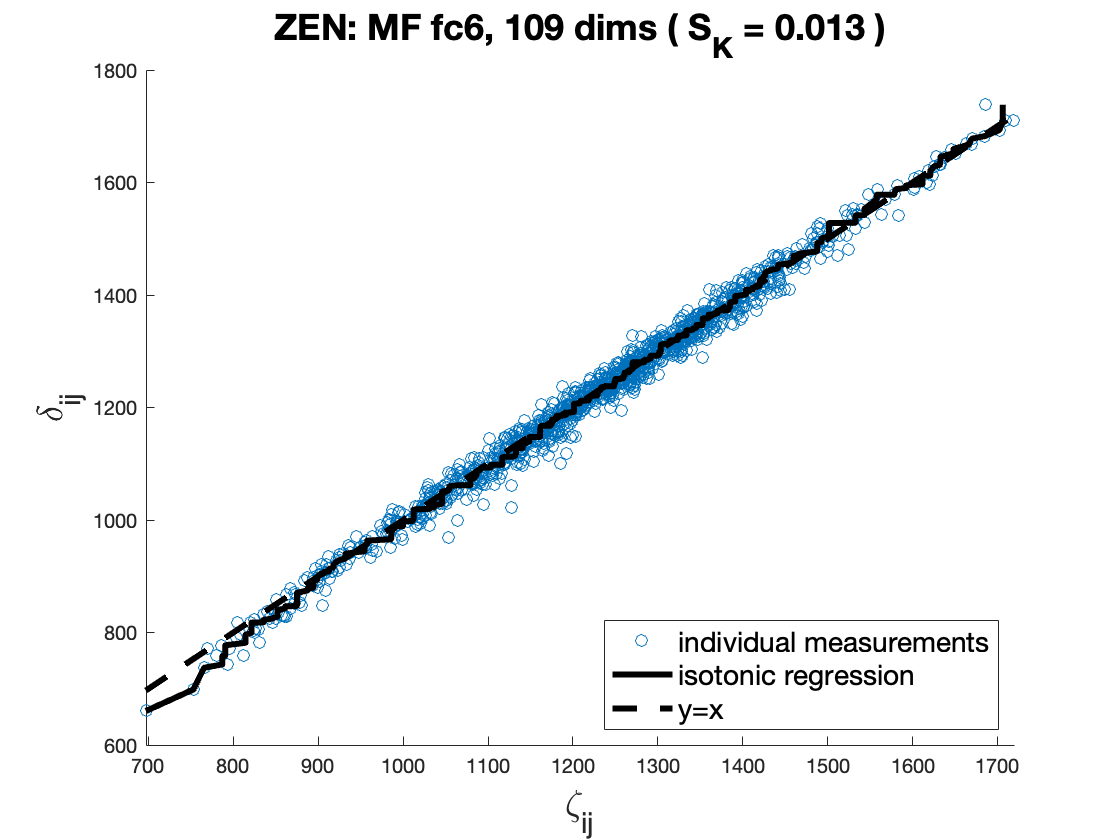}
\includegraphics[width=0.32\textwidth]{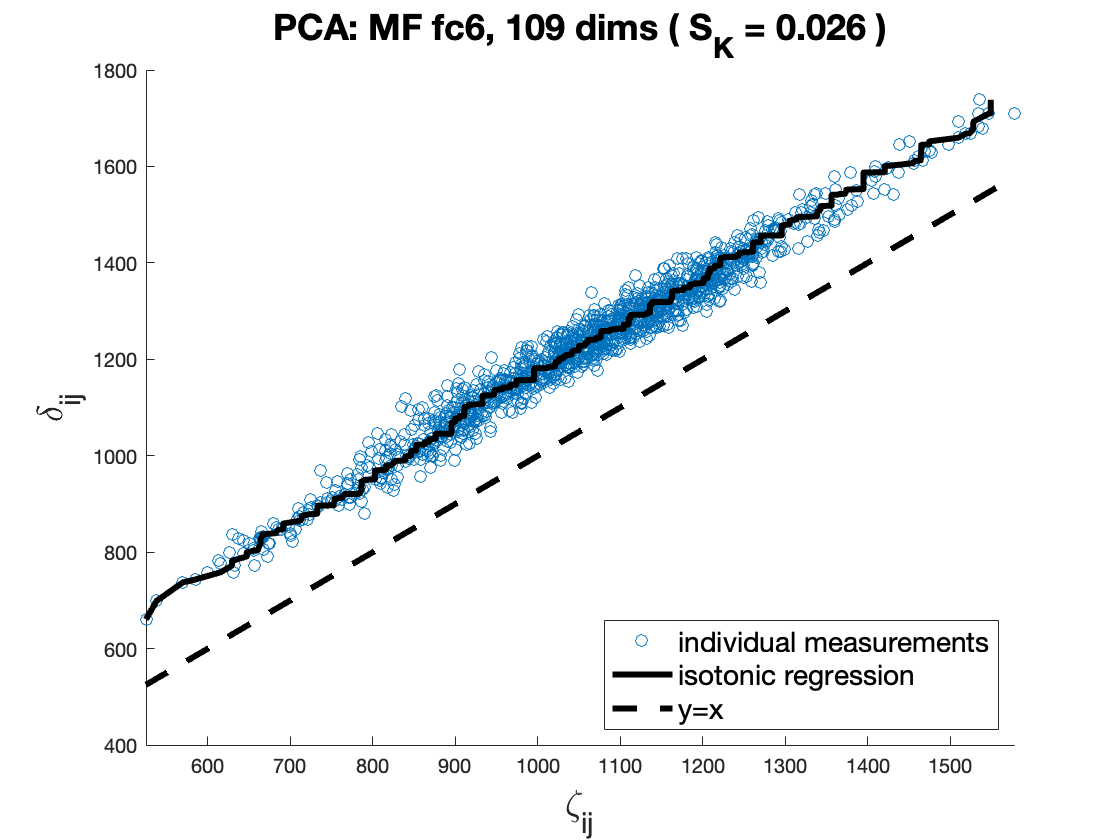}
\includegraphics[width=0.32\textwidth]{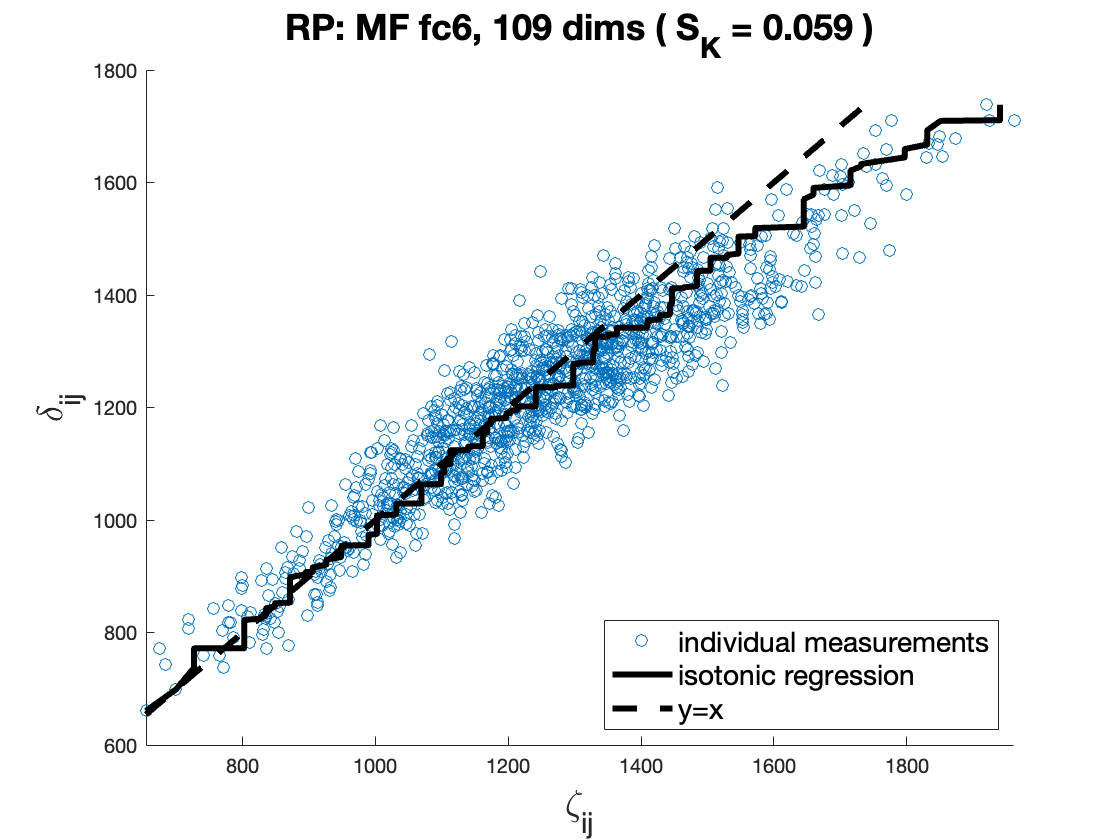}
\caption{MirFlickr fc6 representations, reduced from 4,096 dimensions to 109.}
\label{fig_fc6_shepards}
\end{figure}
\begin{figure}[tbp]
\centering
\begin{subfigure}{0.32\textwidth}
\includegraphics[width=\textwidth]{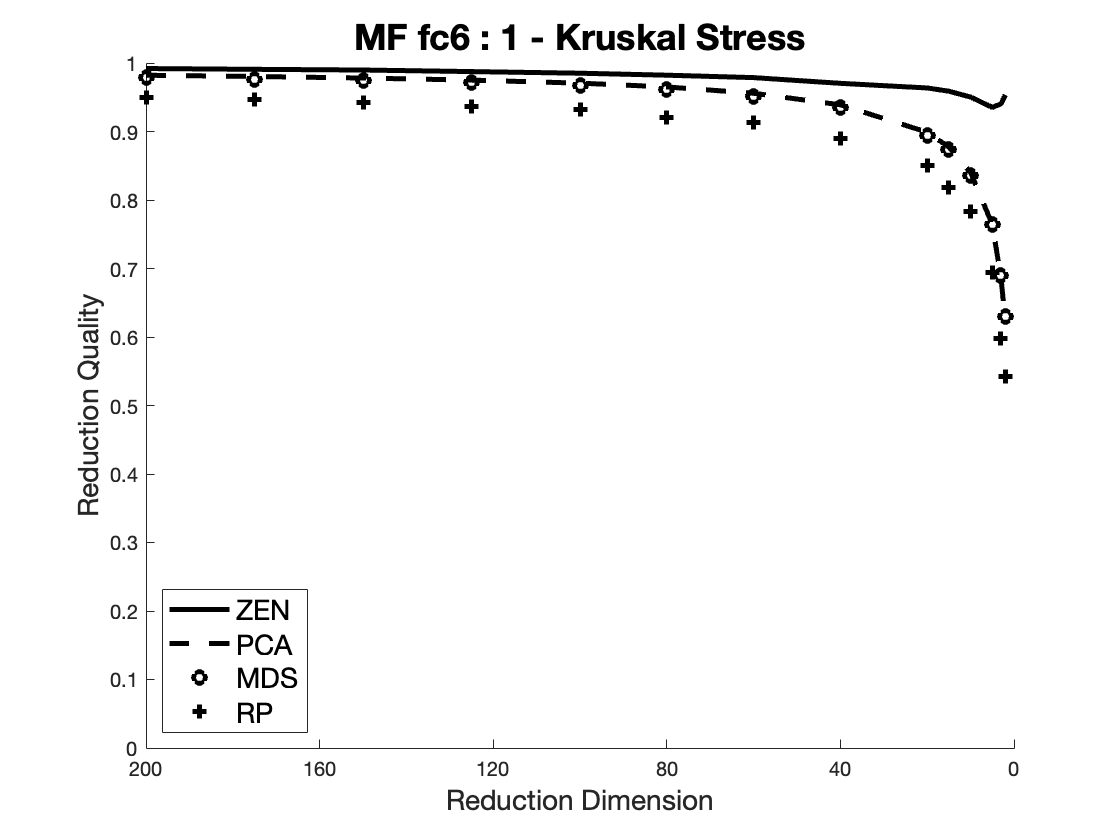}\hfill
\caption{Kruskal stress}
\end{subfigure}
\begin{subfigure}{0.32\textwidth}
\includegraphics[width=\textwidth]{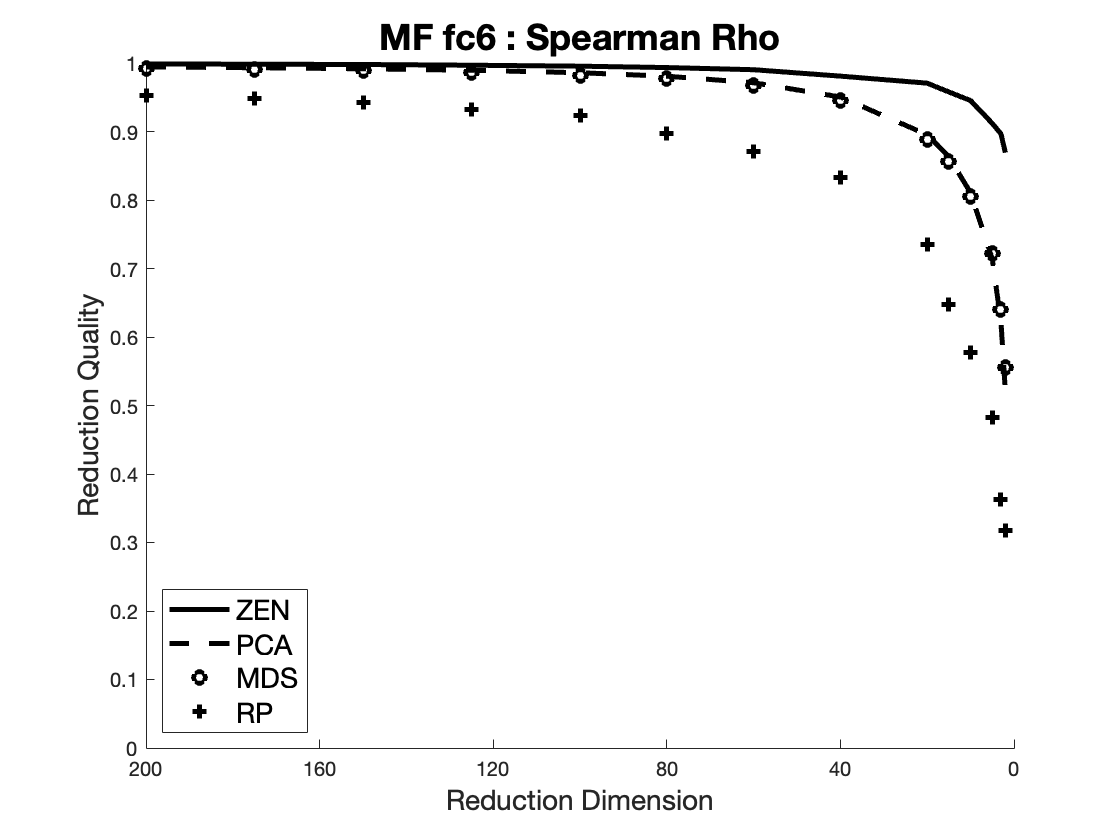}
\caption{Spearman Rho}
\end{subfigure}
\begin{subfigure}{0.32\textwidth}
\includegraphics[width=\textwidth]{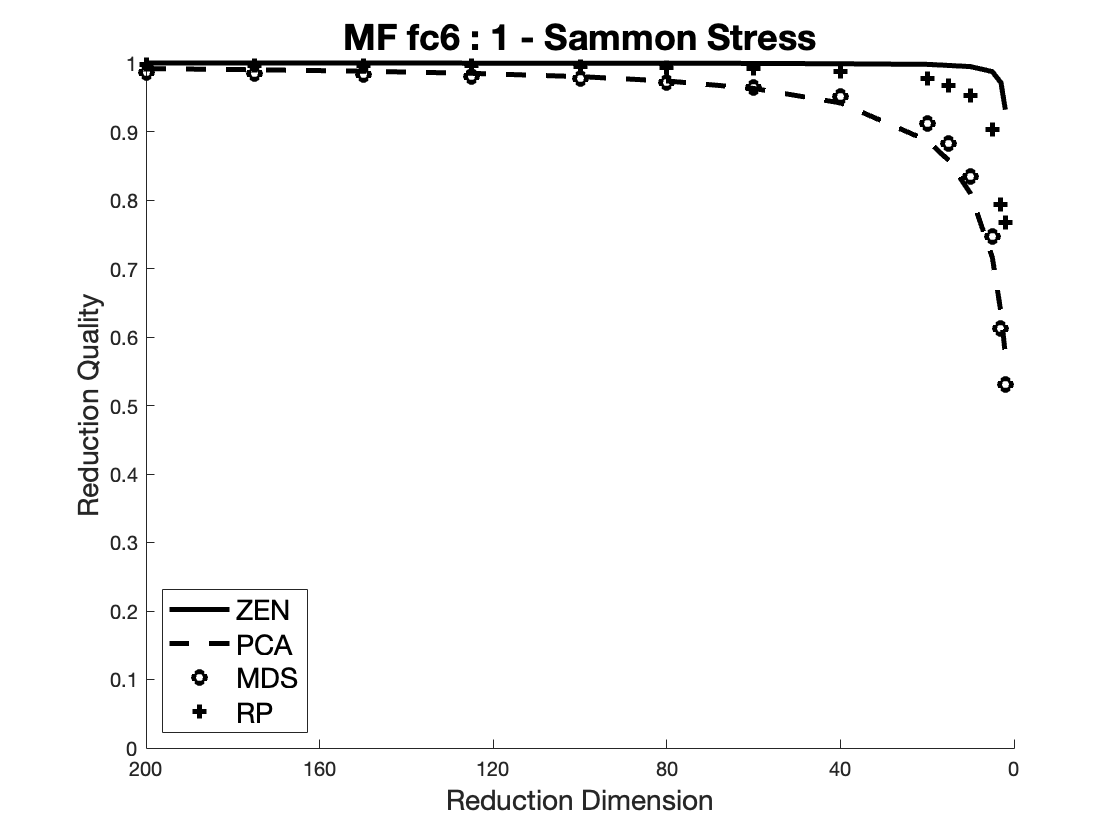}\hfill
\caption{Sammon stress}
\end{subfigure}
\begin{subfigure}{0.32\textwidth}
\includegraphics[width=\textwidth]{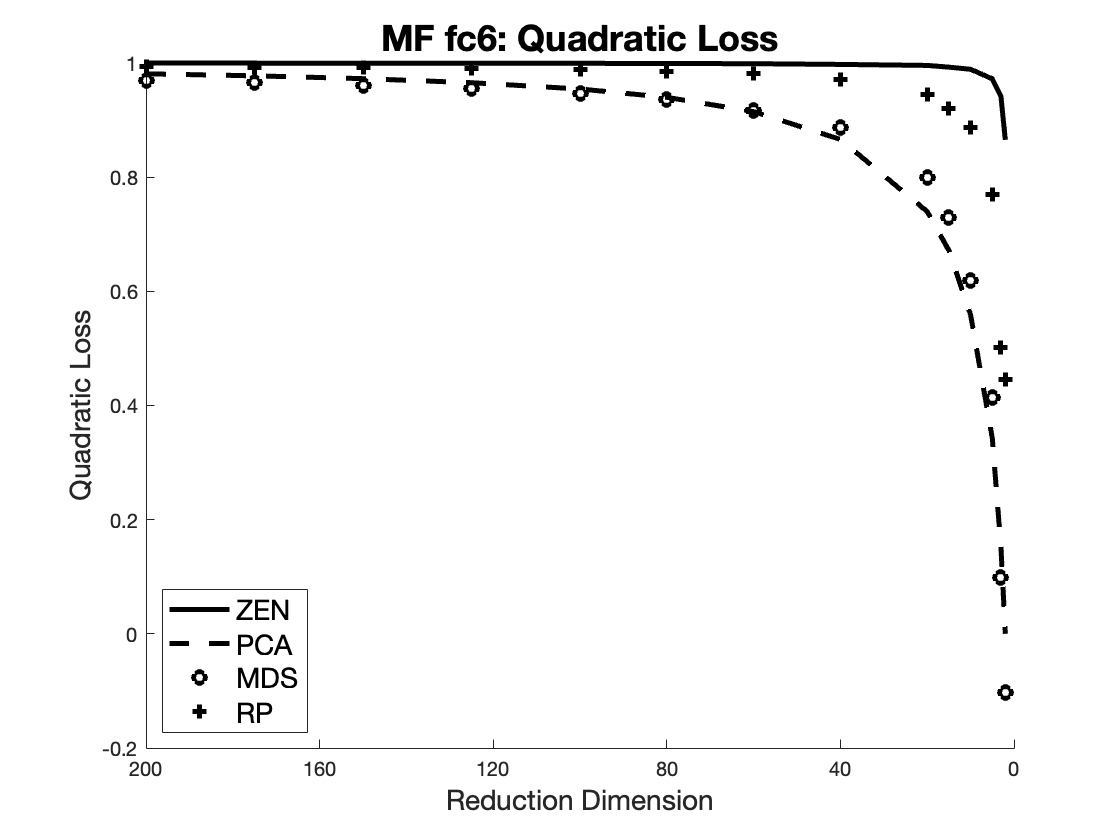}
\caption{Quadratic loss}
\end{subfigure}
\begin{subfigure}{0.32\textwidth}
\includegraphics[width=\textwidth]{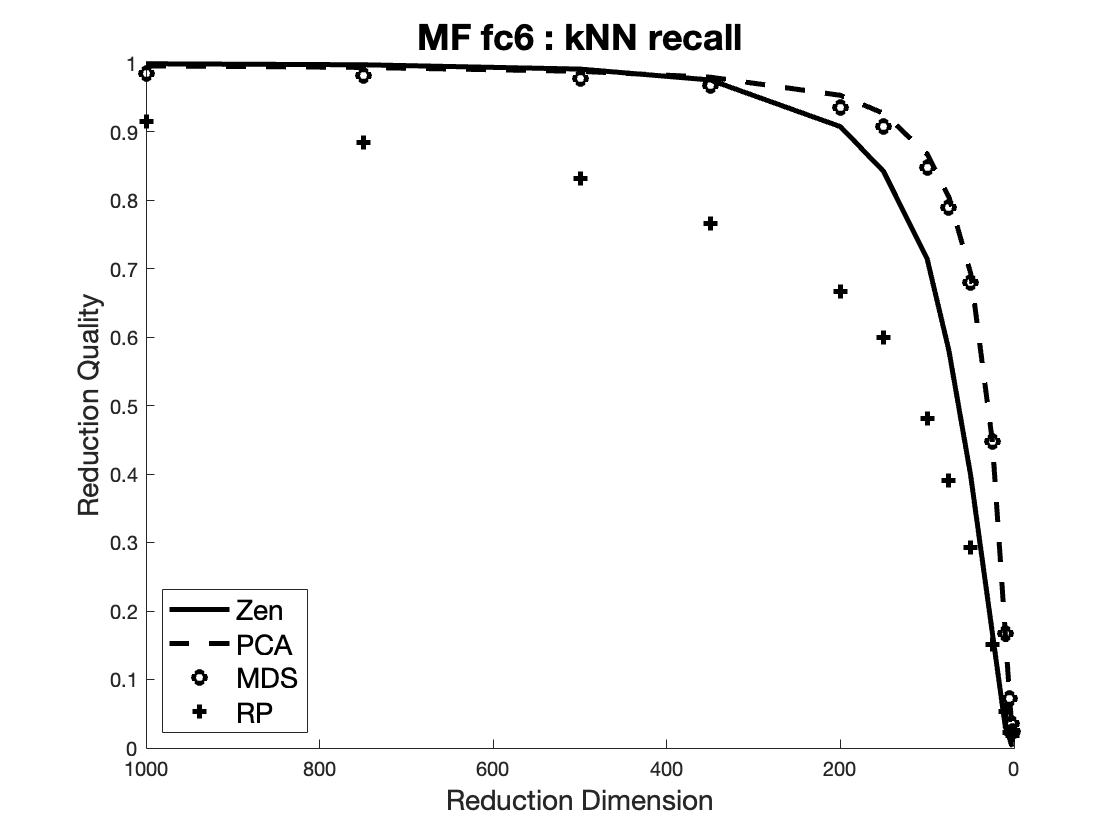}\hfill
\caption{kNN Recall}
\end{subfigure}
\caption{Quality metrics for  Mirflickr fc6. For most plots, the X-axis (dimensionality of the reduction) goes between 200 and 2, as there is very little loss of quality for any of the measurements above 200 dimensions. The recall experiment is run from 1,000 dimensions down to 2.}
\label{fig_fc6_quality}
\end{figure}
This data derives from the AlexNet convolutional neural network \cite{alexnet} applied to the set of one million images available from the MirFlickr project \cite{huiskes08}. While this network is starting to be considered a little dated, as its categorisation performance is less good than some more modern networks, the combination gives  a highly available network applied to a highly available large image collection: the purpose here is just to provide a realistic set of data, with meaningful semantics, in high dimensions. The data is taken from the first fully-connected  (\emph{DeCAF, fc6}) layer of the network, after the initial convolutional layers and before the remaining fully-connected layers of the network. Euclidean distance over this representation has been shown to give an excellent proxy to image similarity even for categories of image that are not included in the original classification \cite{decaf,decaf_novak}. In this experiment we apply Euclidean distance to the data as extracted, and in Section \ref{sec_subsec_mf_relu_cos} we use the same data with Cosine distance applied to the post-RELU filtered version.
% While the latter is the distance function used to train the categorisation, both metric versions have been shown to give excellent similarity proxy functions.

The data  used for this section  therefore comprises 4,096 Euclidean dimensions, including both  positive and negative values, and lies within a complex manifold where the PCA eigenvalues determine that only 109 dimensions are required to explain 80\% of its variance (Eq. \ref{eqn_pca_variance}). As before we compare RP, PCA and \nsimp $Zen$ at this reduction dimension, the results of which are shown in Figure \ref{fig_fc6_shepards}.  Again MDS  and PCA give almost  indistinguishable outcomes for this test.

Again, it is visually evident that the $Zen$ function gives a much tighter fit to the true distances than either RP or PCA, borne out by the lower value of Kruskal stress, and is generally  much closer than either to the true distance.

For the first time, we show a result where the  \nsimp \zen transform is less good than either PCA or MDS across the range of reduced dimensions: while \nsimp \zen remains the highest-scoring mechanism for almost all the quality measures across all reduction dimensions, Figure \ref{fig_fc6_quality} shows that the recall test for \nsimp \zen is  worse for recall than either PCA or MDS when the reduction dimension is less than around 300. This is initially surprising, as the Spearman Rho measure, which tests the preservation of ordering among distances, shows a better performance for \nsimp \zen. The phenomenon being displayed is that \zen's performance in this test is less good over very small distances, which we examine in more detail in Section \ref{sec_subsec_very_small_distances}.

%Repeating the experiment again at a tiny dimensionality...
%
%\begin{figure}[th]
%\includegraphics[width=0.32\textwidth]{figures/fc6_fc6_rp_5_shepard}
%\includegraphics[width=0.32\textwidth]{figures/fc6_fc6_pca_5_shepard}
%\includegraphics[width=0.32\textwidth]{figures/fc6_fc6_zen_5_shepard}
%\caption{MirFlickr fc6 representations, reduced from 4,096 dim to 5.}
%\label{fig_fc6_shepards}
%\end{figure}
%
%\hl{how is the recall?}

\subsection{Other Hilbert spaces - Cosine Distance}
\label{sec_sub_exp_cos}
Cosine similarity is frequently applied over large high-dimensional spaces in the context of Information Retrieval \cite{manning_raghavan_schtze_2008}. As noted in Appendix \ref{appendix_hilbert_metrics}, the most common interpretation of Cosine distance, the complement of the normal cosine similarity (the cosine of the angle between vectors) is not a proper metric. The angle itself does give a proper metric,  and can be used as a proxy which gives the same ordering within a space. To avoid the potentially expensive $\arccos$ function, Euclidean distance measured over the end-points of $\ell_2$-normalised vectors is another proper metric, with the same ordering, and which also has the Hilbert properties. While the Euclidean metric is used, such spaces are however very different from general Euclidean spaces in terms of the distribution of distances.

\subsubsection{ANN SIFT}

The  ANN\_SIFT1M \cite{ann_sift}
 dataset comprises vectors of Angular Quantisation-based Binary Codes (AQBC) \cite{aqbc} deriving from the SIFT \cite{sift} feature analysis of one million images. The similarity of such representations is intended to be assessed using Cosine similarity.
 SIFT is no longer state-of-the-art in image similarity, but the benchmark is still widely used and provides a valuable set of data for this purpose. In these experiments the 128-dimensional data is $\ell_2$-normalised and Euclidean distance is used to provide  a semantic proxy for Cosine distance.
 
 As before, PCA and Eq. \eqref{eqn_pca_variance} are used to determine the number of dimensions required to explain 80\% of the variance in distances, which turns out to be only 28 dimensions. Figure \ref{fig_sift_ann_sift_shepards} shows the Shepard plots of this data for \nsimp \zen, PCA and RP. In this case there is little visual difference between the plots for PCA and \nsimp \zen, although the Kruskal score for \nsimp \zen is  significantly better. 
\begin{figure}[tbp]
\includegraphics[width=0.32\textwidth]{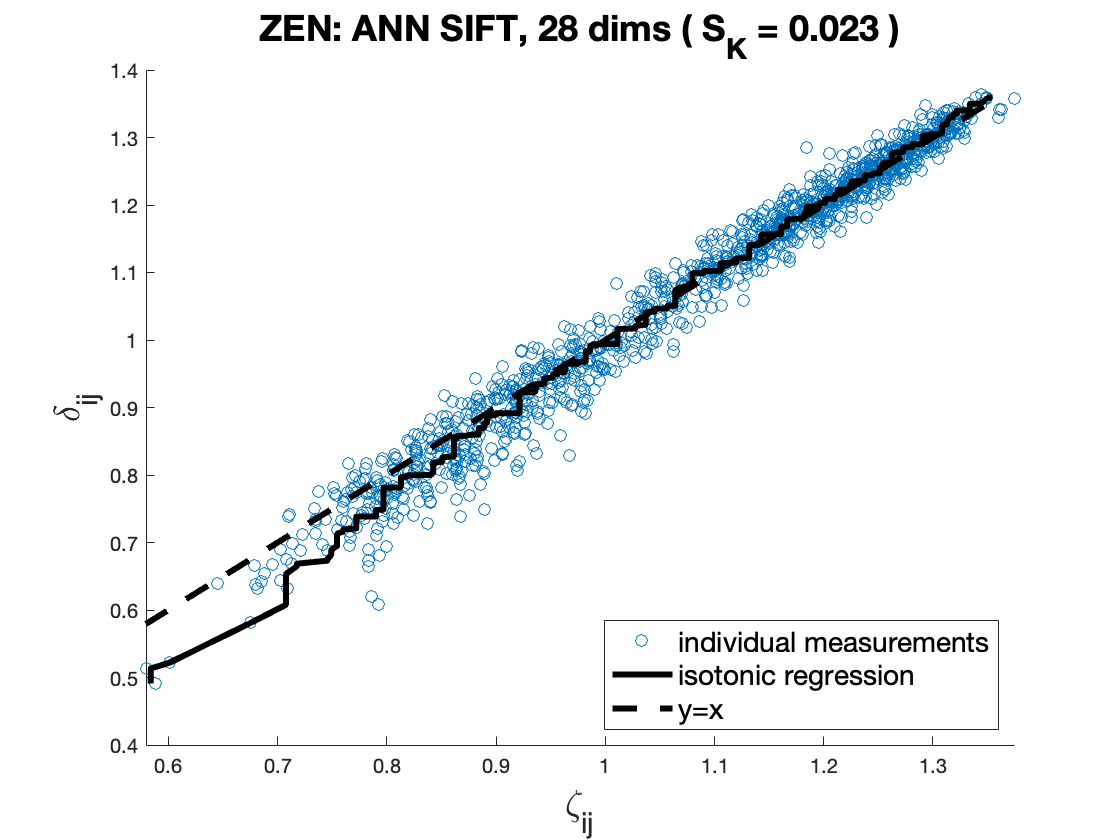} \hfill
\includegraphics[width=0.32\textwidth]{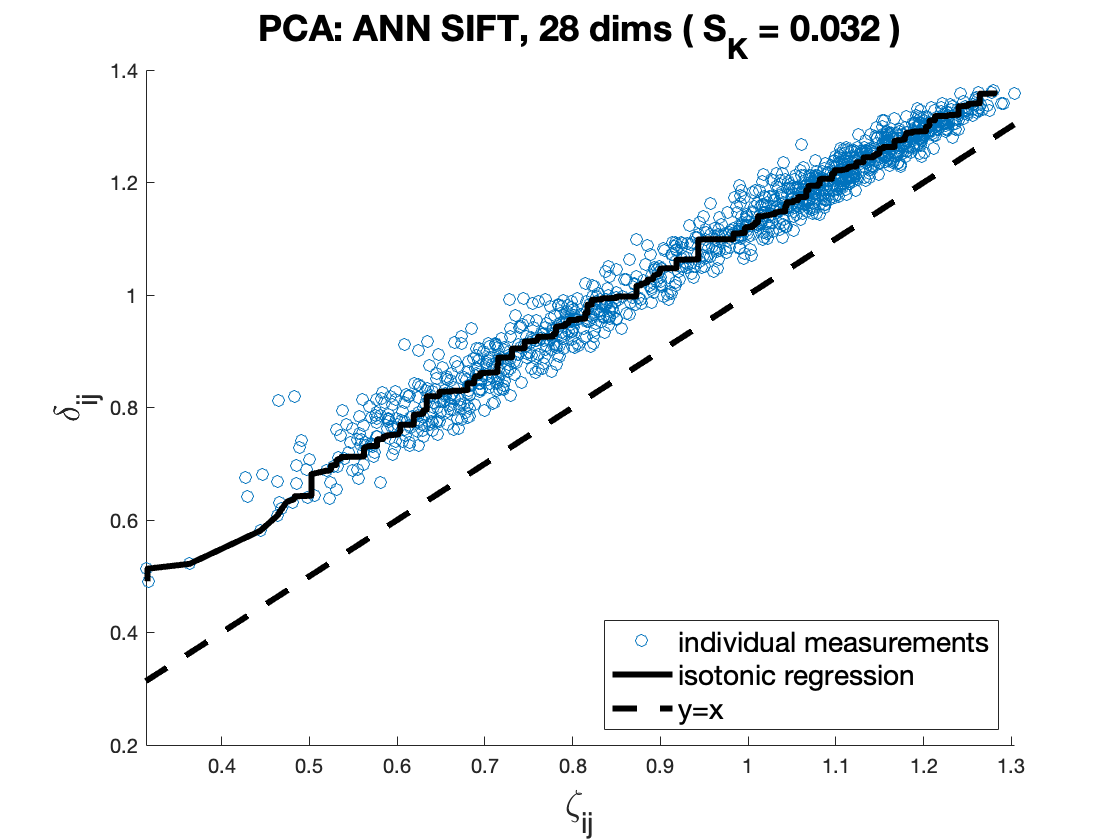} \hfill
\includegraphics[width=0.32\textwidth]{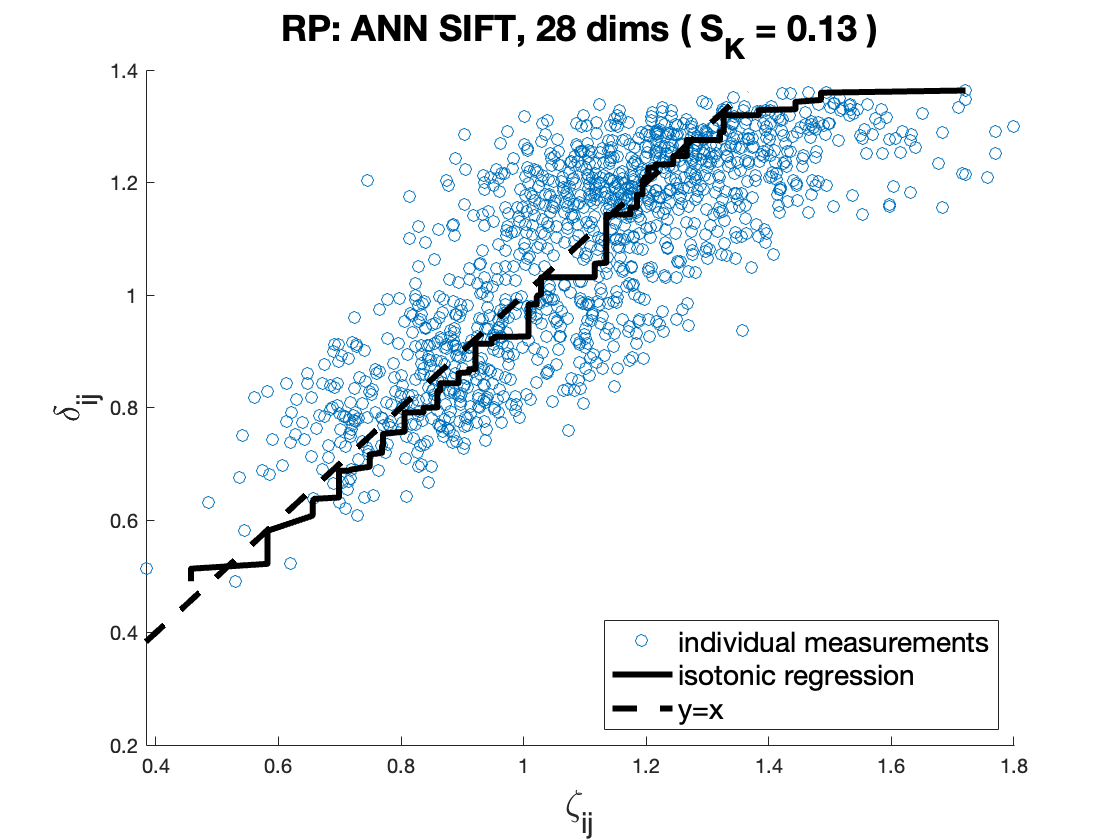}
\caption{ANN-SIFT $\ell_2$-normed representations, reduced from 128 dimensions to 28.}
\label{fig_sift_ann_sift_shepards}
\end{figure}

\begin{figure}[tbp]
\centering
\begin{subfigure}{0.32\textwidth}
\includegraphics[width=\textwidth]{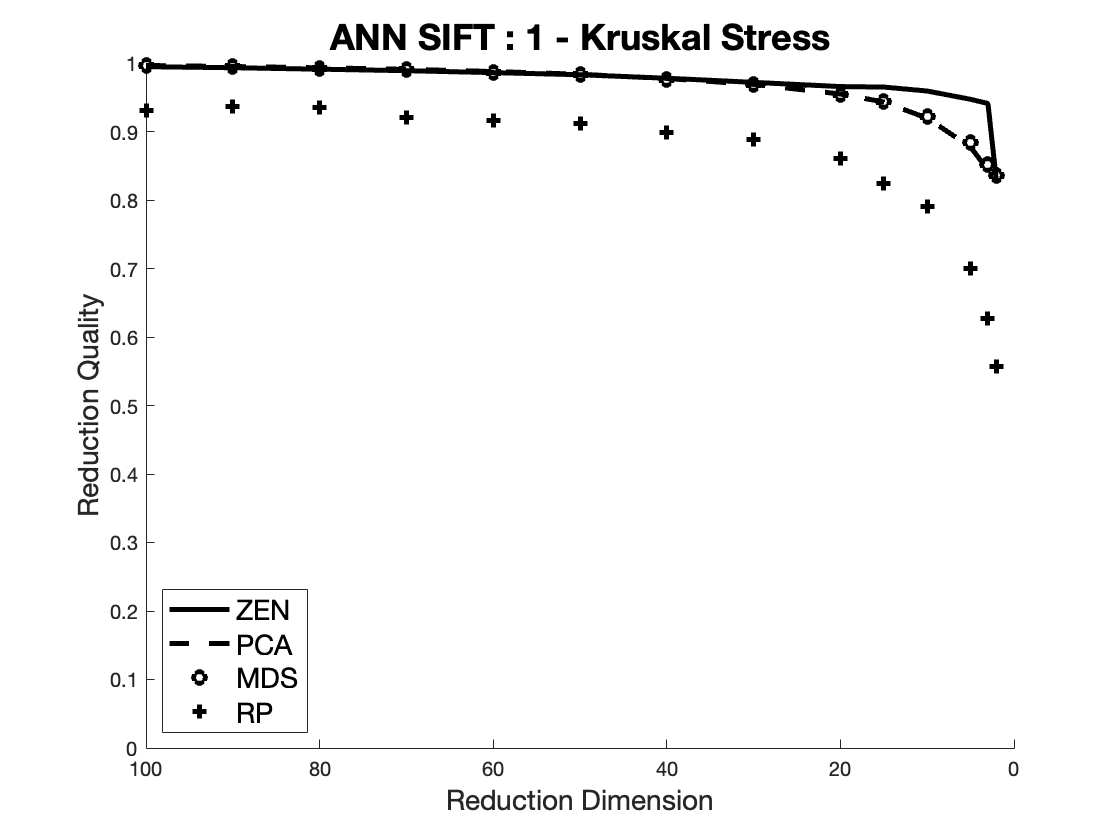}\hfill
\caption{Kruskal stress}
\end{subfigure}
\begin{subfigure}{0.32\textwidth}
\includegraphics[width=\textwidth]{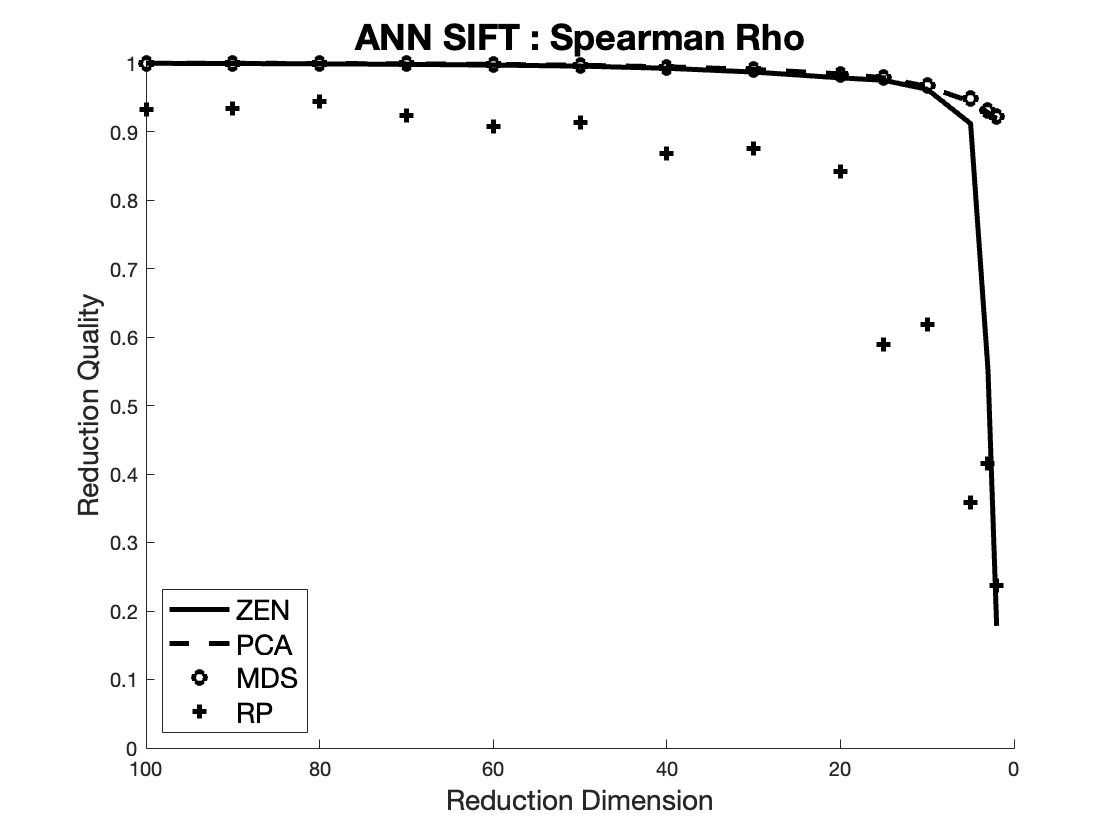}
\caption{Spearman Rho}
\end{subfigure}
\begin{subfigure}{0.32\textwidth}
\includegraphics[width=\textwidth]{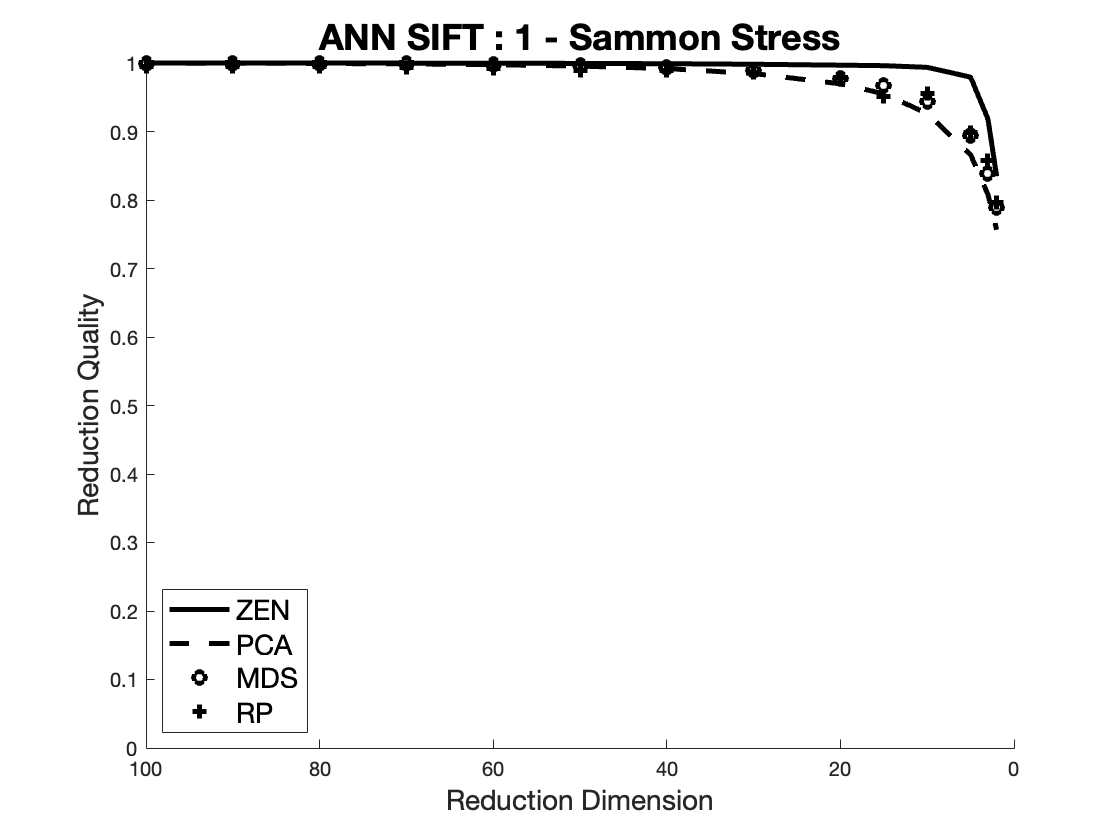}\hfill
\caption{Sammon stress}
\end{subfigure}
\begin{subfigure}{0.32\textwidth}
\includegraphics[width=\textwidth]{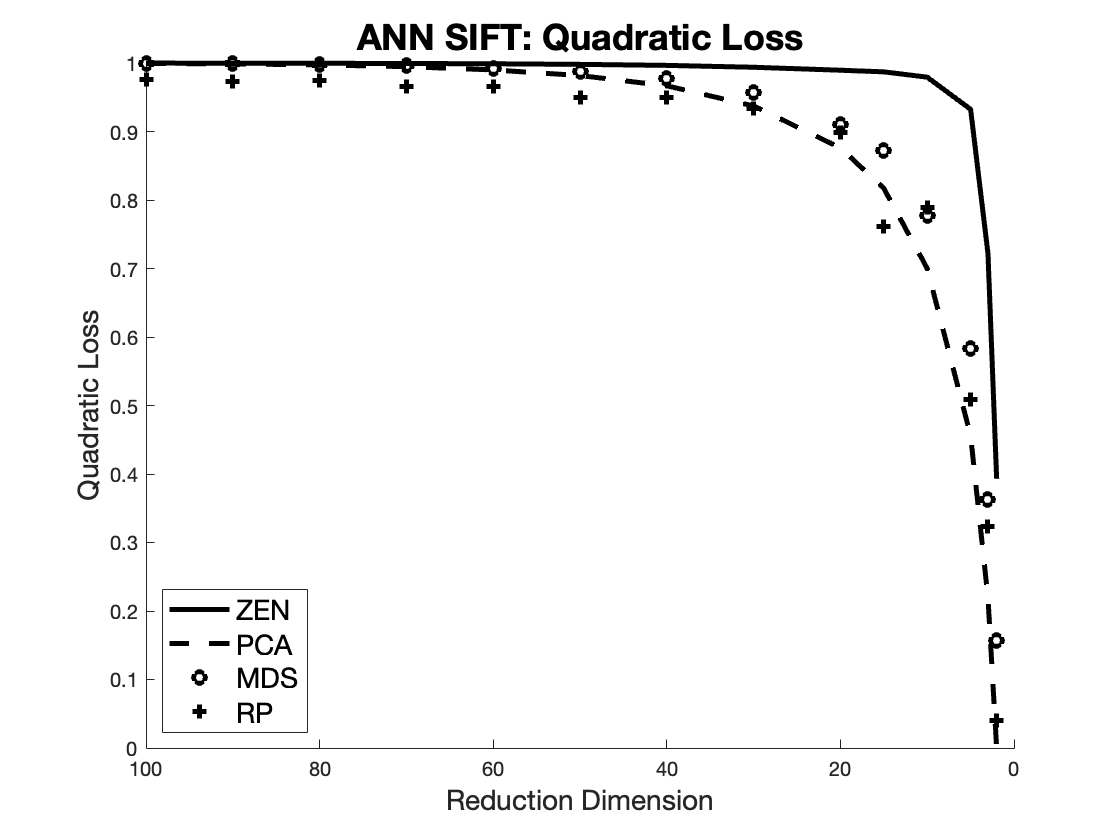}
\caption{Quadratic loss}
\end{subfigure}
\begin{subfigure}{0.32\textwidth}
\includegraphics[width=\textwidth]{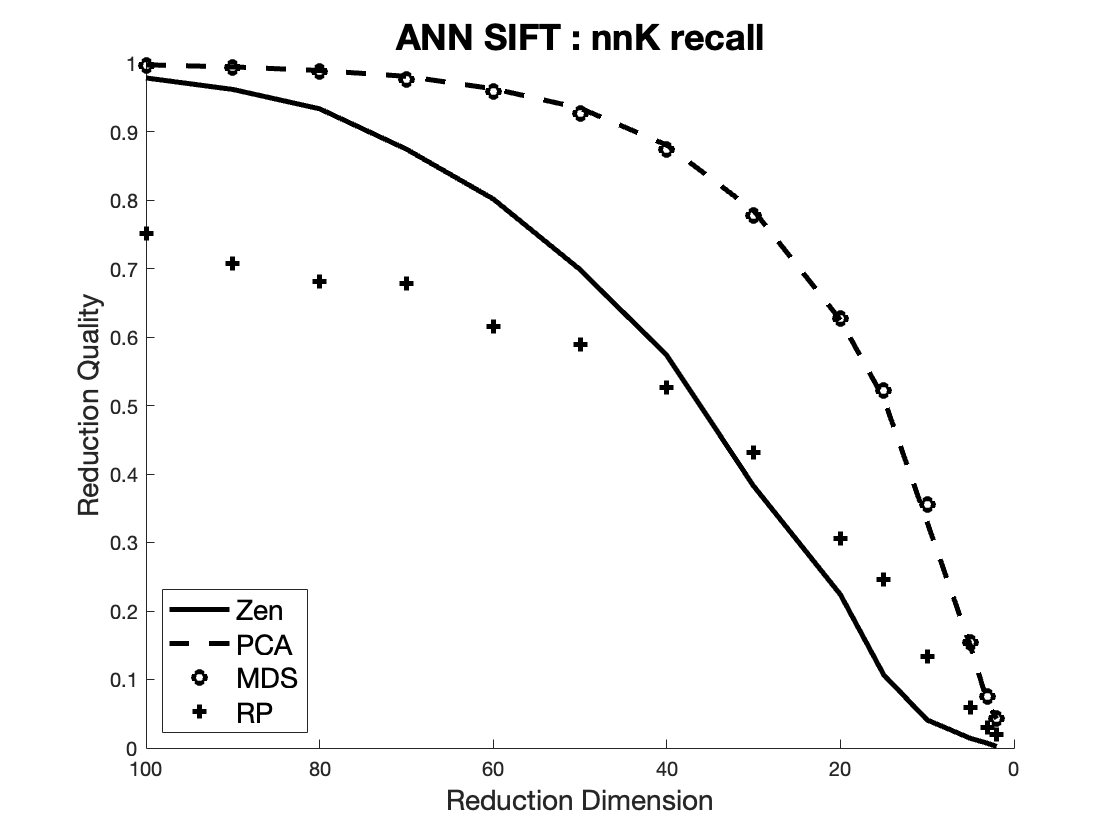}\hfill
\caption{kNN Recall}
\end{subfigure}
\caption{Quality metrics for  ANN SIFT reduced from 128 to between 100 and 2 dimensions.}
\label{fig_ann_sift_quality}
\end{figure}

Figure \ref{fig_ann_sift_quality} shows the  quality measures for the data, for dimensions reducing from 100 down to 2 as, with the exception of RP, there is almost no quality loss at above 100 dimensions. As can be seen, in these tests \nsimp \zen performs best for Kruskal stress, Sammon stress and quadratic loss, but is marginally less good than either PCA or MDS at lower dimensions for Spearman Rho, and is strikingly less good than either for recall at all dimensions.

We do not have a categorical reason for these observations, but believe that there are two main reasons for this difference in performance: first, the data set to start with has a relatively low intrinsic dimensionality, and second, the data  lies on a relatively regular linear manifold within the representational space. These observations are certainly consistent with the relatively small number of dimensions shown by PCA to capture the majority of the distance variance. The relatively poor performance of RP is also explained by these observations.

\subsubsection{MirFlickr  1m / Alexnet}
\label{sec_subsec_mf_relu_cos}

\begin{figure}[tbp]
\includegraphics[width=0.32\textwidth]{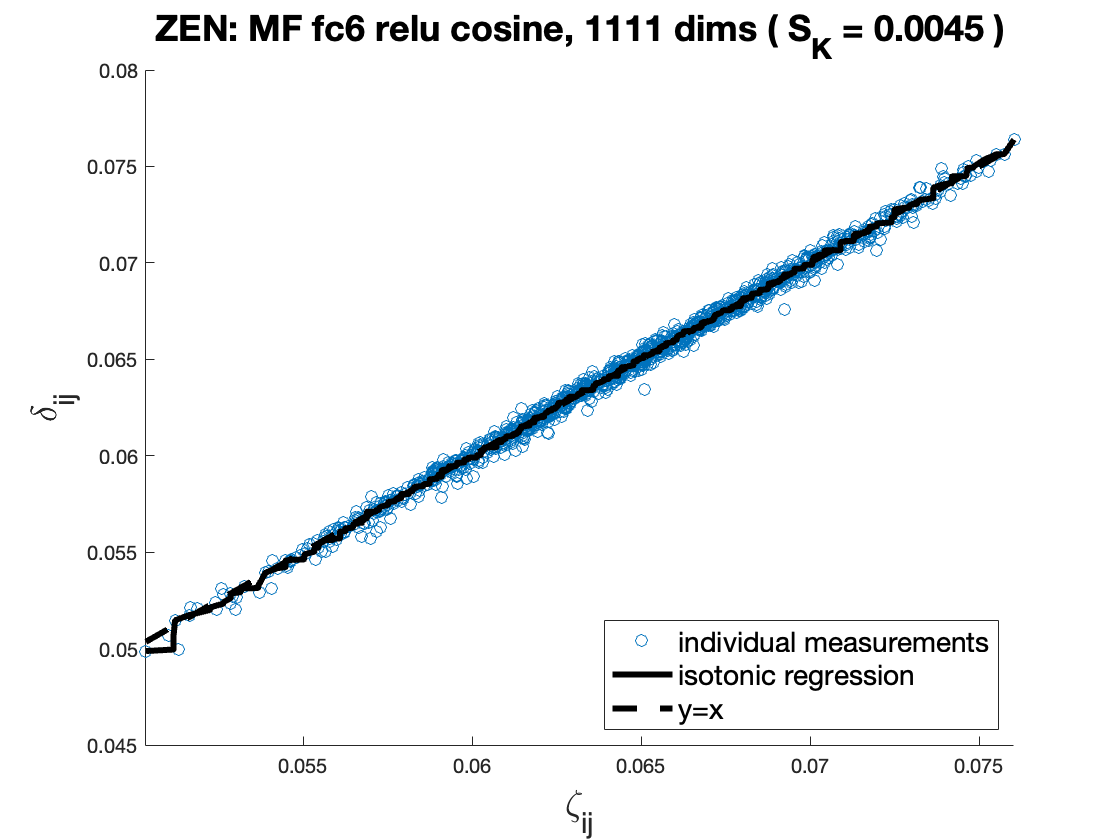} \hfill
\includegraphics[width=0.32\textwidth]{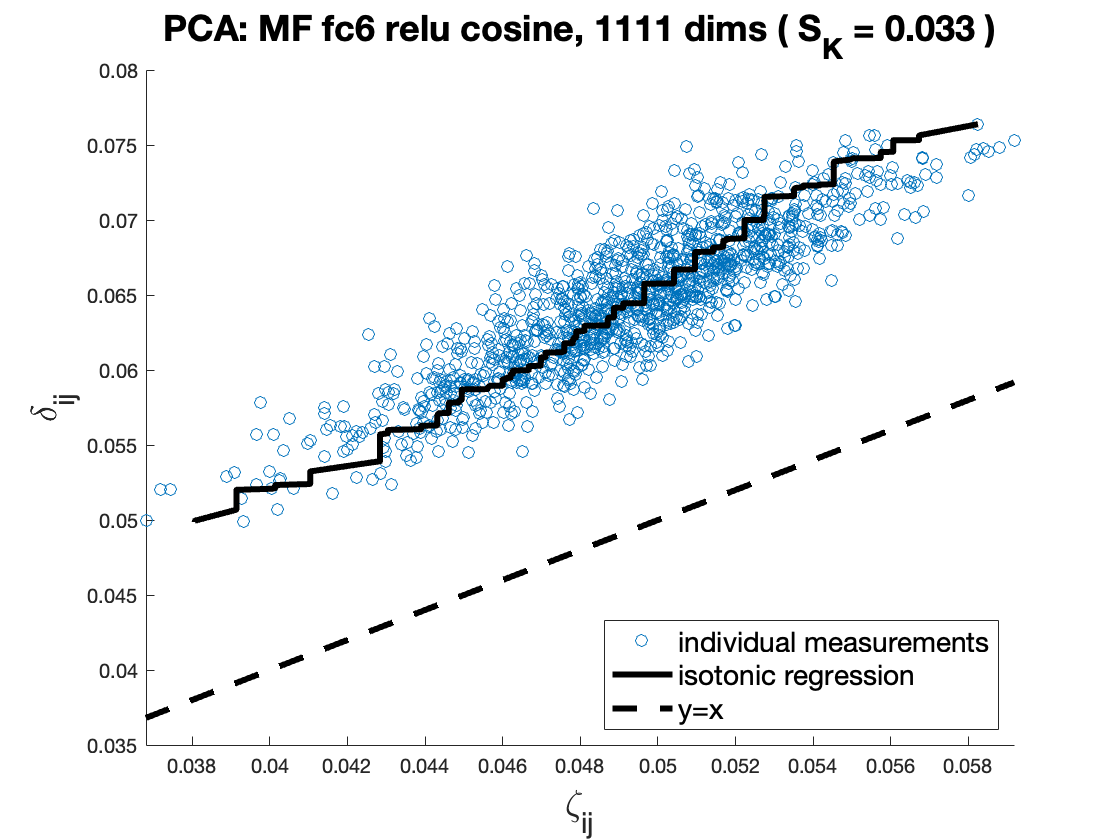} \hfill
\includegraphics[width=0.32\textwidth]{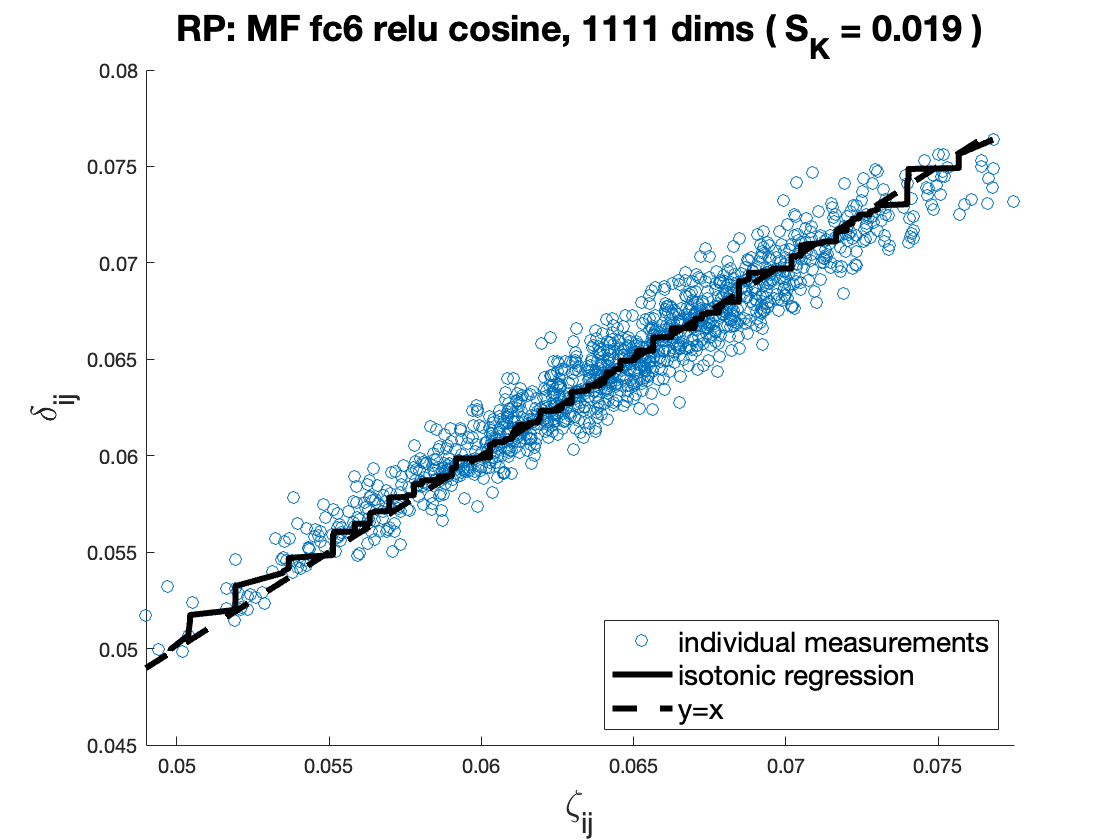}
\caption{fc6 RELU cos $\ell_2$-normed representations, reduced from 4096 dimensions to 1111.}
\label{fig_fc6_relu_cos_shepards}
\end{figure}
\begin{figure}[tbp]
\centering
\begin{subfigure}{0.32\textwidth}
\includegraphics[width=\textwidth]{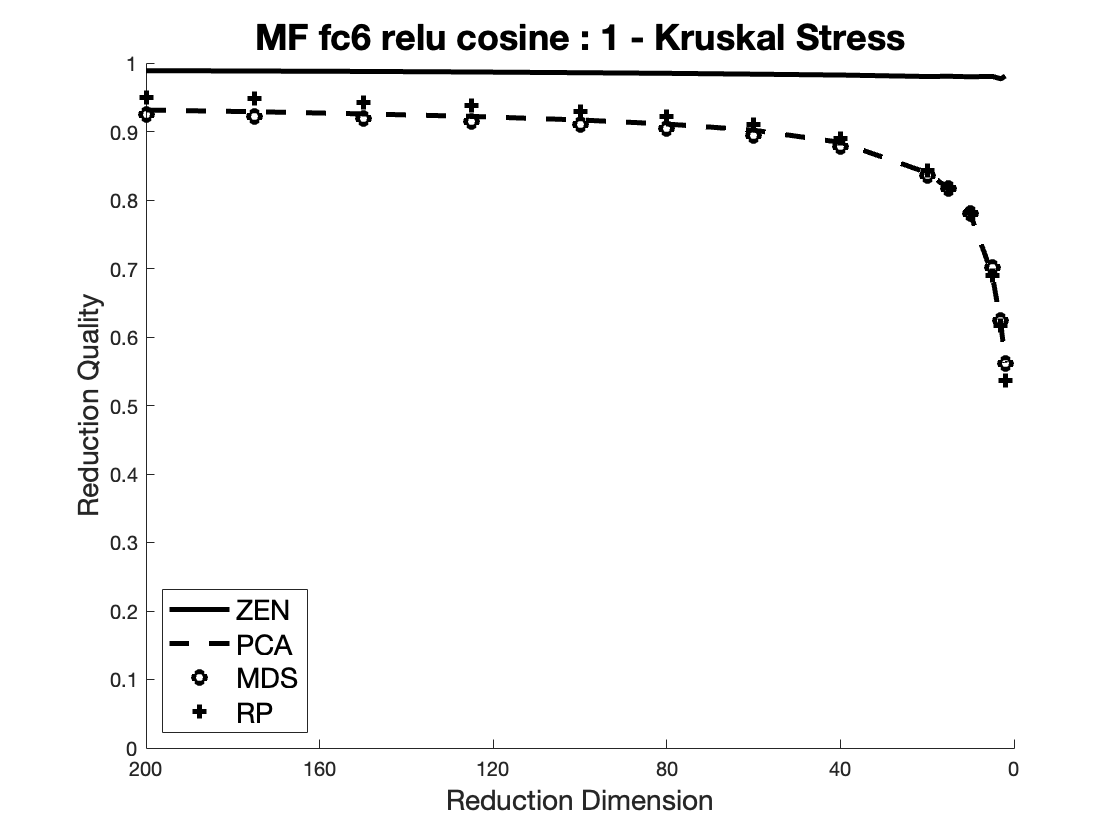}\hfill
\caption{Kruskal stress}
\end{subfigure}
\begin{subfigure}{0.32\textwidth}
\includegraphics[width=\textwidth]{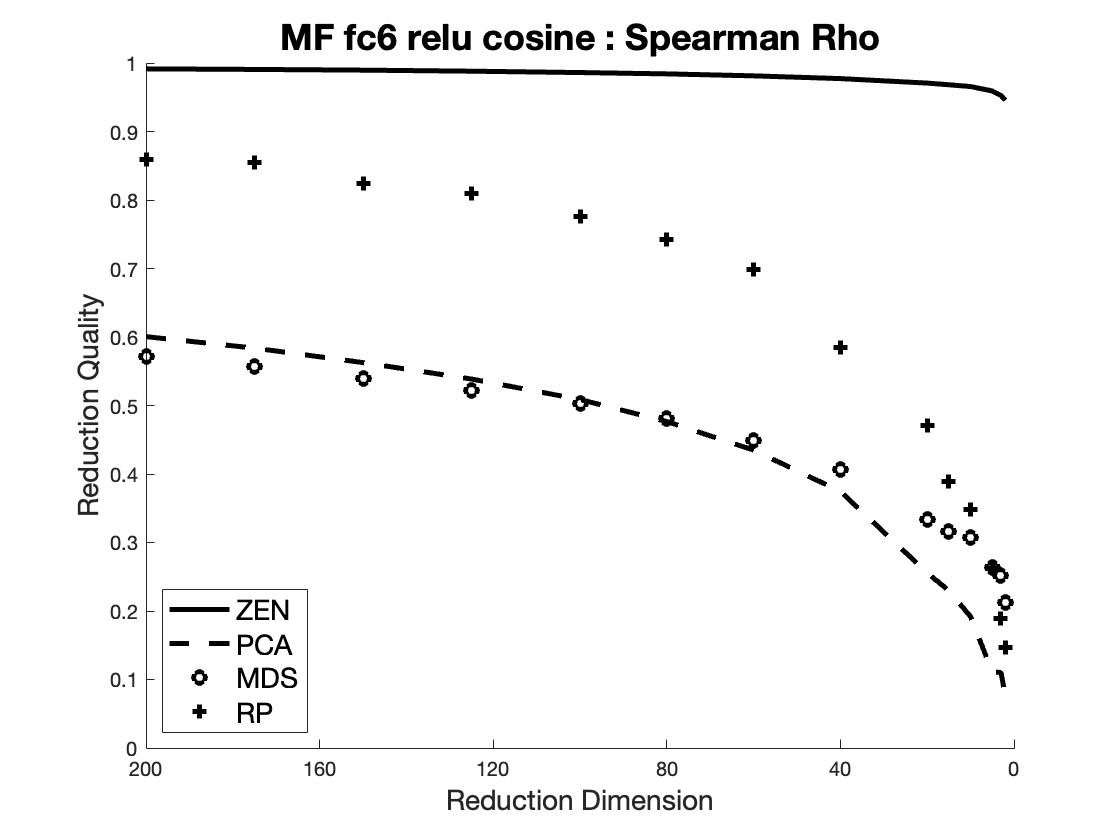}
\caption{Spearman Rho}
\end{subfigure}
\begin{subfigure}{0.32\textwidth}
\includegraphics[width=\textwidth]{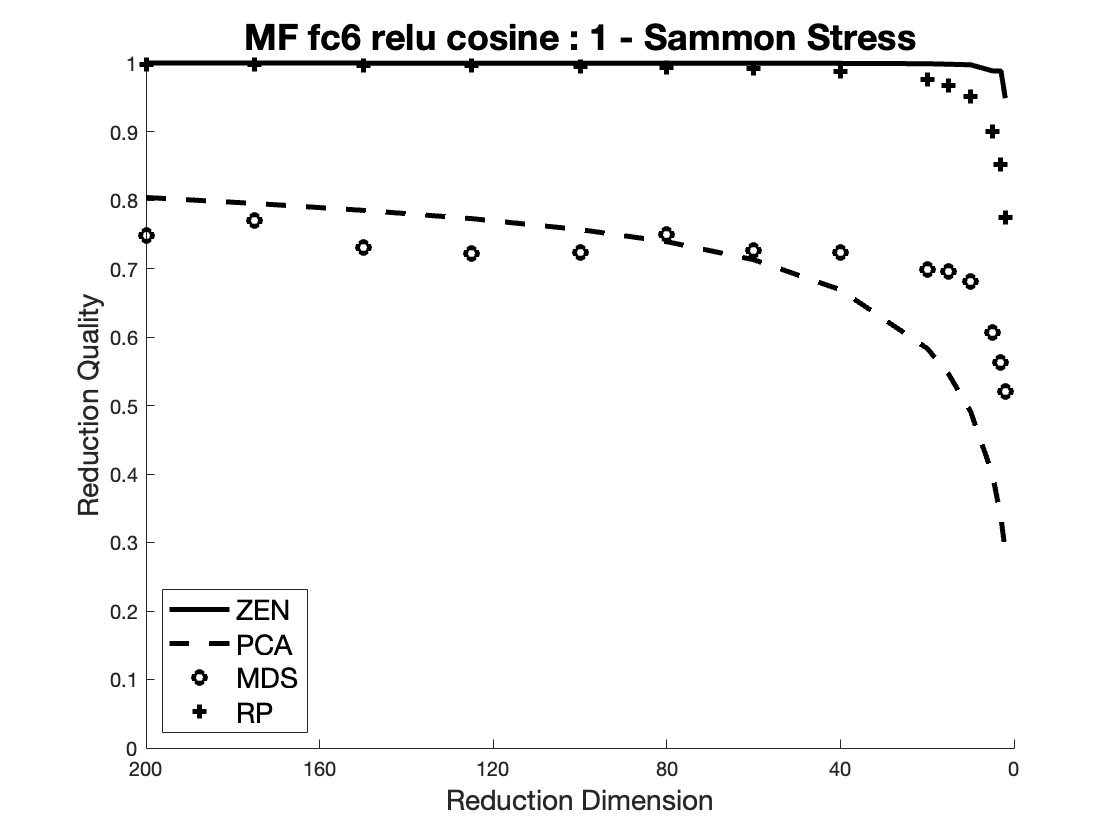}\hfill
\caption{Sammon stress}
\end{subfigure}
\begin{subfigure}{0.32\textwidth}
\includegraphics[width=\textwidth]{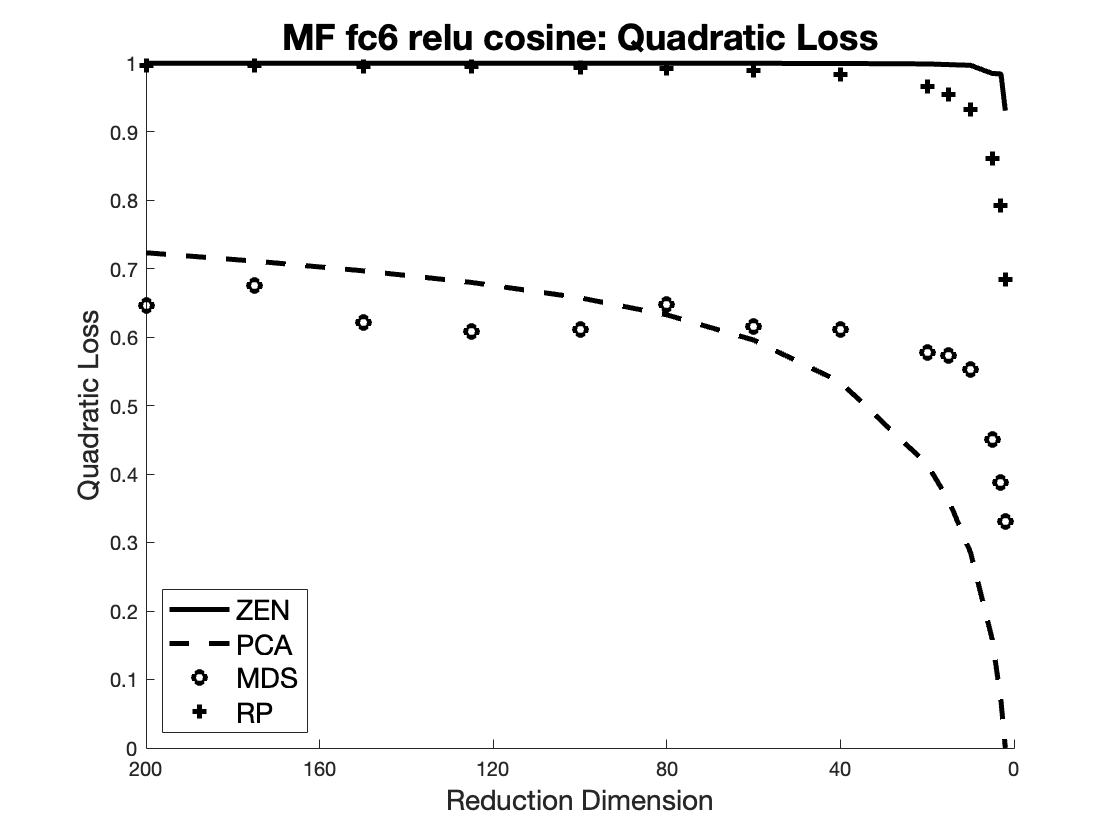}
\caption{Quadratic loss}
\end{subfigure}
\begin{subfigure}{0.32\textwidth}
\includegraphics[width=\textwidth]{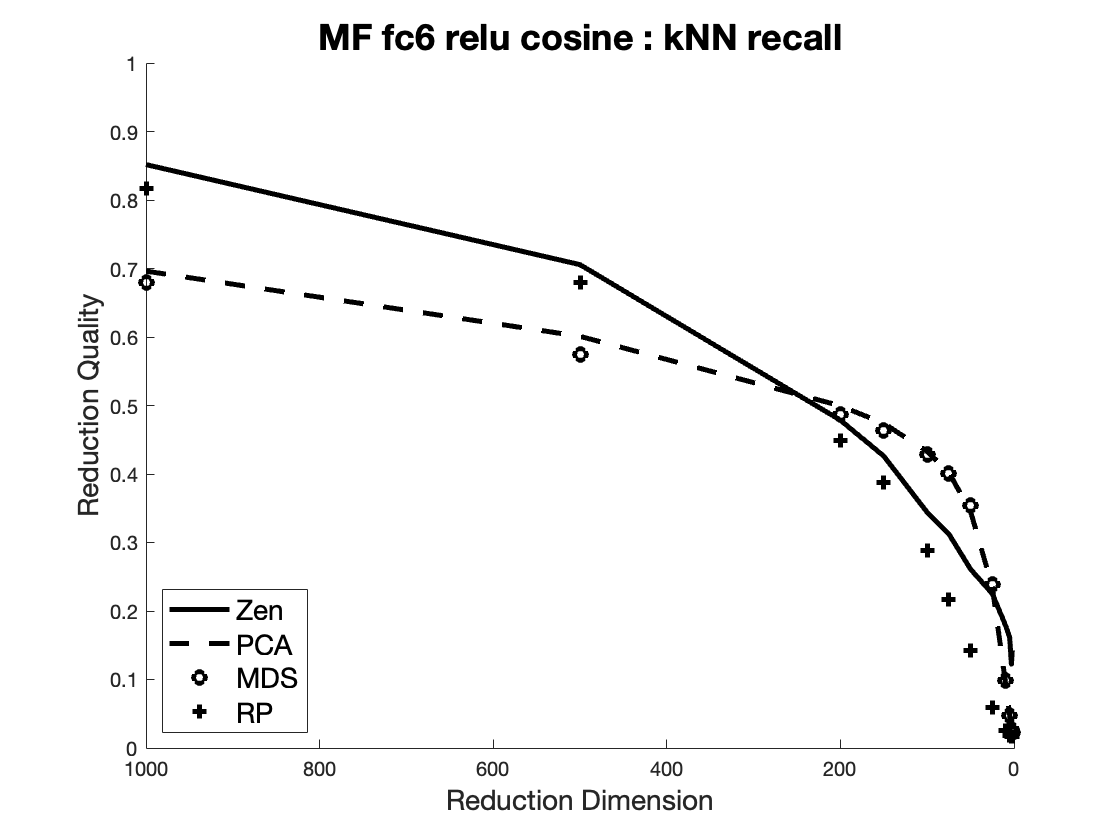}\hfill
\caption{kNN Recall}
\end{subfigure}
\caption{Quality metrics for  fc6 RELU Cosine distance over MirFlickr fc6 data. Apart from recall, all charts plot reductions from the original 4096 to between 200 and 2 dimensions, as \nsimp \zen gives almost perfect results at 200 dimensions. Recall is plotted from between 1000 and 2 dimensions. 1111 dimensions capture 80\% of the variance when using PCA.}
\label{fig_fc6_relu_quality}
\end{figure}

In this section we use the same raw data as that of Section \ref{sec_subsec_MF_Alex}, but with the data converted to give a proxy for the classic ``cosine'' distance after the RELU filter has been applied.  RELU is applied by zeroing out the negative values, then the resulting points are $\ell_2$-normalised by dividing each vector component by the magnitude of the resulting vector. After this transformation, the Euclidean  distance between the values, which now represent the end-points of unit vectors, gives the same rank ordering as Cosine distance over the post-RELU space.

The value of this metric is that it is that used for training the original network, with this transform being applied at each fully-connected layer of the categorisation section of the CNN architecture. It might therefore be expected to give an improved performance in terms of semantic similarity over the simpler \emph{DeCAF} measure, although testing this in practice for a large data set is challenging.

However, as before, our purpose for this data is to give a convincingly realistic large set of values which derive from some application and which require the use of the Cosine metric, which is thus achieved. As previously mentioned, the characteristics of the metric space produced in this way are quite different to those for the same raw data under simple Euclidean distance.

For this data the PCA eigenvalues show that $1,111$ dimensions are required to explain 80\% of the variance in distance according to Eq. \eqref{eqn_pca_variance}, a surprising departure from the 109 dimensions required for the non-RELU $\ell_2$ metric version.
Figure \ref{fig_fc6_relu_cos_shepards} shows Shepard plots at this dimension. For the first time we see that RP outperforms PCA and MDS, implying that no useful information about the manifold containing this data is gleaned from either of these analyses. In fact both can in fact be seen to be harmful, as the randomly generated RP transform, which is approximately orthonormal at this higher dimension, outperforms both. In this case the relatively small stress caused by the reduction transform can only be attributed to the effect highlighted by the Johnson-Lindenstrauss lemma, and this effect is somehow being lessened by linear analysis of the original manifold.

Again, however, our main purpose is to compare the \nsimp \zen transform, which in this case significantly outperforms any of the other three mechanisms. As can be seen, the Kruskal stress in this case is almost an order of magnitude less than for PCA.

In fact, the high quality of the \nsimp \zen transform is maintained down to much lower dimensions. Figure \ref{fig_fc6_relu_quality} shows the first four quality measures applied to reductions of between 200 and 2 dimensions; the starting point of 200 is used as, with the exception of recall, there is almost no loss of quality with \nsimp \zen at any of these dimensions. In all cases it can be seen that \nsimp \zen is the best mechanism, along with the observation that RP is better than either PCA or MDS again across the whole range of reduction dimensions. For Kruskal stress and Spearman Rho measures, \nsimp \zen is much better than RP.

The only exception to this is the recall measure, where it can be seen that \nsimp \zen starts to perform less well than either PCA or MDS when the reduction dimensions is less than around 300. While \nsimp \zen is better than RP across the whole range, the advantage is only relatively small; again, RP gives a surprisingly good outcome for this test.

\subsection{Other Hilbert spaces - Jensen-Shannon Distance}\label{sec_sub_exp_jsd}

In the final experimental section, we apply \nsimp \zen to Hilbert spaces where there is no available coordinate system. As explained in Section \ref{sec_nsimp_projection_main}, the underlying \nsimp transform can be applied to any metric space which is isometrically embeddable in a Hilbert space. As the $n$-dimensional simplexes are constructed in Euclidean space using only the pairwise distances measured in the original space, then any metric space which allows a finite $n$-embedding into $(n-1)$ Euclidean dimensions can be used as the domain, and all metric spaces which are isometrically embeddable in Hilbert space have this property.

One of the most interesting classes in this category  is that of metric spaces governed by the Jensen-Shannon distance, an information-theoretic distance metric which has some interesting properties, and can reasonably  be regarded as a distance which should always be preferred to the semantics-free Cosine distance \cite{four_metrics}. One possible reason for its relatively low uptake may be that its calculation requires many $\log$ calculations, and can be two orders of magnitude slower than Cosine distance over the same dimensions. It is therefore  intriguing that a dimensionality reduction transform exists which not only reduces the size of the representations, but also converts the distance metric from an expensive calculation to a much cheaper one.

The absence of a coordinate system means that neither  PCA nor RP can be applied. While MDS can be applied to a small space, it is not possible to use the extended version described in  Section \ref{sec_related_sub_mds} which is necessary to allow its application to a large space. However in \cite{LMDS} it is shown that LMDS, a different extension  of the MDS principles, can be applied to any metric space, including those without a coordinate system. In this section we therefore compare \nsimp \zen and LMDS in use against two spaces governed by the Jensen-Shannon metric.

\subsubsection{100-dimensional generated space}

\begin{figure}[tbp]
\centering
\includegraphics[width=0.35\textwidth]{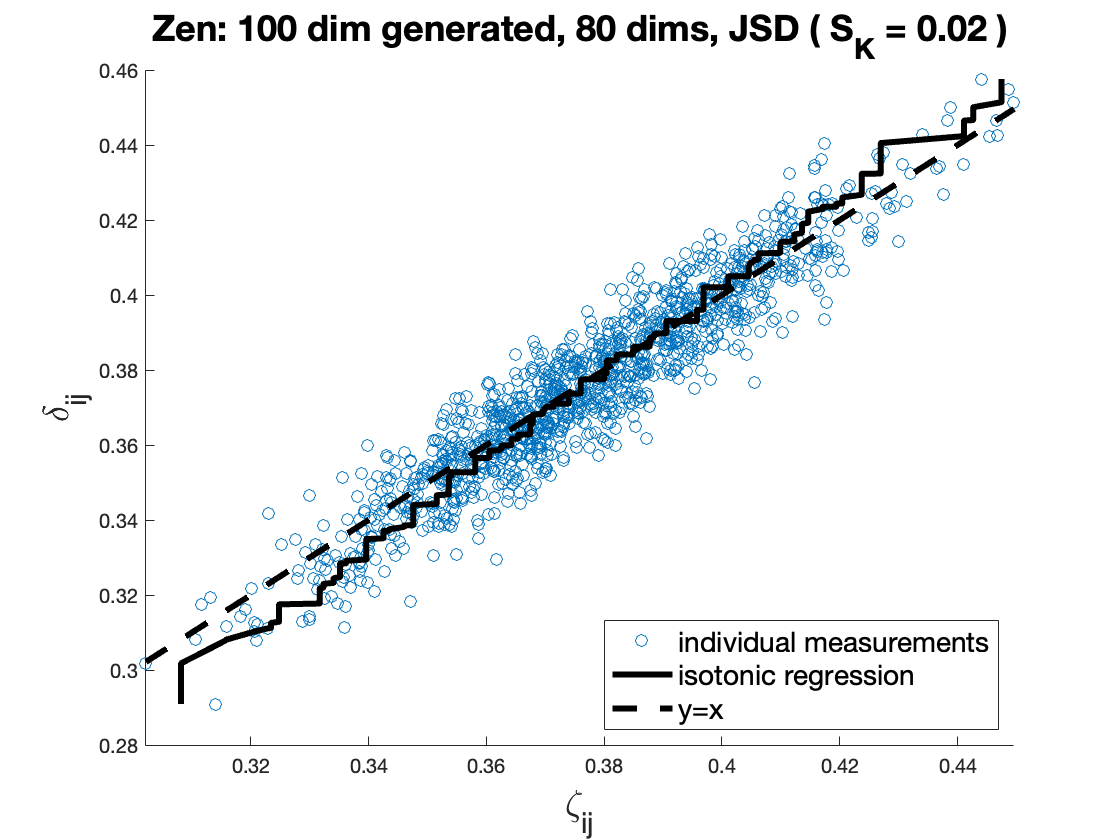} \qquad
\includegraphics[width=0.35\textwidth]{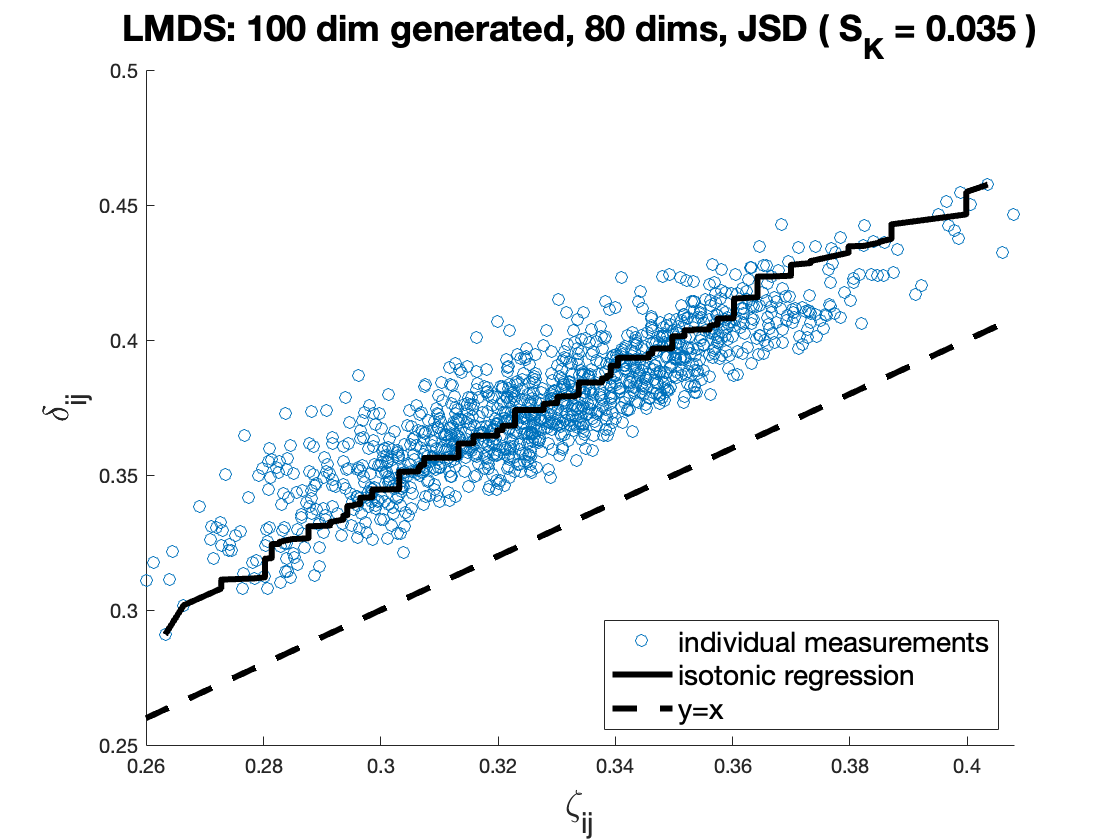} 
\caption{100-dimensional generated probability space with Jensen-Shannon metric, reduced to 80 dimensions using \nsimp \zen and LMDS.}
\label{fig_sift_gen_jsd_shepards}
\end{figure}

\begin{figure}[tbp]
\centering
\begin{subfigure}{0.32\textwidth}
\includegraphics[width=\textwidth]{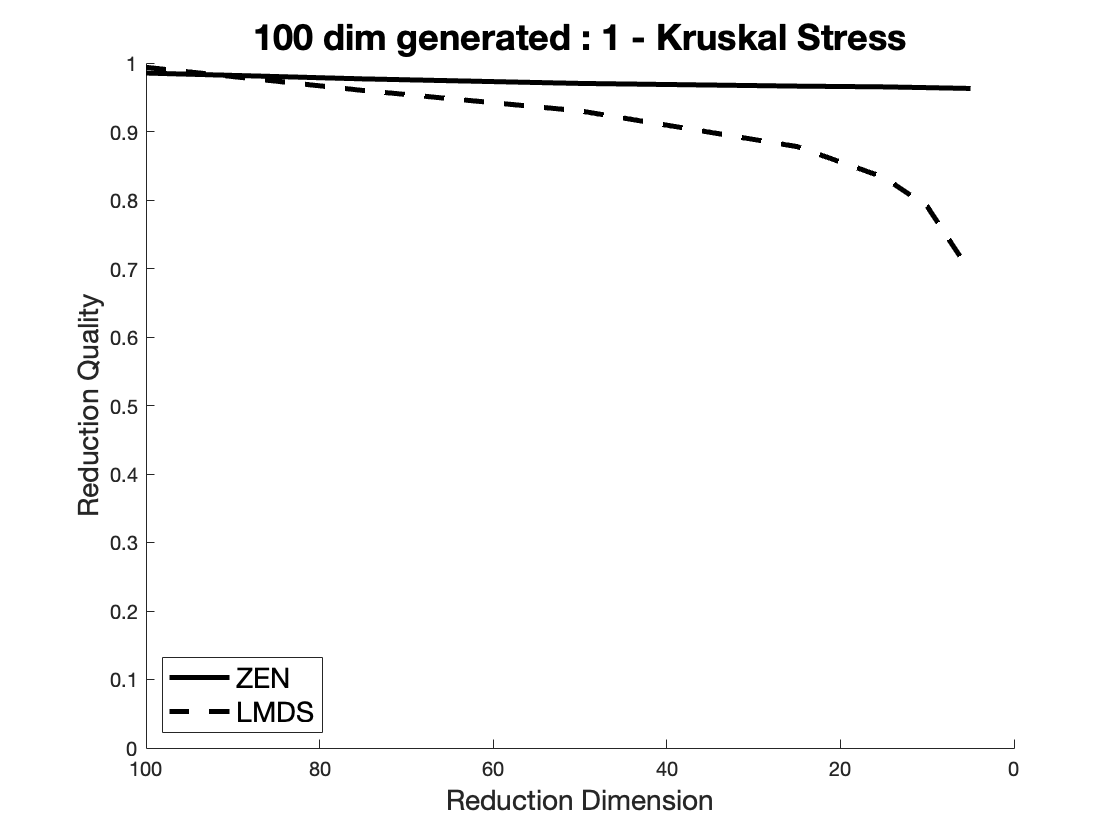}\hfill
\caption{Kruskal stress}
\end{subfigure}
\begin{subfigure}{0.32\textwidth}
\includegraphics[width=\textwidth]{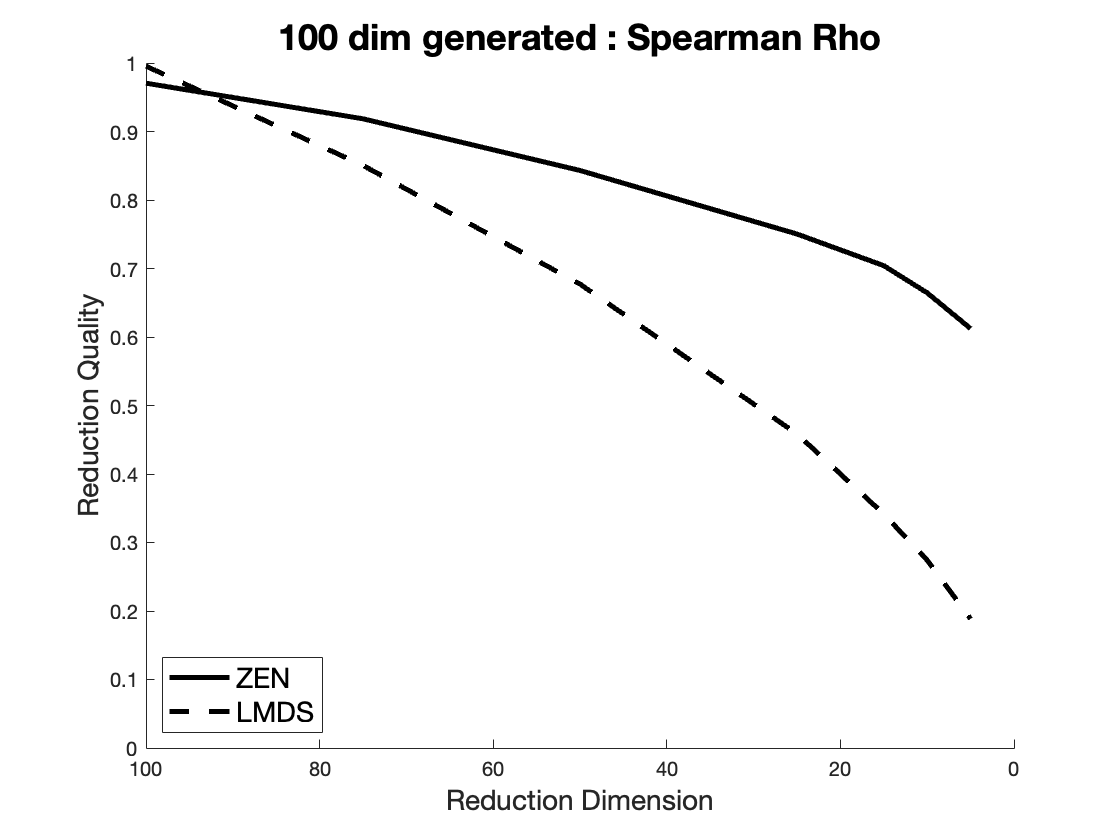}
\caption{Spearman Rho}
\end{subfigure}
\begin{subfigure}{0.32\textwidth}
\includegraphics[width=\textwidth]{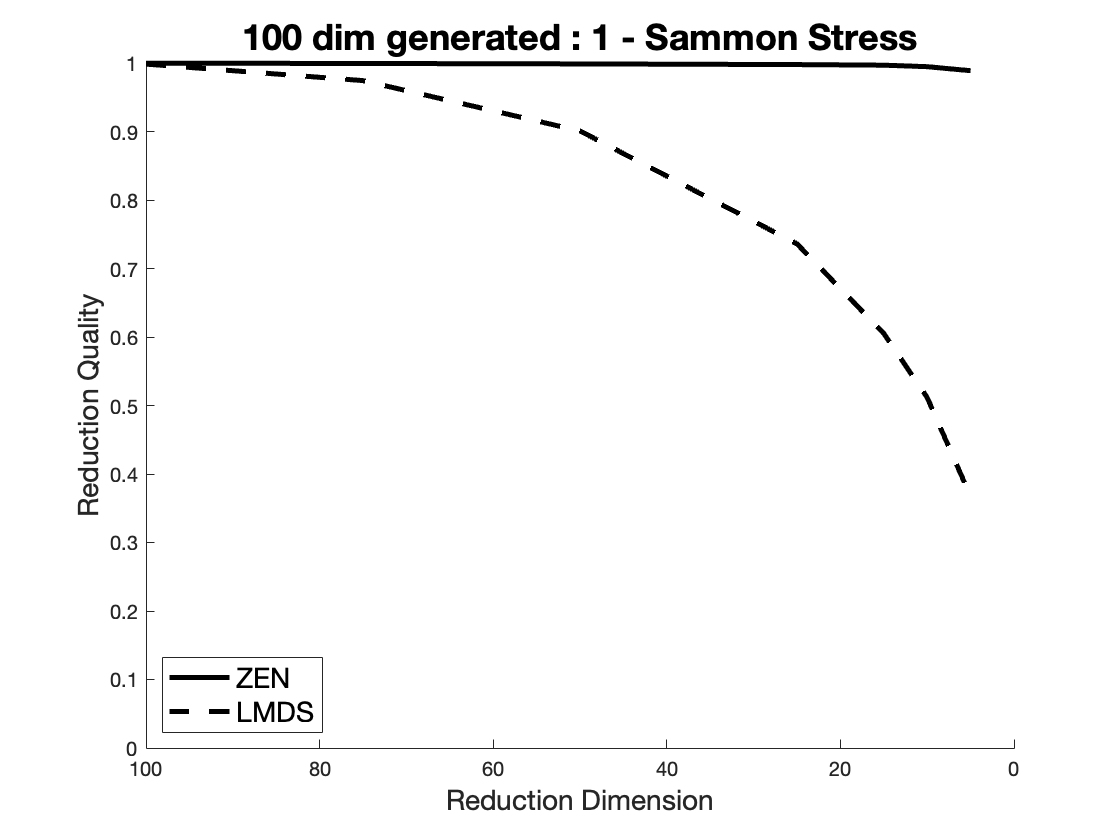}\hfill
\caption{Sammon stress}
\end{subfigure}
\begin{subfigure}{0.32\textwidth}
\includegraphics[width=\textwidth]{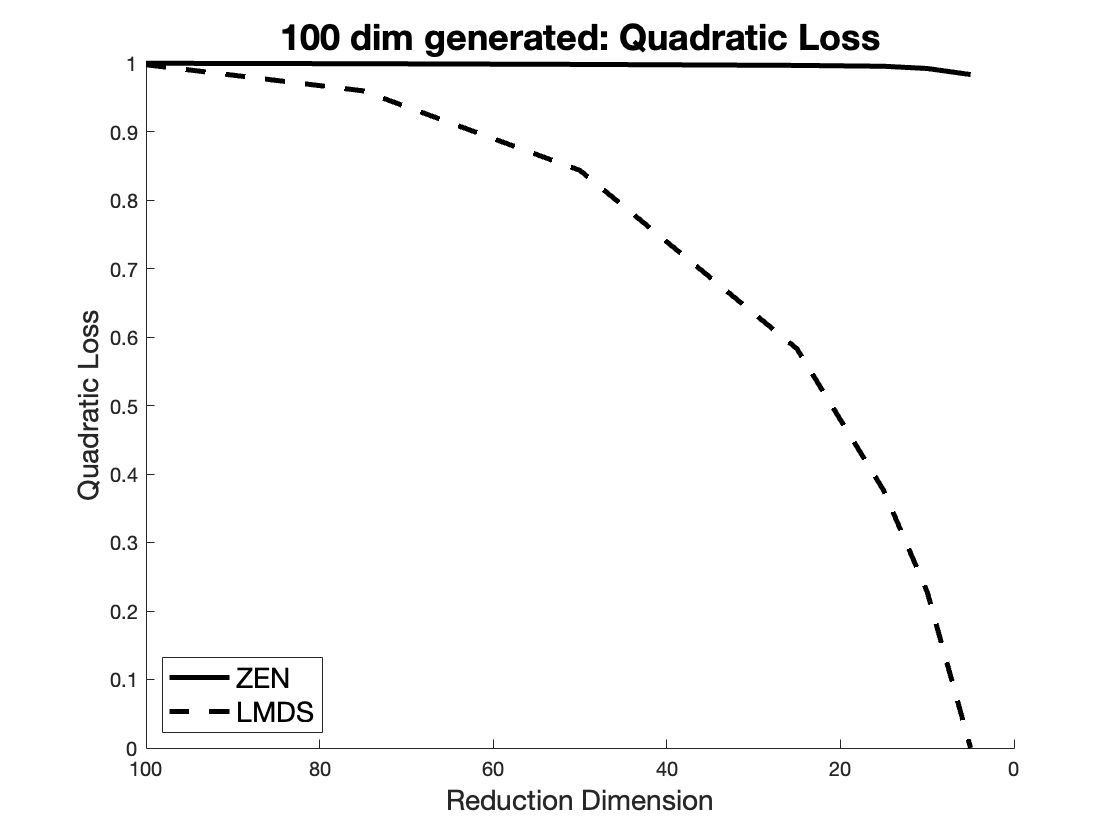}
\caption{Quadratic loss}
\end{subfigure}
\begin{subfigure}{0.32\textwidth}
\includegraphics[width=\textwidth]{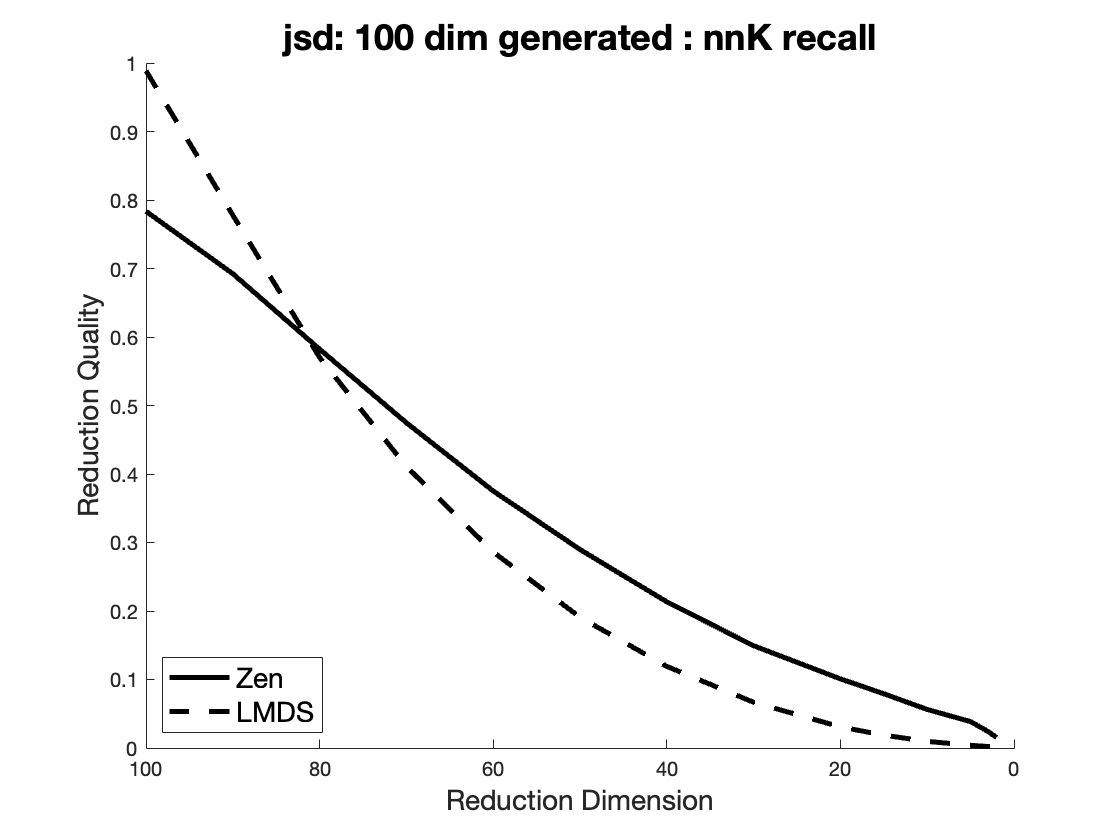}\hfill
\caption{kNN Recall}
\end{subfigure}
\caption{Quality metrics 100-dimensional generated probability space with Jensen-Shannon Distance.}
\label{fig_gen_jsd_quality}
\end{figure}
The first experiment uses a 100 dimensional generated space. 100-dimensional vectors are generated using a uniform random generator with each dimension, bounded in $[0,1]$, and each vector is $\ell_1$-normed  in order to simulate a  probability distribution over 100 independent variables, thus giving an appropriate domain for Jensen-Shannon distance.

Figure \ref{fig_sift_gen_jsd_shepards} shows Shepard plots for \nsimp \zen and LDMS at 80 dimensions (now an arbitrary figure as PCA is not possible over the data), from which it can be seen that \nsimp \zen gives less Kruskal stress than LMDS.

Figure \ref{fig_gen_jsd_quality} shows the usual quality charts across the range of 100 down to 2 dimensions. Again, \nsimp \zen is generally the better of the  mechanisms. It is interesting to note that \nsimp \zen does not give perfect results even at 100 dimensions, and that the only cases where LMDS outperforms \nsimp \zen are at  100 dimensions for the Spearman Rho test, and at about 80 dimensions for recall. While one property of a Hilbert space is a finite $n$-embeddability in $(n-1)$ Euclidean dimensions, this does not of course imply that an $n$-dimensional Jensen-Shannon space should in general be isometrically embeddable in an $n$-dimensional Euclidean space. Thus there is no reason to expect perfect performance when any space with Hilbert properties is ``reduced'' to a Euclidean space with the same physical dimensions.

\subsubsection{GIST}
Our final experiment is with Jensen-Shannon distance applied to GIST image descriptors. GIST \cite{gist} is  a  representation of the image based on a set of perceptual dimensions that represent the dominant spatial structure of a scene. Although  again GIST is no longer  the state of the art in image similarity, it has been shown that GIST representations used in conjunction with Jensen-Shannon distance gives an excellent technique for finding near-duplicate images for forensic purposes \cite{connor2016quantifying}, a specialist application quite different from more general image similarity.

The MirFlickr 1M image set was again used, and GIST representations were obtained. Each representation is a 480-dimensional vector, again $\ell_1$-normalisation is applied to achieve a set of one million values suitable as a domain for the Jensen-Shannon distance metric.

After initial analysis it was found that these representations are quite amenable to dimensionality reduction using both \nsimp \zen and LMDS, so Shepard diagrams were produced at the 100-dimensional reduction. As shown in Figure \ref{fig_gist_jsd_shepards}, both techniques give a relatively low Kruskal stress even at around one-fifth or their initial size, and again the \nsimp \zen transform gives a significantly lower stress than LMDS.

Figure \ref{fig_gist_jsd_quality} repeats the usual quality analysis over these descriptors. Most charts are from 100 down to 2 dimensions again as there is little quality loss at higher dimensions; recall is measured between 200 and 2 dimensional reductions.

For this space, \nsimp \zen is  better across all dimensionalities for all tests other than for recall, where it is not as good as LMDS. It is noteworthy that both techniques, relying on distances alone, perform much better for the ``real'' data than for the uniformly generated data. The reason for this is presumably that the data is contained within a manifold contained within the representational space whose characteristics are being usefully captured by the distance-based analysis.

\begin{figure}[tbp]
\centering
\includegraphics[width=0.35\textwidth]{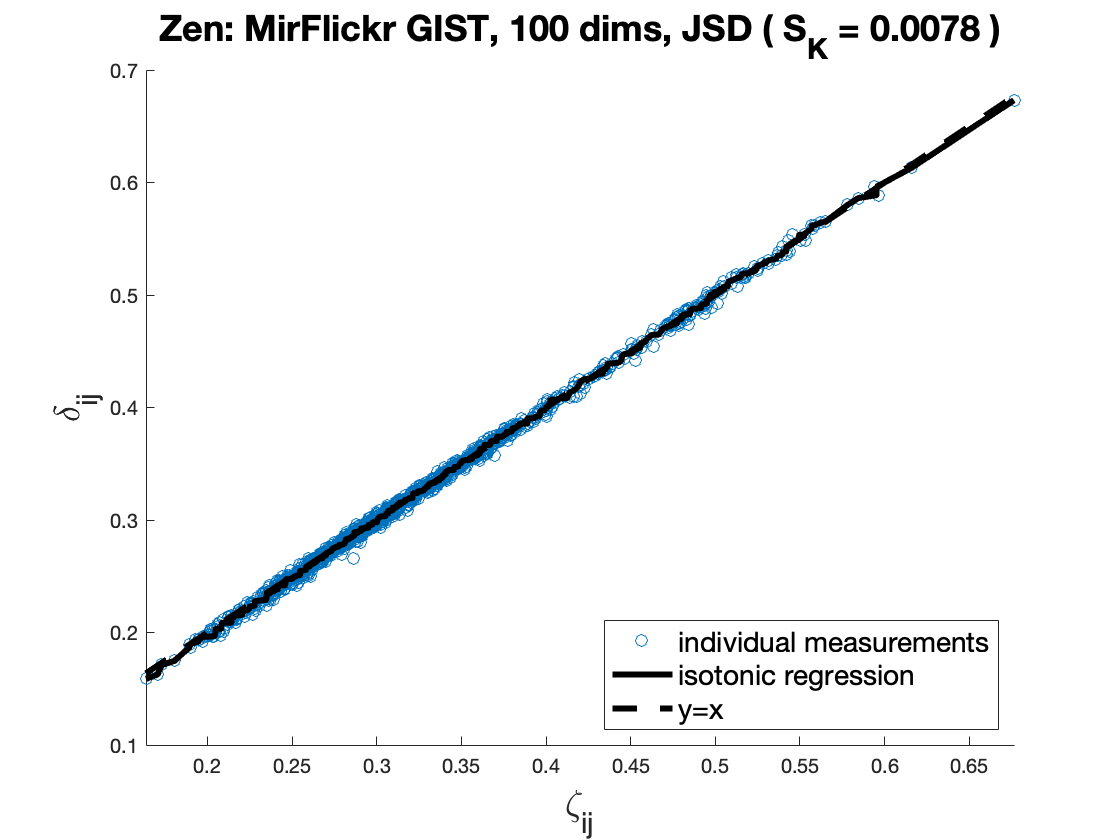} \qquad
\includegraphics[width=0.35\textwidth]{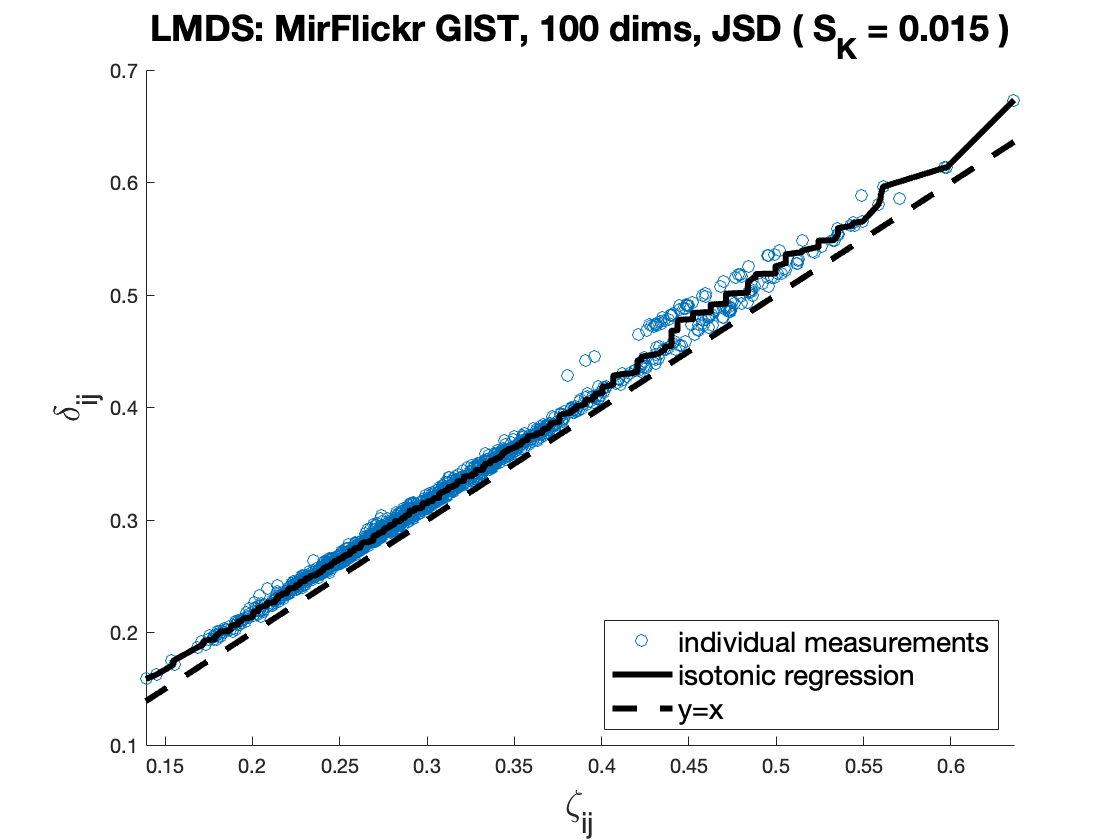} 
\caption{MirFlickr GIST with Jensen-Shannon metric, reduced to 400 dimensions using \nsimp \zen and Landmark MDS.}
\label{fig_gist_jsd_shepards}
\end{figure}

\begin{figure}[tbp]
\centering
\begin{subfigure}{0.32\textwidth}
\includegraphics[width=\textwidth]{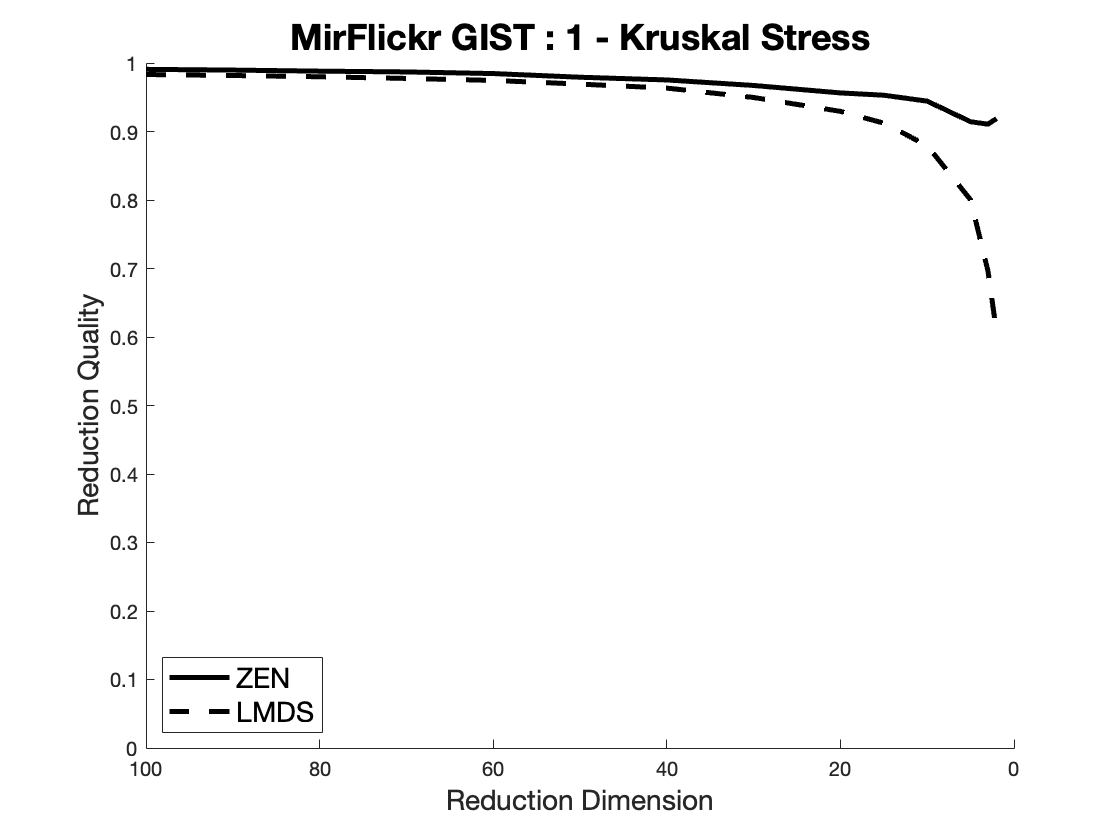}\hfill
\caption{Kruskal stress}
\end{subfigure}
\begin{subfigure}{0.32\textwidth}
\includegraphics[width=\textwidth]{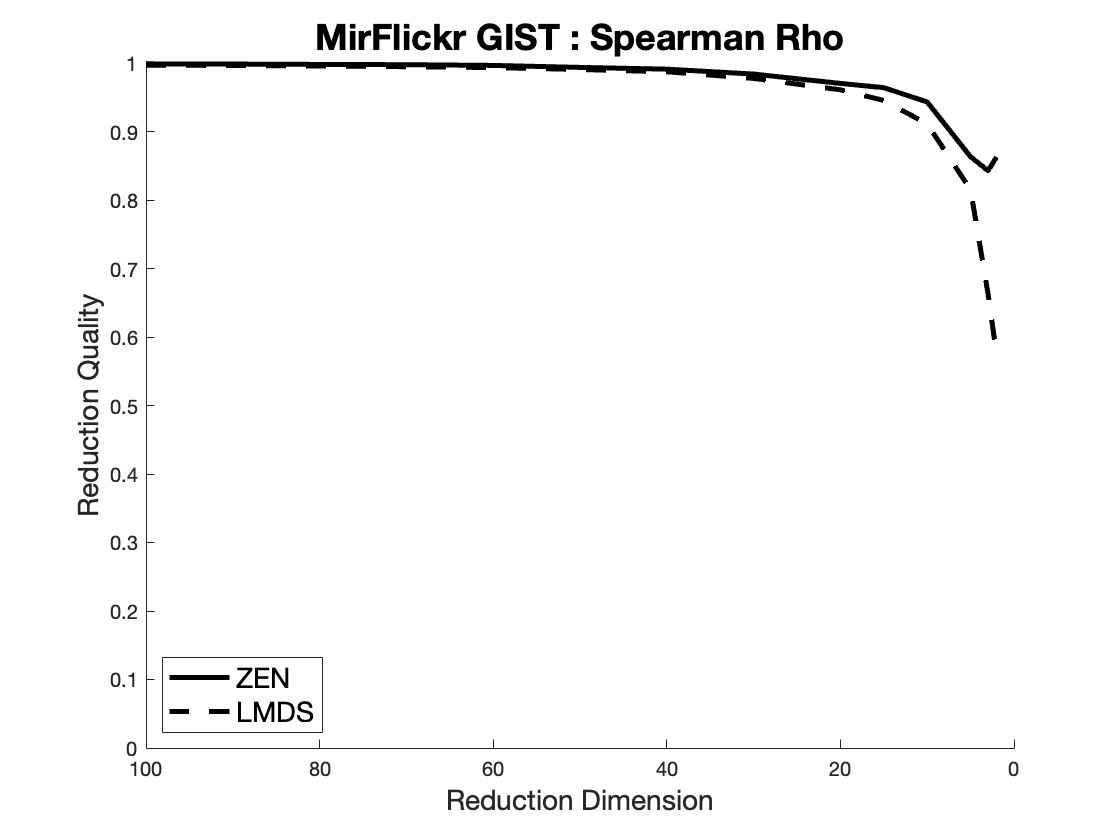}
\caption{Spearman Rho}
\end{subfigure}
\begin{subfigure}{0.32\textwidth}
\includegraphics[width=\textwidth]{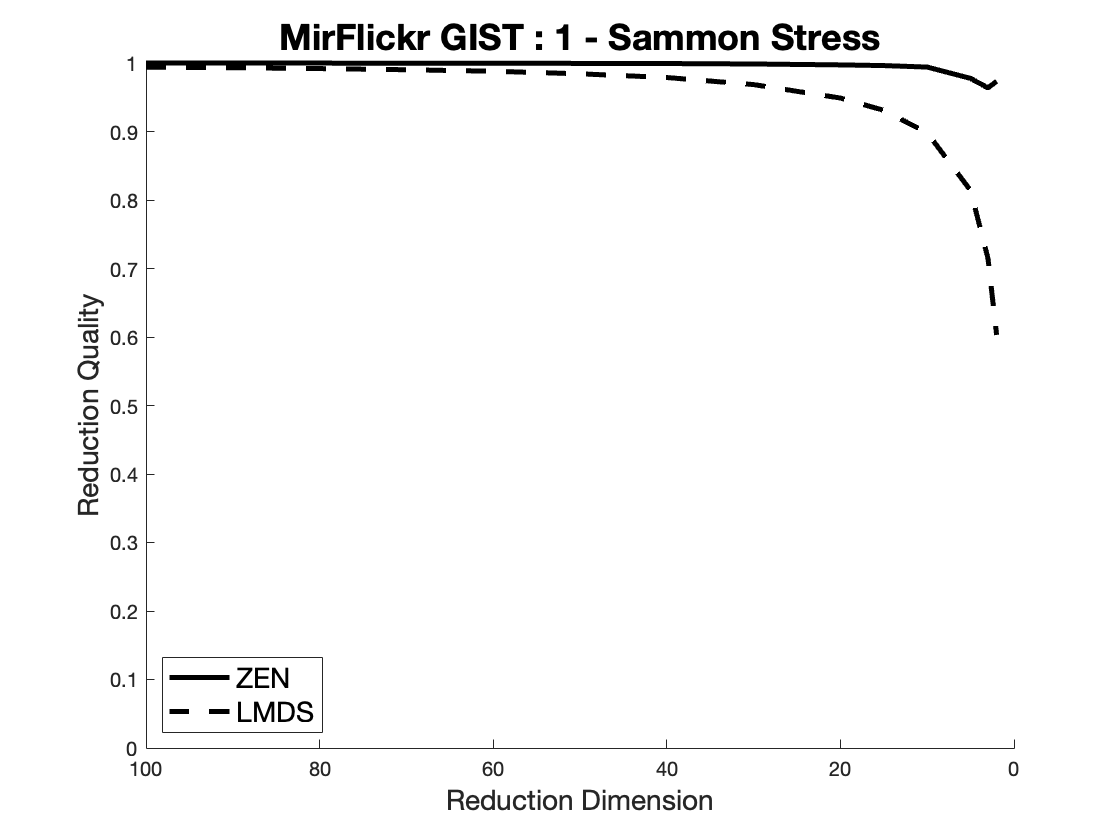}\hfill
\caption{Sammon stress}
\end{subfigure}
\begin{subfigure}{0.32\textwidth}
\includegraphics[width=\textwidth]{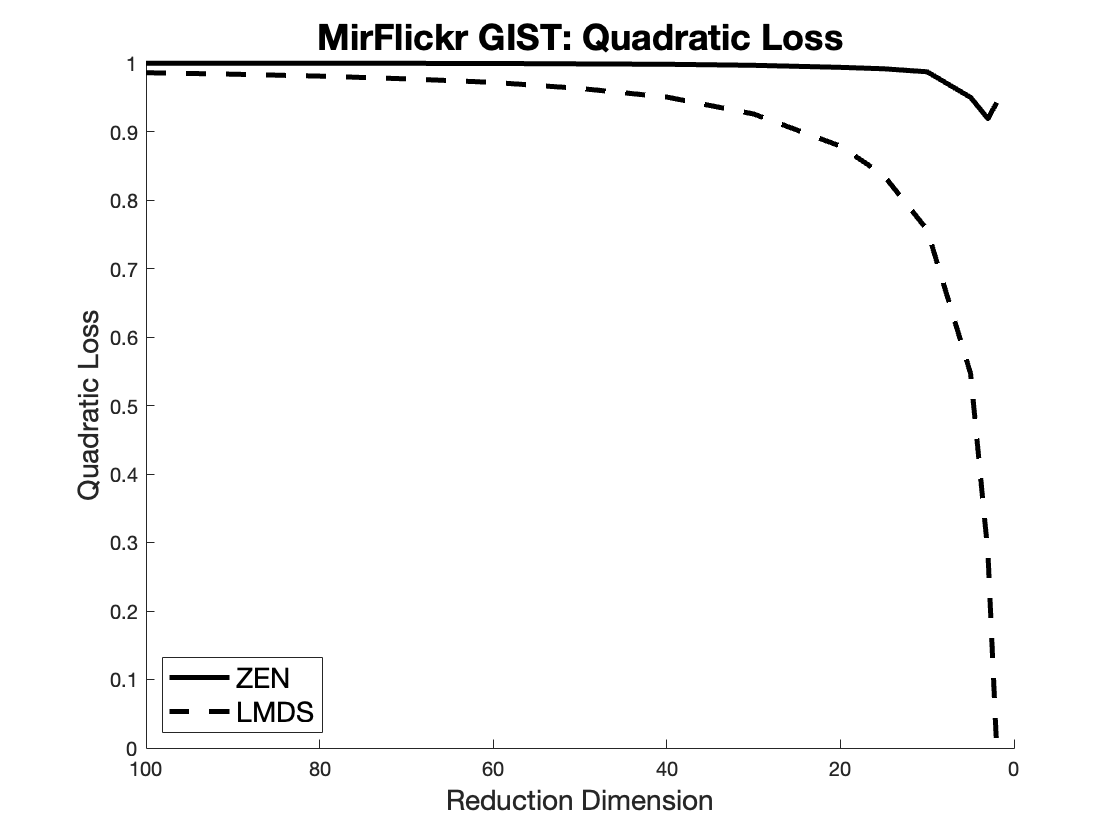}
\caption{Quadratic loss}
\end{subfigure}
\begin{subfigure}{0.32\textwidth}
\includegraphics[width=\textwidth]{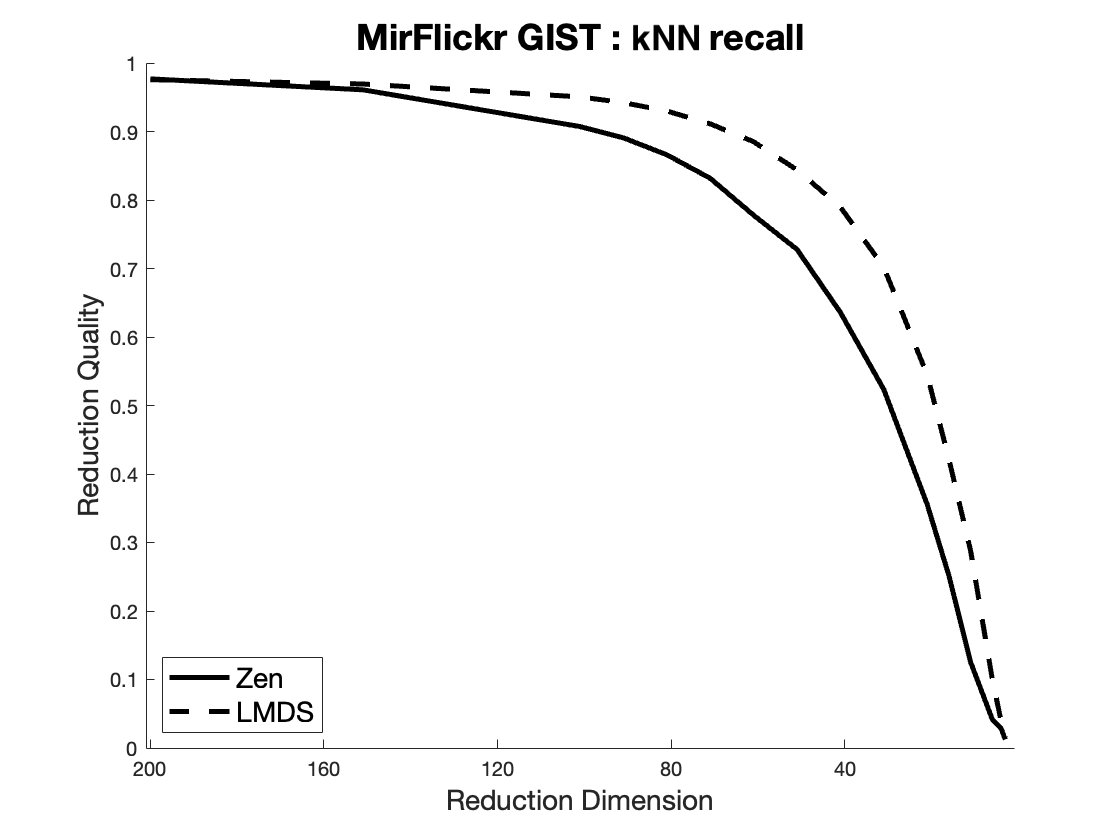}\hfill
\caption{kNN Recall}
\end{subfigure}
\caption{Quality metrics for  GIST/JSD. The data is 480 dimensions, however there is very little quality loss when reduced  to 100 dimensions, therefore most charts are plotted from 100 down to 2 dimensions; kNN recall is plotted from 200 dimensions downwards.}
\label{fig_gist_jsd_quality}
\end{figure}

\section{Run-time Costs}
\label{sec_performance}

\begin{figure}[tbp]
\includegraphics[width=0.45\textwidth]{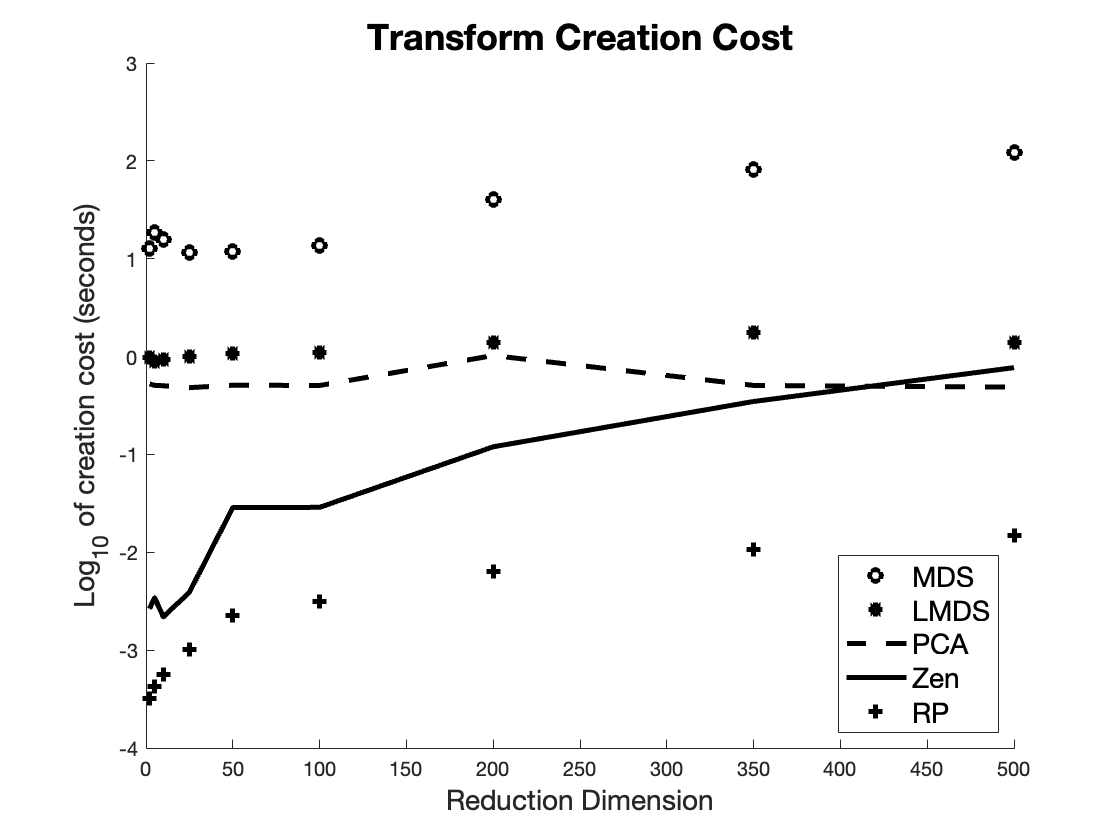} \hfill
\includegraphics[width=0.45\textwidth]{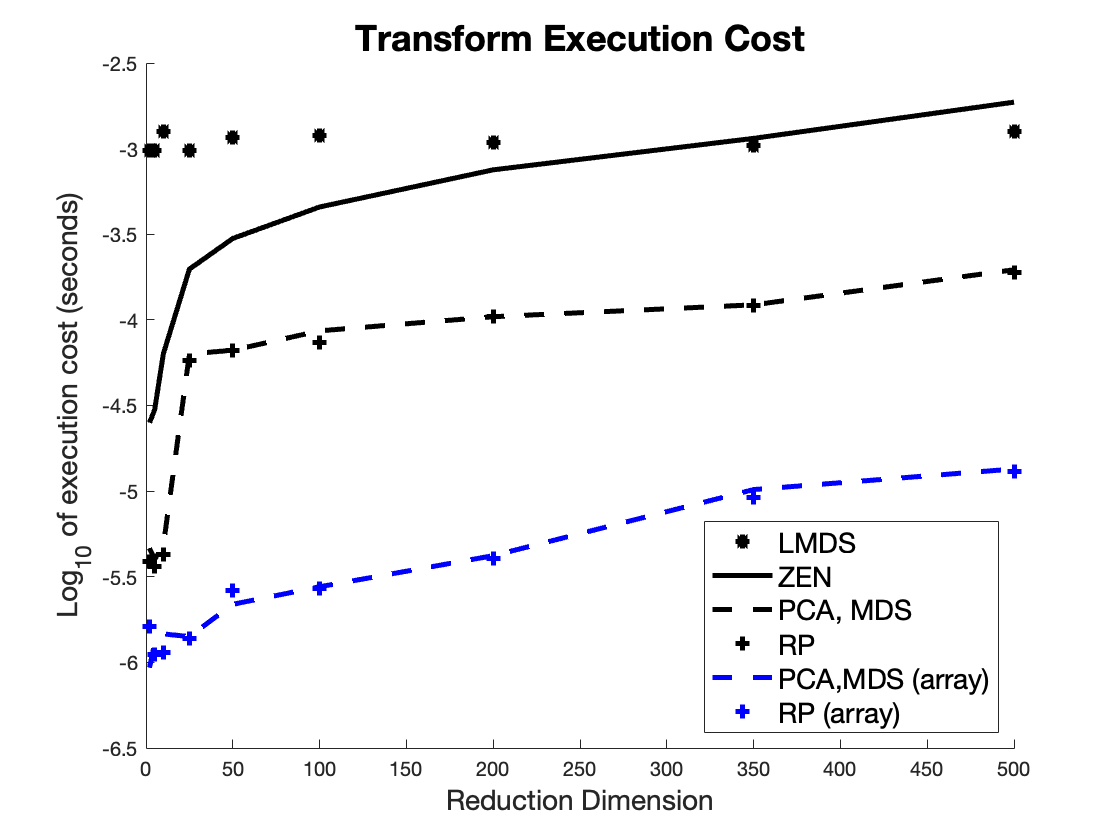} \hfill
\caption{Log-scaled costs of the creation and execution of the different DR transforms over a synthetic Euclidean space of 1000 dimensions.}
\label{fig_performance}
\end{figure}

As explained in \cite{kriegel2017black}, it is difficult or impossible to make fundamental comparisons on computational cost for a novel mechanism. In this case, we note that mechanisms such as PCA have had many years of study as to their optimisation, and specialist mathematical programming systems provide extremely fast versions, whereas the  LMDS and \nsimp \zen transforms reported here are (possibly naively) coded by ourselves following the high-level definitions. However measurements of the systems as used may be pragmatically useful, although these caveats should be taken into account.

As a further caveat, the experimental results reported here are performed using MatLab%
\footnote{R2022a update 4, 64-bit (Mac i64)}, and some of the tests performed in Java implementations give quite different outcomes. The MatLab system is highly optimised for array manipulation, and provides optimised implementations of PCA and MDS.

There are two key aspects to the performance of the mechanisms: the cost of producing the transform, and the cost of applying it to a data item, or a data set.

For PCA, MDS and RP, the transform function comprises a matrix which is multiplied by the data in order to produce the reduced-dimensional form. So in all cases, when an $m$-dimensional Euclidean space is to be reduced to $k$ dimensions, the transform takes the shape of an $m \times k$ matrix. The data to be reduced is an $n \times m$ matrix, comprising $n$ rows of $m-$dimensional data. The outcome of multiplying the data by the transform is then an $n \times k$ matrix representing the reduced-dimension data. 
Although the cost of producing the transform matrix varies widely with the technique, the application to the data is thus largely constant for a given $m$ and $k$. In theory, this cost with the RP mechanism we have used throughout could be much lower, as the transform matrix is deliberately constructed to contain many zero values and thus reduce the potential cost of matrix multiplication. We have not observed this in our tests, using standard matrix multiplication as provided by MatLab, but the potential exists.

For \nsimp \zen and LMDS, the execution of the transform in both cases depends  on distances measured between each element of the data and a set of reference objects. For LMDS this is a typically large fixed set of reference objects, whereas for \nsimp \zen the magnitude of the set of reference objects is the same as the reduction dimension.

Figure \ref{fig_performance} shows a set of experimental results for the creation and execution of the different reduction transforms. The context is a generated Euclidean space of 1000 dimensions, and each transform is used to reduce this to between 2 and 500 dimensions. The charts show the cost in each case of creating the transform, and the per-object cost of applying it to the data once the transform is created.

As can be seen, the cost of the transform creation varies widely, and this we believe to show fundamental differences in each approach. However in most cases the cost of transform creation is relatively unimportant compared to the cost of its execution.

The execution times show the cost of transformation per object. Two major effects are visible here; first, the mechanisms using matrix multiplication are approximately an order of magnitude faster than the mechanisms using object distances. In terms of the number of individual arithmetic operations performed, there is no such order of magnitude difference: in fact both \nsimp \zen and LMDS require  fewer individual arithmetic operations for these examples, and the difference seen is due to the optimisation of array multiplication.

Secondly, two execution times are shown for each of PCA, MDS and RP. The greater of these apply the matrix multiplication inside an iterative loop, to give a fair comparison with \nsimp \zen and LMDS where this is necessary. The faster outcomes, again approximately an order of magnitude better, apply the transform via a single matrix multiplication over the entire data set as would normally be possible.

In summary, the application of any of PCA, MDS or RP is around two orders of magnitude faster than \nsimp \zen or LMDS, when our naive implementations are compared against a professional matrix-optimised programming system. It is possible that \nsimp \zen could be optimised to reach an equivalent performance, but whether this is the case or not, and indeed how to achieve it, are open questions. Meantime, our provided version of \nsimp \zen can perform reduction at a typical rate of between $10^{-4}s$ and $10^{-3}s$ per object, which may be fast enough for many useful purposes.

\section{Discussion}
\label{sec_discussion}

In almost all quality measures, the \nsimp \zen transform outperforms all of the other well-known general techniques for Euclidean spaces.  As its application requires only the measurement of pairwise distances, rather than inspection of a coordinate space, it can also be applied to any metric space which is isometrically embeddable in a Hilbert space, where the necessary Euclidean properties also exist. In this context, it again outperforms LMDS in almost all measures.

There are two main reasons why the technique can perform better than other linear methods. First, it uses a well-known feature of high-dimensional spaces, namely the high probability of two sampled vectors being nearly orthogonal, to build a more accurate geometric model. Secondly, for a data set within a complex manifold, the use of a small number of sampled reference points, as opposed to  a much larger number required to produce a linear transform, seems to give a better reference model for the transform. It is noteworthy that, for example, the two-dimensional reductions shown in the example are derived from a set of 1000 reference points for the PCA transform, but only two randomly-selected reference points for the \nsimp \zen transform; it seems scarcely credible that the latter almost always give much better outcomes.

One drawback of the technique is that the range of the \nsimp \zen transform is not a Euclidean space, and therefore cannot be used for low-dimensional visualisations of data. However, small projected sets can be re-modelled using MDS to produce such a visualisation if desired. The \nsimp \zen function does possess the triangle inequality property, and the reduced space can therefore be used with metric indexing techniques.

\subsection{Very small distances}
\label{sec_subsec_very_small_distances}

The \nsimp \zen function has  been seen in some spaces to be less good at preserving ordering over very small distances than other techniques, in particular when applied to relatively low-dimensional spaces, or spaces which lie within a relatively uniform low-dimensional manifold. This is a significant drawback as it means the technique may not be the best for performing similarity search over a large reduced-dimension space. The reason for this is understood, and explained in Section \ref{sec_zen_function}. In pragmatic terms, there is an absolute lower-bound on any distance measured within the reduced space based on the altitude of the last derived component of the representative simplex.

As the simplest case, the \nsimp \zen distance between any object and itself, i.e. $d(u,u)$, projected into any dimension, is calculated as $\sqrt{2x_k^2}$, where $x_k$ is the value of the final-dimension coordinate in the transformed space.  This \nsimp \zen distance may well be greater than to another object $y$ where the other components are similar and $y_k$ is coincidentally smaller than $x_k$.

In high-dimensional spaces very small distances are very rare, and the problem is more evident in the lower-dimensional spaces we have tested. The probability of it occurring  however is currently beyond our full understanding; in some of our experiments it presents a problem, in others it does not. We would be reasonably optimistic that the effect could be at least partly overcome with further research. One factor that does make a difference is the choice of reference objects, which in all experiments we have reported has been random. Reference points which are mutually very close improves this particular outcome.

\subsection{Choice of Reference Objects}

As would be expected, the choice of reference objects used to construct the simplex has a significant effect on the quality of the reduction transforms.

In particular, it is possible for a pathological choice to result in the formation of a simplex which can be embedded in less than the required number of Euclidean dimensions. This does not lead to an incorrect situation,  but one where the vectors comprising the individual vertex points do not form a basis for the desired projection space, therefore  leading to a loss of information potential. In a  high-dimensional space the probability of this happening by chance is in fact vanishingly small, and it is easy to check during simplex construction at which point a different choice of reference object can be made. The problem is only likely to occur in practice if the space is contained in a manifold whose intrinsic dimensionality is close to  the dimensionality of the projection.

In all other cases, the quality of the transform can still be greatly affected by the choice of objects. We have spent some effort in seeking an optimal strategy, and have so far failed to improve, in general, on a random selection, other than when the projection is to very low dimensions. We  used the random strategy in all of the reported experimental results, and consider this point as further work. In outline, other than for a very small selection of reference objects, a random choice is highly likely to reflect the properties of the manifold in which the data is contained, and thus form a natural basis for the projection of the rest of that manifold.

\section{Conclusions}
\label{sec_conclusions}

We have presented a novel dimensionality reduction technique based on a geometric model of high-dimensional metric spaces. In an extensive range of tests, it outperforms any of the other well-known general techniques for Euclidean spaces. It gives particularly good relative performance when reductions from high to low dimensions are performed. Furthermore, it can be applied to a wide range of Hilbert spaces.

While there are still many unanswered questions as to its improvement, and the optimisation of its performance, we are convinced that the \nsimp \zen transform provides an exciting new tool to the dimensionality reduction toolbox.

\section*{Acknowledgements}
This work was partially funded by AI4Media - A European Excellence Centre for Media, Society
and Democracy (EC, H2020 n. 951911), SUN - Social and hUman ceNtered XR (EC, Horizon
Europe n. 101092612), and National Centre for HPC, Big Data and Quantum Computing (CUP
B93C22000620006).

%%%%%%%%%%%%%%%%%%%%%%%
% ---- Bibliography ----
%
\bibliographystyle{plain}
\bibliography{hilbert_dr.bib}

\begin{thebibliography}{10}

\bibitem{stack_overflow}
Large powers of sine appear gaussian — why?
\newblock \url{https://math.stackexchange.com/questions/2293330}, Accessed:
  2023-01-19.

\bibitem{achlioptas_rp}
Dimitris Achlioptas.
\newblock Database-friendly random projections.
\newblock In {\em Proceedings of the Twentieth ACM SIGMOD-SIGACT-SIGART
  Symposium on Principles of Database Systems}, PODS '01, page 274–281, New
  York, NY, USA, 2001. Association for Computing Machinery.

\bibitem{sqfd}
Christian Beecks, Merih~Seran Uysal, and Thomas Seidl.
\newblock Signature quadratic form distances for content-based similarity.
\newblock In {\em Proceedings of the 17th ACM International Conference on
  Multimedia}, MM '09, page 697–700, New York, NY, USA, 2009. Association for
  Computing Machinery.

\bibitem{blum_hopcroft_kannan_2020}
Avrim Blum, John Hopcroft, and Ravindran Kannan.
\newblock {\em Foundations of Data Science}.
\newblock Cambridge University Press, 2020.

\bibitem{chapter_blum_hopcroft_kannan_2020}
Avrim Blum, John Hopcroft, and Ravindran Kannan.
\newblock {\em High-Dimensional Space}, page 4–28.
\newblock Cambridge University Press, 2020.

\bibitem{blumenthal1953}
L.~M. Blumenthal.
\newblock {\em {Theory and applications of distance geometry}}.
\newblock Clarendon Press, 1953.

\bibitem{blumenthal1933note}
Leonard~M Blumenthal.
\newblock A note on the four-point property.
\newblock {\em Bulletin of the American Mathematical Society}, 39(6):423--426,
  1933.

\bibitem{7974879}
Michael~M. Bronstein, Joan Bruna, Yann LeCun, Arthur Szlam, and Pierre
  Vandergheynst.
\newblock Geometric deep learning: Going beyond euclidean data.
\newblock {\em IEEE Signal Processing Magazine}, 34(4):18--42, 2017.

\bibitem{DCG}
Chris Burges, Tal Shaked, Erin Renshaw, Ari Lazier, Matt Deeds, Nicole
  Hamilton, and Greg Hullender.
\newblock Learning to rank using gradient descent.
\newblock In {\em Proceedings of the 22nd International Conference on Machine
  Learning}, ICML '05, page 89–96, New York, NY, USA, 2005. Association for
  Computing Machinery.

\bibitem{cai2013distributions}
T~Tony Cai, Jianqing Fan, and Tiefeng Jiang.
\newblock Distributions of angles in random packing on spheres.
\newblock {\em Journal of Machine Learning Research}, 14:1837, 2013.

\bibitem{Chavez2001similarity}
Edgar Ch\'avez, Gonzalo Navarro, Ricardo Baeza-Yates, and Jos{\'e}~Luis
  Marroqu\'{\i}n.
\newblock Searching in metric spaces.
\newblock {\em ACM Comput. Surv.}, 33(3):273--321, September 2001.

\bibitem{Connor2016:HilbertExclusion}
R.~Connor, F.~A. Cardillo, L.~Vadicamo, and F.~Rabitti.
\newblock {Hilbert Exclusion}: Improved metric search through finite isometric
  embeddings.
\newblock {\em ACM Transactions on Information Systems (TOIS)},
  35(3):17:1--17:27, December 2016.

\bibitem{Connor2016:SISAP_Supermetric}
R.~Connor, L.~Vadicamo, F.~A. Cardillo, and F.~Rabitti.
\newblock Supermetric search with the four-point property.
\newblock In {\em Proceedings of 9th International Conference on Similarity
  Search and Applications (SISAP 2016)}, Lecture Notes in Computer Science,
  pages 51--64. Springer International Publishing, 2016.

\bibitem{Connor2016:IS_Supermetric}
R.~Connor, L.~Vadicamo, F.~A. Cardillo, and F.~Rabitti.
\newblock Supermetric search.
\newblock {\em Information Systems}, 80:108--123, 2019.

\bibitem{Connor2017:nsimplex}
R.~Connor, L.~Vadicamo, and F.~Rabitti.
\newblock High-dimensional simplexes for supermetric search.
\newblock In {\em Proceedings of 10th International Conference on Similarity
  Search and Applications ({SISAP} 2017)}, Lecture Notes in Computer Science,
  pages 96--109. Springer International Publishing, 2017.

\bibitem{four_metrics}
Richard Connor.
\newblock A tale of four metrics.
\newblock In Laurent Amsaleg, Michael~E. Houle, and Erich Schubert, editors,
  {\em Similarity Search and Applications}, pages 210--217, Cham, 2016.
  Springer International Publishing.

\bibitem{connor2016quantifying}
Richard Connor and Franco~Alberto Cardillo.
\newblock Quantifying the specificity of near-duplicate image classification
  functions.
\newblock In {\em VISAPP 2016}, 2016.

\bibitem{connor2019modelling}
Richard Connor, Al~Dearle, and Lucia Vadicamo.
\newblock Modelling string structure in vector spaces.
\newblock In {\em Proceedings of the 27th Italian Symposium on Advanced
  Database Systems (SEBD 2019)}. CEUR-WS.org, 2019.

\bibitem{connor2020sampled}
Richard Connor and Alan Dearle.
\newblock Sampled angles in high-dimensional spaces.
\newblock In {\em International Conference on Similarity Search and
  Applications}, pages 233--247. Springer, 2020.

\bibitem{connor2017:arxiv_nSimplex}
Richard Connor, Lucia Vadicamo, and Fausto Rabitti.
\newblock High-dimensional simplexes for supermetric search.
\newblock {\em CoRR}, abs/1707.08370, 2017.

\bibitem{cox2008multidimensional}
Michael~AA Cox and Trevor~F Cox.
\newblock Multidimensional scaling.
\newblock In {\em Handbook of data visualization}, pages 315--347. Springer,
  2008.

\bibitem{dasgupta2013experiments}
Sanjoy Dasgupta.
\newblock Experiments with random projection.
\newblock {\em arXiv preprint arXiv:1301.3849}, 2013.

\bibitem{decaf}
Jeff Donahue, Yangqing Jia, Oriol Vinyals, Judy Hoffman, Ning Zhang, Eric
  Tzeng, and Trevor Darrell.
\newblock Decaf: A deep convolutional activation feature for generic visual
  recognition.
\newblock In {\em International conference on machine learning}, pages
  647--655. PMLR, 2014.

\bibitem{Pearson1901}
Karl~Pearson F.R.S.
\newblock Liii. on lines and planes of closest fit to systems of points in
  space.
\newblock {\em The London, Edinburgh, and Dublin Philosophical Magazine and
  Journal of Science}, 2(11):559--572, 1901.

\bibitem{aqbc}
Yunchao Gong, Sanjiv Kumar, Vishal Verma, and Svetlana Lazebnik.
\newblock Angular quantization-based binary codes for fast similarity search.
\newblock {\em Advances in neural information processing systems}, 25, 2012.

\bibitem{GRACIA20141}
Antonio Gracia, Santiago González, Victor Robles, and Ernestina Menasalvas.
\newblock A methodology to compare dimensionality reduction algorithms in terms
  of loss of quality.
\newblock {\em Information Sciences}, 270:1--27, 2014.

\bibitem{hotelling1933analysis}
Harold Hotelling.
\newblock Analysis of a complex of statistical variables into principal
  components.
\newblock {\em Journal of educational psychology}, 24(6):417, 1933.

\bibitem{huiskes08}
Mark~J. Huiskes and Michael~S. Lew.
\newblock The mir flickr retrieval evaluation.
\newblock In {\em MIR '08: Proceedings of the 2008 ACM International Conference
  on Multimedia Information Retrieval}, New York, NY, USA, 2008. ACM.

\bibitem{ann_sift}
Herve Jégou, Matthijs Douze, and Cordelia Schmid.
\newblock Product quantization for nearest neighbor search.
\newblock {\em IEEE Transactions on Pattern Analysis and Machine Intelligence},
  33(1):117--128, 2011.

\bibitem{kriegel2017black}
Hans-Peter Kriegel, Erich Schubert, and Arthur Zimek.
\newblock The (black) art of runtime evaluation: Are we comparing algorithms or
  implementations?
\newblock {\em Knowledge and Information Systems}, 52:341--378, 2017.

\bibitem{alexnet}
Alex Krizhevsky, Ilya Sutskever, and Geoffrey~E. Hinton.
\newblock Imagenet classification with deep convolutional neural networks.
\newblock In {\em Proceedings of the 25th International Conference on Neural
  Information Processing Systems - Volume 1}, NIPS'12, page 1097–1105, Red
  Hook, NY, USA, 2012. Curran Associates Inc.

\bibitem{kruskal1964multidimensional}
Joseph~B Kruskal.
\newblock Multidimensional scaling by optimizing goodness of fit to a nonmetric
  hypothesis.
\newblock {\em Psychometrika}, 29(1):1--27, 1964.

\bibitem{sift}
D.G. Lowe.
\newblock Object recognition from local scale-invariant features.
\newblock In {\em Proceedings of the Seventh IEEE International Conference on
  Computer Vision}, volume~2, pages 1150--1157 vol.2, 1999.

\bibitem{mahalanobis1936generalized}
Prasanta~Chandra Mahalanobis.
\newblock On the generalized distance in statistics.
\newblock {\em Proceedings of the National Institute of Sciences (Calcutta)},
  2:49--55, 1936.

\bibitem{manning_raghavan_schtze_2008}
Christopher~D. Manning, Prabhakar Raghavan, and Hinrich Schütze.
\newblock {\em Introduction to Information Retrieval}.
\newblock Cambridge University Press, 2008.

\bibitem{matousek2013book}
Jiri Matousek.
\newblock {\em Lectures on discrete geometry}, volume 212.
\newblock Springer Science \& Business Media, 2013.

\bibitem{mcinnes2018umap}
Leland McInnes, John Healy, and James Melville.
\newblock Umap: Uniform manifold approximation and projection for dimension
  reduction.
\newblock {\em arXiv preprint arXiv:1802.03426}, 2018.

\bibitem{Menger1928}
K.~Menger.
\newblock Untersuchungen ber allgemeine metrik.
\newblock {\em Mathematische Annalen}, 100:75--163, 1928.

\bibitem{decaf_novak}
David Novak, Jan Cech, and Pavel Zezula.
\newblock Efficient image search with neural net features.
\newblock In {\em Proceedings of the 8th International Conference on Similarity
  Search and Applications - Volume 9371}, SISAP 2015, page 237–243, Berlin,
  Heidelberg, 2015. Springer-Verlag.

\bibitem{gist}
Aude Oliva and Antonio Torralba.
\newblock Modeling the shape of the scene: A holistic representation of the
  spatial envelope.
\newblock {\em International Journal of Computer Vision}, 42:145--175, 2004.

\bibitem{pearson1901:PCA}
Karl Pearson.
\newblock Liii. on lines and planes of closest fit to systems of points in
  space.
\newblock {\em The London, Edinburgh, and Dublin philosophical magazine and
  journal of science}, 2(11):559--572, 1901.

\bibitem{pennington2014glove}
Jeffrey Pennington, Richard Socher, and Christopher~D. Manning.
\newblock Glove: Global vectors for word representation.
\newblock In {\em Empirical Methods in Natural Language Processing (EMNLP)},
  pages 1532--1543, 2014.

\bibitem{rubenstein_hardness}
Aviad Rubinstein.
\newblock Hardness of approximate nearest neighbor search.
\newblock In {\em Proceedings of the 50th Annual ACM SIGACT Symposium on Theory
  of Computing}, STOC 2018, page 1260–1268, New York, NY, USA, 2018.
  Association for Computing Machinery.

\bibitem{LMDS}
Vin Silva and Joshua Tenenbaum.
\newblock Sparse multidimensional scaling using landmark points.
\newblock {\em Technology}, 01 2004.

\bibitem{Szegedy2015_CVPR}
Christian Szegedy, Wei Liu, Yangqing Jia, Pierre Sermanet, Scott Reed, Dragomir
  Anguelov, Dumitru Erhan, Vincent Vanhoucke, and Andrew Rabinovich.
\newblock Going deeper with convolutions.
\newblock In {\em Proceedings of the IEEE Conference on Computer Vision and
  Pattern Recognition (CVPR)}, June 2015.

\bibitem{tenenbaum2000global}
Joshua~B Tenenbaum, Vin~de Silva, and John~C Langford.
\newblock A global geometric framework for nonlinear dimensionality reduction.
\newblock {\em science}, 290(5500):2319--2323, 2000.

\bibitem{thomee2016yfcc100m}
B.~Thomee, B.~Elizalde, D.~.A Shamma, K.~Ni, G.~Friedland, D.~Poland, D.~Borth,
  and L-J Li.
\newblock {YFCC100M: The new data in multimedia research}.
\newblock {\em Communications of the ACM}, 59(2):64--73, 2016.

\bibitem{vadicamo2023induced}
Lucia Vadicamo, Giuseppe Amato, and Claudio Gennaro.
\newblock Induced permutations for approximate metric search.
\newblock {\em Information Systems}, 119:102286, 2023.

\bibitem{vadicamo2019splx}
Lucia Vadicamo, Richard Connor, Fabrizio Falchi, Claudio Gennaro, and Fausto
  Rabitti.
\newblock {SPLX-perm}: A novel permutation-based representation for approximate
  metric search.
\newblock In {\em Similarity Search and Applications: 12th International
  Conference, SISAP 2019, Newark, NJ, USA, October 2--4, 2019, Proceedings 12},
  pages 40--48. Springer, 2019.

\bibitem{Vadicamo2021_IS_Re-ranking}
Lucia Vadicamo, Claudio Gennaro, Fabrizio Falchi, Edgar Chávez, Richard
  Connor, and Giuseppe Amato.
\newblock Re-ranking via local embeddings: A use case with permutation-based
  indexing and the nsimplex projection.
\newblock {\em Information Systems}, 95:101506, 2021.

\bibitem{vadicamo2019metric}
Lucia Vadicamo, Vladimir Mic, Fabrizio Falchi, and Pavel Zezula.
\newblock Metric embedding into the hamming space with the n-simplex
  projection.
\newblock In {\em Similarity Search and Applications: 12th International
  Conference, SISAP 2019, Newark, NJ, USA, October 2--4, 2019, Proceedings 12},
  pages 265--272. Springer, 2019.

\bibitem{van2008visualizing}
Laurens Van~der Maaten and Geoffrey Hinton.
\newblock Visualizing data using t-sne.
\newblock {\em Journal of machine learning research}, 9(11), 2008.

\bibitem{voelker2017}
Aaron~R. Voelker, Jan Gosmann, and Terrence~C. Stewart.
\newblock Efficiently sampling vectors and coordinates from the n-sphere and
  n-ball.
\newblock Technical report, Centre for Theoretical Neuroscience, Waterloo, ON,
  01 2017.

\bibitem{wilson1932relation}
Wallace~A Wilson.
\newblock A relation between metric and euclidean spaces.
\newblock {\em American Journal of Mathematics}, 54(3):505--517, 1932.

\bibitem{zezula2006similarity}
Pavel Zezula, Giuseppe Amato, Vlastislav Dohnal, and Michal Batko.
\newblock {\em Similarity search: the metric space approach}, volume~32 of {\em
  Advances in Database Systems}.
\newblock Springer, 2006.

\end{thebibliography}

%%%%%%%%%%%%%%%%%%%%%%%%%%%%%%%%%%%%%%%
\appendix
\section*{Appendices}

\section{Hilbert-embeddable Distance Metrics}
\label{appendix_hilbert_metrics}

Distance metrics are usually referred to by name, but these names are often subject to details of context and may mean subtly different things to different readers. The following gives unambiguous definitions of metrics to which we refer  in the text, all of which are  isometrically embeddable in Hilbert space. In all cases we refer to a domain of vectors $\bb v,\bb w \in \mathbb{R}^n$ indexed as $v_i,w_i,\,\, 1 \le i \le n$.
\subsection{Euclidean Distance}

\begin{equation}
%D_\textit{euc}(\bb v,\bb w) = \sqrt{  \sum_{i=1}^n {(v_i - w_i)^2}}
\ell_2(\bb v,\bb w) = \sqrt{  \sum_{i=1}^n {(v_i - w_i)^2}}
\end{equation}

\subsection{Cosine Distance}
This term is particularly problematic; in some contexts it refers to simply the complement of the cosine of the angle between vectors, which is not a proper metric; in some it refers to the angle between vectors, which is a proper metric, and in some cases it means the Euclidean distance between the $\ell_2$-normalised vectors, which is a proper, Hilbert-embeddable metric. Note that all three forms give the same rank ordering. We use the last form:
\begin{equation}
D_\textit{cos}(\bb v,\bb w) = \sqrt{  \sum_{i=1}^n{\left(\frac{v_i}{\|\bb v\|} - \frac{w_i}{\|\bb  w\|}\right)^2}}
\end{equation}
which is generally efficient to evaluate as it is equivalent to Euclidean distance over $\ell_2$-normalised data. Note however that the general properties of such spaces are generally very different to Euclidean spaces due to this tight constraint over the data distribution.

\subsection{Jensen-Shannon Distance}

 This distance is applicable only to $\ell_1$-normalised positive vectors, as it derives from a metric over probability distributions. It is defined as
 \begin{equation}
\label{eq_jsd}
D_\textit{jsd}(\bb v,\bb w) = \sqrt{\mathcal{K}(\bb v,\bb w)}
\end{equation}
where
\begin{equation}
\mathcal{K}(\bb v,\bb w) = 1 - \tfrac{1}{2}\sum_{i=1}^n (h(v_i) + h(w_i) - h(v_i + w_i) )
\end{equation}
\begin{equation}
h(x) = -x \log_2 x
\end{equation}

In sparse spaces the term $0\log0$ may occur; this is taken as $0$, rather than undefined. This is a reasonable interpretation as this is the limit of the term $e \log e$ as $e$ tends to $0$ from above.

\subsection{Triangular Distance}

 This distance is applicable only to $\ell_1$-normalised positive vectors. Its main value is as a (much cheaper and very accurate in high dimensions) estimator for Jensen-Shannon distance \cite{four_metrics}.
 
 \begin{equation}
D_\textit{tri}(\bb v,\bb w) = \sqrt{ \tfrac{1}{2} \sum_{i=1}^n {\frac{(v_i - w_i)^2}{v_i + w_i}}}
\end{equation}

In sparse spaces the term $0 / 0$ may occur; this is taken as $0$, rather than undefined.

\subsection{Quadratic Form Distance}

The Quadratic Form Distance associated to a symmetric semi-definite positive matrix $M\in \mathbb{R}^{n\times n}$ is defined as

\begin{equation}
D_\textit{M}(\bb v,\bb w)= \sqrt{(\bb v-\bb w)^T M (\bb v-\bb w)}
\end{equation}

When the matrix $M$ is diagonal the corresponding distance is a weighted Euclidean distance.

Notable examples include Mahalanobis distance \cite{mahalanobis1936generalized}, and the Signature Quadratic Form distance \cite{sqfd}.

\section{Simplex Construction}\label{appendix_simplex_construction}

\begin{algorithm}[tbp]   
	{\KwIn{$n+1$ reference points $r_1,\dots,r_{n+1}\in (\mathcal{U},d)$}
		\KwOut{$n$-dimensional simplex in $\ell_2^n$ represented by the matrix $\Sigma\in \R{(n+1)\times n}$  }
		\BlankLine
		$\Sigma=  0\in \R{(n+1)\times n}$\;
		\If{$n=1$}{
			$\delta=d(r_1,r_2)$\;
			$\Sigma=\begin{bmatrix}
			0\\
			\delta
			\end{bmatrix}
			$\;
			\Return$\Sigma$;
		}
		$\Sigma_{Base}=$ nSimplexBuild($r_1,\dots,r_{n}$)\;
		$Distances=  \bb 0\in \R{n}$\;
		% \For{$i=1$ \KwTo $n$}{
		%   $Distances[i]=d(r_i,r_{n+1})$\;
		%   }
		\textbf{for} $1\leq i\leq n$ set $Distances[i]=d(r_i,r_{n+1})$\;
		$newApex=$ ApexAddition($\Sigma_{Base},Distances$)\;
		\textbf{for} $1\leq i\leq n$ and  $1\leq j\leq i-1$ set $\Sigma[i][j]$ to $\Sigma_{Base}[i][j]$\;
		\textbf{for} $1\leq j\leq n$ set $\Sigma[n+1][j]$ to $newApex[j]$\;
		\Return $\Sigma$;
		\caption{nSimplexBuild}\label{alg:nsimplex}}
\end{algorithm}
\begin{algorithm}[tbp]
	\KwIn{A $(n-1)$-dimensional base simplex and the distances between a new (unknown) apex point and the vertices of the base simplex:
		\begin{align*}
		%\footnotesize
		&\Sigma_{\text{Base}} = 
		\begin{bmatrix}
		0			&				&			&			\\
		v_{2,1}	&	0			&			\multicolumn{2}{c}{\text{\huge0}}		\\
		v_{3,1}	& v_{3,2}		&	\ddots		&				 \\
		\colon	& 				&\ddots		&  0			 	\\
		% v_{n-1,1}	&				&\cdots		& v_{n-1,n-1}	& 0				\\
		v_{n,1}	&				&\cdots		& v_{n,n-1}
		\end{bmatrix}\in \R{n\times n-1}%, 
		  \\
		  &\\
		& Distances =
		\begin{bmatrix}
		\delta_1 & \cdots & \delta_n
		\end{bmatrix}\in \R{n}
		\end{align*}
	}
	\KwOut{The cartesian coordinates of the new apex point
	}
	\BlankLine
	$Output =
	\begin{bmatrix}
	\delta_1& 0 &\cdots & 0 
	\end{bmatrix} \in \R{n}
	$\; 
	\For{$i=2$ \KwTo $n$}{
		$l = \ell_2(\Sigma_{Base}[i],Output)$\;
		$\delta = Distances[i]$\;
		$x = \Sigma_{Base}[i][i - 1]$\;
		$y = Output[ i - 1]$\;
		$Output[i-1] = y - (\delta^2 - l^2)/2x$\;
		$Output[i] = +\sqrt{y^2 - (Output[i-1])^2}$\;
	}
	\Return $Output$
	\caption{ApexAddition}\label{alg:apex}
\end{algorithm}
This section gives an inductive
algorithm (Algorithm \ref{alg:nsimplex}) to construct a simplex in $n$ dimensions based only on the distances measured among $n+1$ points.
%%, whenever the $n+1$ points are isometrically embeddable into $n$-dimensional euclidean space. 
%The inductive nature of this algorithm means that it can work with a ``reusable" simplex face, which suits cases where many $n$-dimensional simplexes are required to be constructed over a common base represented by a single $(n-1)$-dimensional simplex. The resulting simplex, represented by the matrix $\Sigma=[v_{i,j}]_{i,j}$, is such that $v_{i,j}=0$ for all $j\geq i$ and $v_{i,i-1}\geq 0$ for all $i=1,\dots,n+1$.

%Building a simplex of the above form, from only the edge lengths, can be achieved inductively.
For the base case of a one-dimensional simplex (i.e. two points with a single distance $\delta$) the construction is simply
\begin{equation}
\Sigma=\begin{bmatrix}
0\\
\delta
\end{bmatrix}
\end{equation}

%\[
%\begin{bmatrix}
%0		\\
%\delta
%\end{bmatrix}
%\]
 For an $n$-dimensional simplex, where $n \ge 2$, the distances among $n+1$ points are given. In this case, an $(n-1)$-dimensional simplex is first constructed using the first $n$ points. This simplex is used as a simplex base to which a new apex, the ${(n+1)}^{th}$ point, is added by the following \emph{ApexAddition} algorithm (Algorithm \ref{alg:apex}).
 
  For an arbitrary set of objects $s_i \in \mathcal{U}$, the apex $\sigma(s_i)$ can be pre-calculated. 
 When a query is performed,  only $n$ distances in the metric space require to be calculated to discover the new apex $\sigma(q)$  in $\ell_2^n$.  
 
 In essence, the \emph{ApexAddition} algorithm is derived from exactly the same intuition as the lower-bound property explained earlier. Proofs of correctness for both the construction and the lower-bound property are included %as an Appendix 
 hereafter for the interested reader.

%%%%% end of appendix about simplex construction
\section{Proof of correctess, \emph{ApexAddition} and \nsimp \emph{lwb}}\label{appendix:ProofCorrectness}

\begin{lemma}[Correctness of the ApexAddition algorithm]
	Let $ \Sigma_{\text{Base}}\in \R{n\times n-1}$ representing a $(n-1)$-dimensional simplex of vertices $\Sigma_{\text{Base}}[i]\in \ell_2^{n-1}$, with $\Sigma_{\text{Base}}[i][j]=0$ for all $j\geq i$ and $\Sigma_{\text{Base}}[n][n-1]\geq0$. Let $\bb v_i$ the corresponding vertices in $\ell_2^n$ (obtained from $\Sigma_{\text{Base}}[i]$ by adding a zero to the end of the vector) and let $\delta_i$ the distance between an unknown apex point and the vertex $\bb v_i$.
	Let $\bb o=\begin{bmatrix} o_1& \dots & o_n\end{bmatrix}$ the output  of the \emph{ApexAddition} Algorithm. Then $\bb o$ is a feasible apex, i.e. it is a point in $\R{n}$ satisfying $\ell_2(\bb o,\bb v_i)=\delta_i$ for all $1\leq i\leq n$. The last component $o_n$ is non-negative and represents the \emph{altitude} of $\bb o$ with respect to a base face  $\Sigma_{\text{Base}}$.
\end{lemma} 
\begin{proof}
	It is sufficient to prove that the output $\bb o=\begin{bmatrix} o_1& \dots & o_n\end{bmatrix}$ of the Algorithm \ref{alg:apex} has distance $\delta_i$ from the vertex $\bb v_i$, i.e. satisfies the following equations
	\begin{equation}\label{eq:apexSystem}
%	\footnotesize
	\begin{cases}
	o_1^2+\dots +  o_n^2=\delta_{1}^2 & \qquad(\ref{eq:apexSystem}.1)\\
	\qquad \colon\\
	\sum_{j=1}^{i-1} (v_{i,j}- o_j)^2+ \sum_{j=i}^n  o_j^2=\delta_i^2 & \qquad(\ref{eq:apexSystem}.i)\\
	\qquad \colon\\
	\sum_{j=1}^{n-1} (v_{n,j}-o_j)^2+  o_n^2=\delta_n^2 &\qquad (\ref{eq:apexSystem}.n)\\
	\end{cases}
	\end{equation}
	
	Note that the $i$-th component of the output $\bb o$ is updated only at the iteration $i$ and $i+1$ of the \emph{ApexAddition} Algorithm. 
	So, if we denote with $\bb o^{(i)}$ the output at the end of iteration $i$ we have:
	\begin{align}
%	\footnotesize
	& \bb o^{(1)}=\begin{bmatrix} \delta_1& 0&\dots & 0\end{bmatrix} \label{eq:o1}\\
	%& o^{(i)}_h=0 &  1\leq i < h\leq n\\
	&  o_i=  o^{(h)}_i, \quad o_n=o^{(n)}_n,\quad o^{(i)}_h=0&  1\leq i < h\leq n   \label{eq:p0} \\
	%%SOME PASSAGEs WERE CUTTED
%	& o_{i-1}^{(i)}=o_{i-1}^{(i-1)}-\frac{\delta_{i}^2-\ell_2(v_{i},o^{(i-1)})}{2v_{i,i-1}} &  2\leq i \leq  n \label{eq:p2}\\
%	&( o_{i}^{(i)})^2= ( o_{i-1}^{(i-1)})^2-( o_{i-1}^{(i)})^2&  1\leq i \leq  n-1 \label{eq:p3}
%	\end{align}
%	and thus
%	\begin{align}
%	\footnotesize
	&o_{i-1}=o_{i-1}^{(i-1)}-\frac{\delta_{i}^2-\sum_{j=1}^{i-2}(v_{i,j}-o_{j})^2-(v_{i,i-1}-o_{i-1}^{(i-1)})^2}{2v_{i,i-1}} &  2\leq i \leq  n \label{eq:p2b}\\
	&( o_{i-1})^2=( o_{i-1}^{(i-1)})^2- (o_{i}^{(i)})^2 &  1\leq i \leq  n-1 \label{eq:p3b}%\\
	%& \sum_{j=i}^{n} o_j^2= o_n^2+ \sum_{j=i}^{n-1}( ( o_{j}^{(j)})^2-( o_{j+1}^{(j+1)})^2)= (o_i^{(i)})^2&  1\leq i \leq  n-2 \label{eq:sum}
	\end{align}
	By combining  Eq. \eqref{eq:p0} and \eqref{eq:p3b} we obtain $\sum_{j=i}^{n} o_j^2= (o_i^{(i)})^2$ for all $1\leq i \leq  n-2$,
%	\begin{equation}
%	\footnotesize
%	\sum_{j=i}^{n} o_j^2%= o_n^2+ \sum_{j=i}^{n-1}( ( o_{j}^{(j)})^2-( o_{j+1}^{(j+1)})^2)
%	= (o_i^{(i)})^2 \qquad   1\leq i \leq  n-2, \label{eq:sum}
%	\end{equation}
	and so Eq. (\ref{eq:apexSystem}.1) clearly holds (case $i=1$). Moreover, %from Eq. \eqref{eq:sum} and Eq. \eqref{eq:p2b}, 
	it follows that $o$ satisfies Eq. (\ref{eq:apexSystem}.$i$) for all $i=2,\dots,n$:
	\begin{align*}
%	\footnotesize
	\sum_{j=1}^{i-1} (v_{i ,j}-o_j)^2+ \sum_{j=i }^n o_j^2& %\stackrel{\eqref{eq:sum}}{=}
	=v_{i ,i-1}^2 -2v_{i ,i-1}\,o_{i-1}+\sum_{j=1}^{i-2} (v_{i ,i-1}-o_j)^2+(o_{i-1}^{(i-1)})^2 \stackrel{\eqref{eq:p2b}}{=} \delta_{i }^2
	\end{align*}
	\qed
\end{proof}
%\section{}
%Consider a metric space $(U,d)$ which is $(n+1)$-embeddable in $\ell_2^n$, for any $n$. 
%Then it is possible to pick any $n$ reference points $r_1,\dots,r_n$ and form a simplex $\sigma_{n}$ in $\ell_2^{n-1}$. For any two further points $q$ and $s$, two further simplexes based on $\sigma_{n}$ can be independently calculated in $\ell_2^{n}$ space, using finite mappings $ q \rightarrow {q^{(n)}}$ and $s \rightarrow s^{(n)}$, where $q^{(n)}, {s^{(n)}} \in \ell_2^{n}$  are computed using the \textit{ApexAddition} Algorithm.  Now we have that $\ell_2^n(s^{(m)},q^{(m)})\leq d(q,s) \leq g(s^{(m)},q^{(m)})$ for an appropriate function $g$.
\begin{lemma}[n-Simplex Distance Constraint]
	Let $(\mathcal{U},d)$ a space $(n+2)$-embeddable in $\ell_2^{n+1}$. Let $r_1,\dots,r_n \in \mathcal{\mathcal{U}}$ and, for any $m\leq n$, let $\sigma_{m}$ the $(m-1)$-dimensional simplex generated from $r_1,\dots,r_{m}$ by using the \emph{nSimplexBuild} Algorithm.
	For any $x\in \mathcal{U}$, let $\bb{x}^{(m)}\in \ell_2^{m}$ the apex point with distance $d(x,r_1), \dots,$ $ d(x,r_m)$ from the vertices of $\sigma_{m}$, computed using the \emph{ApexAddition} Algorithm. 
	%Let $g:\ell_2^{m}\to \ell_2^{m}$ defined as $g(x,y)=\sqrt{\sum_{i=1}^{m-1} (x_i-y_i)^2 +(x_m+y_m)^2}.$ 
	Then for all $q,s\in \mathcal{U}$, 
	\begin{enumerate}
		\item 
		$\ell_2^{m-1} (\bb s^{(m-1)},\bb q^{(m-1)}) \leq \ell_2^{m}(\bb s^{(m)},\bb q^{(m)}) \quad\quad\quad \quad \mbox{for}\quad 2\leq m \leq n \label{eq:simplexlw}$ \label{cond1}
		\medskip
		\item 
		$g (\bb s^{(m-1)},\bb q^{(m-1)}) \geq g(\bb s^{(m)},\bb q^{(m)})  \,\quad\quad\quad\quad\quad\quad \mbox{for}\quad 2\leq m \leq n \label{eq:simplexub}$ \label{condub}
		\medskip
		\item 
		$\ell_2^n(\bb s^{(n)},\bb q^{(n)}) \leq d(s,q) \leq g(\bb s^{(n)},\bb q^{(n)})$
	\end{enumerate}
	where, for any $k\in \mathbb{N}$, $g:\ell_2^{k}\to \ell_2^{k}$ is defined as $g(\bb x,\bb y)=\sqrt{\sum_{i=1}^{k-1} (x_i-y_i)^2+(x_k+y_k)^2}$.
\end{lemma} 
\begin{proof}
	%For any $m\leq n$, given the simplex spanned by $r_1,\dots,r_m$ and represented by the matrix $\Sigma_{m}\in \R{m\times m-1}$ %of vertices $v_1,\dots,v_m \in \R{m-1}$ 
	%and given a generic point $x\in \mathcal{U}$, we denote with $x^{(m)}$ the apex point $f_{\sigma_{m}}(x)$ obtained using the \textit{ApexAddition} (Algorithm \ref{alg:apex}). 
	%The apex $x^{(m)}\in \R{m}$ is such that $\ell_2^{m}(x^{(m)}, [v_i\,\,0])=d(x,r_i)$; moreover, by construction
	By construction, for any $m\leq n$ we have
	\begin{align}
%	\footnotesize
	& x_i^{(m)}= x_i^{(m-1)} & i=1,\dots, m-2 \label{eq:c1}\\ 
	& x_i^{(i)}\geq 0 & i=1,\dots, m \label{eq:c3}\\
	&(x_{m-1}^{(m)})^2+( x_{m}^{(m)})^2 ={(x_{m-1}^{(m-1)})^2} \label{eq:c2}
	\end{align}
	Condition \ref{cond1} directly follows from Eq. \eqref{eq:c1}-\eqref{eq:c2}:
	%	\begin{align*}
	%	\footnotesize
	%	\ell_2^{m} (s^{(m)}&,  q^{(m)})^2= \sum_{i=1}^{m-2}(s^{(m)}_i-q^{(m)}_i)^2+ \sum_{i=m-1}^{m}(s^{(m)}_i-q^{(m)}_i)^2\\
	%	&\stackrel{\eqref{eq:c1}}{=}\ell_2^{m-1} (s^{(m-1)},q^{(m-1)})^2- (s^{(m-1)}_{m-1}-q^{(m-1)}_{m-1})^2+ \sum_{i=m-1}^{m}(s^{(m)}_i-q^{(m)}_i)^2\\
	%	&\stackrel{\eqref{eq:c2}}{=}\ell_2^{m-1} (s^{(m-1)},q^{(m-1)})^2 + 2\Big[-s^{(m)}_{m-1}q^{(m)}_{m-1}-s^{(m)}_{m}q^{(m)}_{m}\\
	%	&\qquad \qquad\qquad\qquad\qquad\qquad+\sqrt{(s_{m-1}^{(m)})^2+(s_{m}^{(m)})^2 } \sqrt{(q_{m-1}^{(m)})^2+(q_{m}^{(m)})^2 }\Big]\\
	%	&\geq \ell_2^{m-1} (s^{(m-1)},q^{(m-1)})^2
	%	%-2  -2+2\sqrt{(s_{m}^{(m)})^2+ (s_{m}^{(m)})^2} \sqrt{(q_{m}^{(m)})^2\\
	%	%+ (q_{m}^{(m)})^2}\right)
	%	\end{align*}
	\begin{align*}
%	\footnotesize
	\ell_2^{m} (\bb s^{(m)},  \bb q^{(m)})^2& 
	%=\sum_{i=1}^{m-2}(s^{(m)}_i-q^{(m)}_i)^2+ \sum_{i=m-1}^{m}(s^{(m)}_i-q^{(m)}_i)^2\\
	%&
	=\ell_2^{m-1} (\bb s^{(m-1)},\bb q^{(m-1)})^2- (s^{(m-1)}_{m-1}-q^{(m-1)}_{m-1})^2+ \sum_{i=m-1}^{m}(s^{(m)}_i-q^{(m)}_i)^2\\
	&=\ell_2^{m-1} (\bb s^{(m-1)},\bb q^{(m-1)})^2 + 2\Big[-s^{(m)}_{m-1}q^{(m)}_{m-1}-s^{(m)}_{m}q^{(m)}_{m}\\
	&\qquad \qquad\qquad\qquad\qquad\quad+\sqrt{(s_{m-1}^{(m)})^2+(s_{m}^{(m)})^2 } \sqrt{(q_{m-1}^{(m)})^2+(q_{m}^{(m)})^2 }\Big]\\
	&\geq \ell_2^{m-1} (\bb s^{(m-1)},\bb q^{(m-1)})^2
	%-2  -2+2\sqrt{(s_{m}^{(m)})^2+ (s_{m}^{(m)})^2} \sqrt{(q_{m}^{(m)})^2\\
	%+ (q_{m}^{(m)})^2}\right)
	\end{align*}
	where the last passage follows from the Cauchy–Schwarz inequality%
\footnote{Cauchy–Schwarz inequality in two dimension is: $(a_1b_1+a_2b_2)^2\leq (a_1^2+a_2^2)(b_1^2+b_2^2)$ $\forall a_1,b_1,a_2,b_2 \in \mathbb{R}$, which implies $$(a_1b_1+a_2b_2)\leq \sqrt{(a_1^2+a_2^2)}\sqrt{(b_1^2+b_2^2)} \quad \forall a_1,b_1,a_2,b_2 \in \mathbb{R}$$}.
	
	Similarly, Condition \ref{condub} also holds: 
	\begin{align*}
%	\footnotesize
	g (\bb s^{(m)}, \bb q^{(m)})^2&=g(\bb s^{(m-1)},\bb q^{(m-1)})^2 + 2\Big[-s^{(m)}_{m-1}q^{(m)}_{m-1}+s^{(m)}_{m}q^{(m)}_{m}\\
	&\qquad \qquad\qquad\qquad\qquad\quad-\sqrt{(s_{m-1}^{(m)})^2+(s_{m}^{(m)})^2 } \sqrt{(q_{m-1}^{(m)})^2+(q_{m}^{(m)})^2 }\Big]\\
	&\leq g(\bb s^{(m-1)},\bb q^{(m-1)})^2.
	\end{align*}

	Now we prove that $\ell_2^n(\bb s^{(n)},\bb q^{(n)})$ and $g(\bb s^{(n)},\bb q^{(n)})$ are, respectively, a lower bound and an upper bound for the actual distance $d(s,q)$.
	The main idea is using the simplex  $\sigma_{n}$ spanned by $r_1,\dots, r_n$ as a base face to  build the simplex  $\sigma_{n+1}$ spanned by $r_1,\dots, r_n, s$ and then use the latter as base face to build the simplex $\sigma_{n+2}$ spanned by $r_1,\dots, r_n, s,q$. In this way, we have an isometric embedding of $r_1,\dots, r_n, s,q $ into $\ell_2^{n+1}$ that is the function that maps  $r_1,\dots, r_n, s,q$  into the vertices  of $\sigma_{n+2}$. % Then we ``rotate" the vertex $q^{(n+1)}$ relative to $q$ around the face $\Sigma_{n}$, by moving $q^{(n+1)}$ so that it coincides with the hyperplane through the other $n+1$ vertices. The rotation is achieved by reducing the coordinates $q^{(n+1)}_{n+1}$ to zero, while adjusting all other $q^{(n+1)}_{i}$ coordinates to preserve distances between the rotated $q^{(n+1)}$ and the vertices of $\Sigma_{n}$. This correspond to consider the point $q^{(n)}$, which will allow us to prove that $\ell_2^n(s^{(n)},q^{(n)})\leq d(s,q)$.
	So, given the base simplex $\sigma_{n}$ (represented by the matrix $\Sigma_{n}$), and the apex $\bb s^{(n)}, \bb q^{(n)}\in \ell_2^n$ we have that the simplex  $\sigma_{n+2}$  is represented by
	\begin{equation}
%	\footnotesize
	\Sigma_{n+2} = 
	\left[\begin{array}{ccc|cc}
	\multicolumn{3}{c|}{\multirow{4}{*}{\large{$\Sigma_{n}$}}}& \multicolumn{2}{c}{\multirow{4}{*}{\large{$0$}}}\\
	\multicolumn{3}{c|}{}\\
	\multicolumn{3}{c|}{}   \\
	\hline
	s^{(n)}_1		&	\cdots		&  s^{(n)}_{n-1}&   s^{(n)}_n & 0	\\
	q^{(n)}_1	&	\cdots		& q^{(n)}_{n-1}&	q^{(n+1)}_{n}	   	& q^{(n+1)}_{n+1}	
	\end{array}\right]\in \R{n+2\times n+1} 
	\end{equation}
	where, by construction, $(q^{(n+1)}_{n+1})^2={(q^{(n)}_{n})^2-(q^{(n+1)}_{n})^2}$, $s^{(n)}_n,q^{(n+1)}_{n+1} \geq 0$, and $d(q,s)$ equals the Euclidean distance between the two last rows of $\Sigma_{n+2}$.
	
	It follows that
	\begin{align} \label{eq:truedist}
	%	\footnotesize
	d(q,s)^2%&=\sum_{i=1}^{n-1}(s^{(n)}_i-q^{(n)}_i)^2+(s^{(n)}_n-q^{(n+1)}_n)^2+(q^{(n+1)}_{n+1})^2\\
	&= \sum_{i=1}^{n-1}(s^{(n)}_i-q^{(n)}_i)^2+(s^{(n)}_n)^2+(q^{(n)}_n)^2-2s^{(n)}_nq^{(n+1)}_n; 
	\end{align}
	and, since $q^{(n)}_n\geq |q^{(n+1)}_n|$, we have
	$$
	%	\footnotesize
	d(q,s)^2=\ell_2^n(\bb s^{(n)},\bb q^{(n)})^2 +2s^{(n)}_n (q^{(n)}_n-q^{(n+1)}_n) \geq \ell_2^n(\bb s^{(n)},\bb q^{(n)})^2,
	$$
	and
	$$
	%	\footnotesize
	d(q,s)^2=g(\bb s^{(n)},\bb q^{(n)})^2 -2s^{(n)}_n (q^{(n)}_n+q^{(n+1)}_n)\leq g(\bb s^{(n)},\bb q^{(n)})^2
	$$
	
%	\hl{Lucia: I've added the  following part:}\\
	Finally, we observe that since $(q^{(n+1)}_{n+1})^2+(q^{(n+1)}_{n})^2={(q^{(n)}_{n})^2}$, $q^{(n+1)}_{n+1}\geq 0$, and $q^{(n)}_{n}\geq 0$, there exists an angle $\theta \in [0,\pi]$ such that
	\begin{equation}
	\begin{cases}
	   q^{(n+1)}_{n}&= q^{(n)}_{n} \cos \theta \\
	    q^{(n+1)}_{n+1} &=q^{(n)}_{n} \sin \theta 
	\end{cases}
	\end{equation}
	Therefore, Eq. \ref{eq:truedist} can be rewritten as
	\begin{equation*}
	  d(q,s)^2= \sum_{i=1}^{n-1}(s^{(n)}_i-q^{(n)}_i)^2+(s^{(n)}_n)^2+(q^{(n)}_n)^2-2s^{(n)}_nq^{(n)}_{n} \cos \theta;
	\end{equation*}
	In other words, if $\sigma: D\to \mathbb{R}^n$ is the \textit{nSimplex} transform defined by a set of $n$ reference points then for any $s, q \in D$ given the transformed points $\bb x=\sigma(s)$ an $\bb y=\sigma(q)$ it holds
	\begin{equation}
  d(q,s)= \sqrt{ \sum_{i=1}^{n-1}(x_i-y_i)^2+x_n^2+y_n^2-2x_ny_{n} \cos \theta}
	\end{equation}
%	\hl{we may include this result in Section 3.1.. }
\end{proof}

\section{Data sets used in experiments}
\label{appendix_data_sets}

While all the software used in experiments described is available from \url{https://github.com/richardconnor/dr-matlab-code}, the data sets are typically too large to provide conveniently  and we therefore provide brief descriptions of their provenance.

\begin{description}
\item[Generated uniform data] 
All generated data is created using the MatLab \emph{rand} function from the Statistics and Machine Learning toolbox. For example a set of one thousand objects of one hundred dimensions is created by the single line
\begin{verbatim}
data = rand(1000,100);
\end{verbatim}

\item[ Twitter GloVe] 
The GloVe data used derives from \url{https://nlp.stanford.edu/projects/glove_} where the data and instructions for downloading it can be found. We used the 200-dimensional vectors.

\item[ MF1M]
The images from which this data derives are available from \url{https://press.liacs.nl/mirflickr/mirdownload.html}. We use the one million image set. For input to AlexNet, whole images were reduced to 227 x 227 using ImageMagick.

The 4096-dimensional vectors were obtained by applying the MatLab release of AlexNet, which is also available from other domains in other languages. In MatLab, the fc6 layer used is simply extracted by code such as
\begin{verbatim}
 fc6 = activations(net,thisImage,"fc6","OutputAs","rows");
\end{verbatim}

\item[ ANN SIFT]
The ANN SIFT data, along with code to extract it, is available from \url{http://corpus-texmex.irisa.fr/}.

\item[ GIST]
Again the Mir Flickr one million image collection was used to produce the GIST data.
Although not fully documented, GIST representations of the images are also available from \url{https://press.liacs.nl/mirflickr/mirflickr1m.v3b/}; alternatively MatLab code to create GIST descriptors is available at \url{https://people.csail.mit.edu/torralba/code/spatialenvelope/}

\end{description}

\section{Measuring the quality of dimensionality reduction}
\label{sec_intro_subsec_quality}

\subsection{Global Structure}
\label{sec_global_stress}
\begin{description}

\item[Shepard Diagrams] give a  visual overview of the quality of a transform. For a given set of data all distances $\delta_{ij}$ are plotted against the reduced dimensional distances $\zeta_{ij}$. In most general terms, the closer the plot lies to the $y=x$ diagonal, the better the reduction.

\item[Kruskal Stress]
To quantify the visual effect, Shepard diagrams are usually overlaid with an isotonic regression function calculated from the original and reduced spaces, as used to calculate Kruskal's stress function.
The stress function  is given by:
\begin{equation}
S_K = \sqrt{\frac{\sum_{i<j}(\zeta_{ij} - d^*_{ij})^2}{\sum_{i<j} \zeta_{ij}^2}}   
\end{equation}
where $d^*$ is the value given by a function implied from the least-squares isotonic regression calculated from the finite data presented in the chart.

The point of fitting an isotonic regression is that the stress function assigns a value according to the monotonicity of the reduction transform, rather than the absolute values produced. PCA, for example, is a reduction mapping: in all cases, $\zeta_{ij} \le \delta_{ij}$. However if the function $\T : \delta_{ij} \rightarrow \zeta_{ij}$ is perfectly monotonic, the stress will be zero despite the  reduction in individual values, even if this is non-linear.
%%
%Figure \ref{fig_shepard_egs} shows some examples of Shepard plots overlaid with the isotonic regression function.
%
%\begin{figure}[th]
%\includegraphics[width=0.32\textwidth]{figures/pca}
%\hfill
%\includegraphics[width=0.32\textwidth]{figures/mds}
%\hfill
%\includegraphics[width=0.32\textwidth]{figures/mds_scaled}
%\caption{Example Shepard plots overlaid with the isotonic regression function of Kruskal's stress, the value of which is also displayed in the title. In this example it can be seen that MDS gives a lower stress than PCA, and that applying an artificial scaling to the plot does not affect the Kruskal stress. In this case, PCA  scores better then MDS according to the distance preservation qualities given in Section \ref{sec_intro_quality_subsect_dist_pres}: whether PCA or MDS is the better transform depends on the context of use. }
%\label{fig_shepard_egs}
%\end{figure}

\end{description}

\subsection{Distance Preservation}
\label{sec_intro_quality_subsect_dist_pres}

\begin{description}
\item[Sammon Stress]  derives from  Sammon Mapping, a non-linear dimensionality reduction technique similar to MDS which minimises the stress function:
\begin{equation}
S_S = \frac{1}{\sum_{i<j}{\delta_{ij}}} \sum_{i<j}{\frac{(\delta_{ij}-\zeta_{ij})^2}{\delta_{ij}}}    
\end{equation}

Unlike Kruskal's stress function, Sammon stress is affected by the absolute differences between $\zeta_{ij}$ and $\delta_{ij}$ rather than their isotonic relationship, and so gives a further useful perspective on the quality of a reduction transform.

\item[Quadratic Loss] is a purely distance-based technique, the quadratic function used to punish the production of outliers:
\begin{equation}
S_Q = \sum_{i<j}{(\delta_{ij}-\zeta_{ij})^2}
\end{equation}

In fact we view this function as somewhat of a blunt instrument, as it punishes even regular deviations from original distances. For example, PCA is a contraction function, while RP is not; RP often performs better for quadratic loss even when the quality of the reduction is clearly, overall, lower. Even if the absolute deviation from original distance is important, it may often be possible to apply a scaling function to the reduced space in order to minimise this.
%\item [average relative error] $\frac{1}{|T|}\sum_{(x,y)\in T} \frac{|d(x,y)-s\tilde{d}(x,y)|}{d(x,y)}$
%\item [mean squared error]  $\frac{1}{|T|}\sum_{(x,y)\in T} {(d(x,y)-s\tilde{d}(x,y))^2}$

As discussed in Appendix \ref{appendix_quality_profiles}, we require all quality measures to be bounded in $[0,1]$ to allow visual comparisons over different reduction dimensions. There is no natural upper-bound on this measure, so we simply convert an outcome of $q$ to $\frac{q_{\text{max}} - q}{q_{\text{max}}}$ to produce the desired range, where $q_{\text{max}}$ is the greatest value obtained in the context of the visualisation.
\end{description}

\subsection{Topology Preservation}
\begin{description}
\item[Spearman Rho] 
 measures the preservation of rank ordering among pairwise distances measured between corresponding objects in the domain and range of the reduction transform. This is a useful measure for many applications, such as nearest-neighbour analysis, where the absolute distances among values are of no interest other than for the ordering which they induce over other elements of the set.

Pairwise distances from a sample set of $n$ objects are used to construct an ordering $z$ of size $T = {n \choose 2}$ of pairs $\delta_{ij}$. A ranking $z'$ is then created according to the relative distances of the same pairs of objects after the transform is applied. The Spearman Rho function is then given as
\begin{equation}
S_R = 1 - \frac
{6 \sum_{i=1}^T (z(i) - \hat{z}(i))^2}
{T^3 - T}    
\end{equation}
where the adjusting factors combine to give an output in the range $[-1,1]$ where 1 implies a perfect preservation of distance ordering and $-1$ implies the  inverse correlation.

There are other forms of this formula, but we choose this one to faithfully follow the exact methodology of \cite{GRACIA20141}. Although all of our other quality measures are normalised into $[0,1]$, the fact that the outcome of zero implies an effectively random ordering suits our purpose in this respect.
%\hl{Lucia: in the formula, $T=n$, right? 
%Moreover, Spearman rho distance is defined as $\sqrt{\sum_{i=1}^n (z(i) - \hat{z}(i))^2}$ and its not-squared root variant has as maximum value  equal to $(n^3-n)/3$, thus 
%$ \dfrac{3 \sum_{i=1}^T (z(i) - \hat{z}(i))^2}
%{n^3 - n}$ ranges between 0 (perfect preservation) and 1 (inverse correlation)---
%why you prefer to translate it in [-1,1]?
%}
\item[kNN Query Recall] When the purpose of dimension reduction is to speed up similarity search, the most important outcome is the nearest-neighbour topology. To an extent this is tested by both Kruskal stress and Spearman Rho quality measures, but with the crucial difference that in the context of query recall it is only the smallest distances, relative to a query, that are relevant. Thus, a transform which  preserves very small distances well, but is less good  over larger distances, will be preferable to one which preserves   distances overall, although the latter may score better in these quality measures.

To measure this quality, it is necessary to construct a nearest-neighbour ground truth over a representative sample of queries for a large data set, and then compare the nearest neighbours of those queries in the reduced-dimension space. 
One problem  is that the pattern of nearest neighbours depends on specific details of the space being considered, as well as more general properties of the reduction. The larger the space, the smaller the nearest-neighbour distances will be, and the higher the probability of having very close matches which are not representative of the general space.
% Therefore a transform that performs very well when searching a space of $10^6$ elements may behave quite differently when searching a space of $10^{12}$ elements.

%{\color{blue}new from here}
To overcome some of these issues we measure recall using the following assumptions:
\begin{enumerate}
\item only a small percentage of the true nearest neighbours are of any significance
\item the nearer true neighbours are considerably more significant than the further neighbours
\item preservation of order in the results is also important
\end{enumerate}
Our recall measurement therefore uses a discounted cumulative  gain (DCG) function over a relevance function based on nearest-neighbour rank.

The ranking function is constructed to give significantly higher importance to the closer neighbours by using the logistic function to give an inverse sigmoid function over rank. In our experiments we have collected 1,000 nearest neighbours
from a collection of one million data, and rank the relevance of each true nearest neighbour as
\begin{equation}
R_i = 1 - \frac{1}{\left(1 + e^{-\frac{i- 500}{100}}\right)}
\end{equation}
for the $i$th true nearest neighbour.

We then compare the 1,000 nearest neighbours returned by the DR function using the DCG variant defined in \cite{DCG}:
\begin{equation}
DCG_{DR} = \sum_i^{1000}\frac{2^{DR_{i}} - 1}{\log_2 i + 1}
\end{equation}
where $DR_i$ is $R_i$ applied to the position in the true nearest neighbours of the object found in the $i$th position of the nearest neighbours according to the DR transform.

Finally, this function produces an arbitrary maximum value of $66.0435$ when  lists of length 1,000  are in perfect correlation, and so the outcome is divided by this factor to give a normalised value in the range $[0,1]$, where 0 means there is no overlap between the lists and 1 means they are in perfect correlation.
\end{description}

\subsection{Quality profiles}
\label{appendix_quality_profiles}

As discussed in \cite{GRACIA20141}, it is instructive to consider the quality of transforms as a profile  over different reduction dimensions. This may be shown as a plot where one or more of the numeric quality functions is plotted against the dimension of the reduction, typically as this is reduced from the original dimensionality of the original domain down to 2. For most mechanisms and useful quality measures, this will  result in a monotonic decreasing plot, and will allow the selection of the most useful compromise in terms of quality loss for a given reduction dimension.

To allow presentation of all quality measures within the same bounds, Kruskal and Sammon and  stress measurements are subtracted from 1 to give a quality rather than a stress measure, and the results in the range $[0,1]$ are given. A negative value can arise from either Spearman Rho (which is bounded in $[-1,1]$ or Sammon  stress (which has no formal upper bound) but in reality a value of less than 0 for Spearman Rho, or greater than 1 for Sammon stress, effectively means that the transform has no practical value and a zero quality rating is reasonable. For quadratic loss there are no natural bounds, and this is handled as explained in Appendix \ref{sec_intro_quality_subsect_dist_pres}.

\end{document}